\documentclass[conference,onecolumn,11pt]{IEEEtran}

\usepackage[margin=2cm,bottom=2.5cm,footskip=1cm]{geometry}

\usepackage{IEEEpreamble}
\usepackage[style=numeric-comp,maxbibnames=99,giveninits=true]{biblatex}
\renewbibmacro*{urldate}{
    (visited on \printfield{urlday}/\printfield{urlmonth}/\printfield{urlyear})
}
\usepackage{amsmath,amssymb,amsfonts}
\usepackage{algorithmic}
\usepackage{graphicx}
\usepackage{textcomp}
\usepackage{xcolor}
\usepackage{csquotes}

\begin{document}

\title{Composition in Differential Privacy\\ for General Granularity Notions\\ 
{\large\textit{Long version}}}

\makeatletter
\newcommand{\linebreakand}{%
  \end{@IEEEauthorhalign}
  \hfill\mbox{}\par
  \mbox{}\hfill\begin{@IEEEauthorhalign}
}
\makeatother

\author{
    \IEEEauthorblockN{Patricia Guerra-Balboa*}
    \IEEEauthorblockA{\textit{Karlsruhe Institute of Technology, Germany}\\
    patricia.balboa@kit.edu}
    
\and
    \IEEEauthorblockN{Àlex Miranda-Pascual*}
    \IEEEauthorblockA{\textit{Karlsruhe Institute of Technology, Germany}\\
    \textit{Universitat Politècnica de Catalunya, Spain}\\
    alex.pascual@kit.edu}
    
\linebreakand
    \IEEEauthorblockN{Javier Parra-Arnau}
    \IEEEauthorblockA{\textit{Universitat Politècnica de Catalunya, Spain}\\
    javier.parra@upc.edu}
    
\and
    \IEEEauthorblockN{Thorsten Strufe}
    \IEEEauthorblockA{\textit{Karlsruhe Institute of Technology, Germany}\\
    strufe@kit.edu}
}

\maketitle

\thispagestyle{plain}
\pagestyle{plain}

\nnfootnote{This is the long version of the paper accepted at the 37th IEEE Computer Security Foundations Symposium (2024). 


*These authors contributed equally.}

\begin{abstract} 
    The composition theorems of differential privacy (DP) allow data curators to combine different algorithms to obtain a new algorithm that continues to satisfy DP. 
    However, new granularity notions (i.e., neighborhood definitions), data domains, and composition settings have appeared in the literature that the classical composition theorems do not cover. For instance, the original parallel composition theorem does not translate well to general granularity notions. This complicates the opportunity of composing DP mechanisms in new settings and obtaining accurate estimates of the incurred privacy loss after composition. 
    
    To overcome these limitations, we study the composability of DP 
    in a general framework and for any kind of data domain or neighborhood definition. We give a general composition theorem in both independent and adaptive versions and we provide analogous composition results for approximate, zero-concentrated, and Gaussian DP\@. Besides, we study the hypothesis needed to obtain the best composition bounds.
    Our theorems cover both parallel and sequential composition settings. Importantly, they also cover every setting in between, allowing us to compute the final privacy loss of a composition with greatly improved accuracy.
\end{abstract}

\section{Introduction}\label{sec:intro}
Differential privacy ($\varepsilon$-DP)~\cite{dwork2006Differential} is a well-known privacy notion in the field of data protection. 
One advantage of DP over other privacy notions, such as, for instance, syntactic notions~\cite{soria-comas2016Big}, is that DP possesses the key property of \textit{composability}: It is possible to form a new DP mechanism by composing a finite number of given DP mechanisms. The DP composition theorems serve as a reliable measure for any privacy loss suffered in the newly composed DP mechanism. For these reasons, the advantages of DP composition are recognized throughout the privacy community. For example, composability is key for the construction of most DP algorithms; further, the privacy protection of adaptive updates (e.g., in a streaming scenario or model learning) could not be computed without composition. 

Currently, DP composition is represented by two results: \textit{sequential composition}~\cite{dwork2014algorithmic} and \textit{parallel composition}~\cite{mcsherry2009privacy}. 
Parallel composition is applied when all combined mechanisms access mutually disjoint databases, the maximum loss before combination determines the total privacy loss after composition. 
Sequential composition covers any case when arbitrary DP mechanisms with access to the entire data are combined. 
The total privacy loss in sequential composition is computed as the sum of the losses of each composed mechanism.

DP and the sequential and parallel composition theorems were originally defined for tabular databases in the \textit{unbounded}~\cite{kifer2011No} scenario. Nowadays, however, the literature works both with different \textit{database domains} (i.e., classes of the input databases of a privacy mechanism) and with different \textit{neighborhood definitions} (also called \textit{granularity notions}~\cite{dwork2014algorithmic}), such as \textit{bounded DP}~\cite{kifer2011No} or \textit{edge-DP}~\cite{hay2009Accurate}. 
Consequently, the mechanisms we compose can be defined for different domains and granularities. There also can be alternatives to accessing either the whole database or disjoint parts of it. Therefore, we need new composition rules for more general settings.

However, the existing composition theorems may not extend directly to these general settings. 
For instance, \citeauthor{li2016Differential}~\cite{li2016Differential} show that the proof of the parallel composition theorem~\cite{mcsherry2009privacy} does not hold if we change the original granularity to bounded DP.
Since composition for new domains and new granularity notions may be non-trivial or even impossible, curators need to understand how composition results work for each case and when they yield no significant results. Otherwise, curators risk misapplying DP composition, for example, by using parallel composition in a bounded scenario.

To provide a context where all granularities can be composed and where the final privacy loss can be systematically interpreted and compared with the initial ones, we set up a general mathematical framework based on the notion of $d$-privacy introduced by \citeauthor{chatzikokolakis2013Broadening}~\cite{chatzikokolakis2013Broadening}. 
Using this framework we present composition theorems~(\ref{th:ICTheoremVariableDomain} and~\ref{th:ACTheoremVariableDomain}) for when a mechanism is applied independently of the others (the \textit{independent} scenario) or using the output of a mechanism as input in the following ones (the \textit{adaptive} scenario). 
Our results allow us to obtain new composition theorems  for any domain and granularity notion, both existing and future, and even allow combining different domains and granularity notions. Consequently, we improve the understanding of how different granularity notions affect composition in DP. 
Furthermore, our results facilitate a more accurate calculation of the privacy loss upon any possible composition of DP mechanisms and showcase the effect that preprocessing has on the computation. 
For instance, if the mechanisms take as input non-necessarily disjoint subsets of the initial database, it is now possible to obtain better bounds than the sum obtained using sequential composition (see \Cref{ex:IntermediateSetting}). 

Besides, we study the settings that are common in the literature and provide the corresponding privacy estimates obtained by using our composition theorems.
Furthermore, we study sufficient conditions to obtain the ``$\max{\varepsilon_i}$'' bound when the mechanisms take as input disjoint parts of the initial database. For the cases where this bound cannot be achieved, we provide a new variation on composition (see \Cref{sec:CommonDomainIndependent}) that allows us to achieve better results. In particular, we provide a solution to the open problem of \citeauthor{li2016Differential}~\cite{li2016Differential} by giving the lowest possible privacy loss for the composition of bounded DP mechanisms executed on mutually disjoint databases (\Cref{th:BoundedParallel}).

To further showcase our results, we extend our composition theorems to other privacy notions based on DP where the granularity can be changed. These other privacy notions are \textit{approximate DP} ($(\varepsilon,\delta)$-DP), \textit{zero-concentrated DP} ($\rho$-zCDP) and \textit{Gaussian DP} ($\mu$-GDP). To the best of our knowledge, we are the first to define the $d$-private counterparts of $(\varepsilon,\delta)$-DP, $\rho$-zCDP, and $\mu$-GDP  in order to gain a more general perspective on these three notions.
Besides, we provide the first statement of the zCDP composition over disjoint databases. Moreover, we provide a tighter bound than $\max_{i\in[k]}\mu_i d$ for Gaussian DP over disjoint databases (see \Cref{ex:ultra-parallel}). 

\begin{figure}[tb!]
    \centering\small
    \begin{tikzpicture}[every text node part/.style={align=center}]
        
        \newlength{\myheight}
        \setlength{\myheight}{7.75cm}

        \newlength{\mywidth}
        \setlength{\mywidth}{0.564\columnwidth}
        
        \newlength{\boxposition}
        \setlength{\boxposition}{0.171\columnwidth}
        
        \newlength{\boxwidth}
        \setlength{\boxwidth}{0.21\columnwidth}
        
        \node[rectangle,
        fill = kit-green50,
        opacity = 0.7,
        minimum width = \mywidth, 
        minimum height = 0.4\myheight
        ] at (0,0.2\myheight) {};
        
        \node[rectangle,
        fill = kit-blue50,
        opacity = 0.7,
        minimum width = \mywidth, 
        minimum height = 0.4\myheight
        ] at (0,-0.2\myheight) {};
        
        \node[rectangle,
        draw,
        thick,
        color = kit-blue70!70,
        minimum width = \mywidth, 
        minimum height = 0.4\myheight
        ] at (0,-0.2\myheight) {};
        
        \node[rectangle,
        draw,
        thick,
        color = kit-green70!70,
        minimum width = \mywidth, 
        minimum height = 0.4\myheight
        ] at (0,0.2\myheight) {};
        
        \node[] at (-0.53\mywidth,0.2\myheight) (Independent) {\rotatebox{90}{\mybox[fill=kit-green70!50]{Independent}}};
        
        \node[] at (-0.53\mywidth,-0.2\myheight) (Adaptive) {\rotatebox{90}{\mybox[fill=kit-blue70!50]{Adaptive}}};
        
        \draw[line width=0.6pt, double distance=1.75pt, -{Classical TikZ Rightarrow[length=3.5pt]},to path={|- (\tikztotarget)}]
        (Adaptive)--(Independent);
        
        \node[] at (0,0.247\myheight) (CT) {Independent Composition CD (\ref{th:ICTheoremCommonDomain})}; 
        
        \node[] at (0,-0.247\myheight) (ACT) {Adaptive Composition CD (\ref{th:ACTheoremCommonDomain})}; 
        
        \node[] at (0,0.35\myheight) (GCT) {Independent Composition (\ref{th:ICTheoremVariableDomain})}; 
        
        \node[] at (0,-0.35\myheight) (GACT) {Adaptive Composition (\ref{th:ACTheoremVariableDomain})}; 
        
        \node[minimum height = 0] at (0.47\mywidth,0) (center) {}; 
    
        \draw[line width=0.6pt, double=kit-blue50!70, double distance=1.75pt, -]
        (GACT)-|(center.center);
        
        \draw[line width=0.6pt, double=kit-green50!70, double distance=1.75pt, -{Classical TikZ Rightarrow[length=3.5pt]}]
        (center.center)|-(GCT);
        
        \node[rectangle,
        rounded corners=10pt,
        fill = kit-blue50,
        opacity = 0.8,
        minimum width = \boxwidth, 
        minimum height = 0.36\myheight
        ] at (-\boxposition,0) {};
        
        \node[rectangle,
        fill = kit-blue70,
        minimum width = \boxwidth,
        opacity = 0.8] at (-\boxposition,0) {\textbf{Sequential} \\ {\scriptsize(\ref{th:OriginalSequentialComposition} \cite{dwork2006Our})}};

        \node[] at (-\boxposition,0.138\myheight) {Independent};

        \node[] at (-\boxposition,0.09\myheight) {\scriptsize (O: \ref{th:ISCDwork} \cite{dwork2014algorithmic}, G: \ref{th:GeneralizedISC})}; 

        \node[] at (-\boxposition,-0.09\myheight-0.002\myheight) {Adaptive};
        
        \node[] at (-\boxposition,-0.138\myheight-0.002\myheight) {\scriptsize (O: \ref{th:ASCLi} \cite{li2016Differential}, G: \ref{th:GeneralizedASC})};
        
        \node[rectangle,
        rounded corners=10pt,
        fill = kit-green50,
        opacity = 0.8,
        minimum width = \boxwidth, 
        minimum height = 0.36\myheight
        ] at (\boxposition,0) {};
        
        \node[rectangle,
        fill = kit-green70,
        minimum width = \boxwidth,
        opacity = 0.8] at (\boxposition,0) {\textbf{Disjoint Inputs} \\{\scriptsize as in parallel (\ref{th:ParallelMcSherry} \cite{mcsherry2009privacy})}};

        \node[] at (\boxposition,0.138\myheight) {Independent};

        \node[] at (\boxposition,0.09\myheight) {\scriptsize BB: (\ref{th:GeneralizedIPCVariableDomain}, CD: \ref{th:GeneralizedIPCCommonDomain})}; 

        \node[] at (\boxposition,-0.09\myheight-0.002\myheight) {Adaptive};
        
        \node[] at (\boxposition,-0.138\myheight-0.002\myheight) {\scriptsize BB: (\ref{th:GeneralizedAPCVariableDomain}, CD: \ref{th:GeneralizedAPCCommonDomain})};

        \draw[line width=0.6pt, double=kit-green50!70, double distance=1.75pt, -{Classical TikZ Rightarrow[length=3.5pt]}]
        (GCT)--(CT);
        
        \node[] at (-0.125\mywidth,0.09\myheight) (IS) {};
        
        \node[] at (0.125\mywidth,0.09\myheight) (IP) {};
        
        \node[] at (0,0.094\myheight) (Icenter) {};
        
        \draw[line width=0.6pt, double=kit-green50!70, double distance=1.75pt, {Classical TikZ Rightarrow[length=3.5pt]}-{Classical TikZ Rightarrow[length=3.5pt]},to path={|- (\tikztotarget)}]
        (IS.center)--(IP.center);
        
        \draw[line width=0.6pt, double=kit-green50!70, double distance=1.75pt, -]
        (CT)--(Icenter.center);

        \draw[line width=0.6pt, double=kit-blue50!70, double distance=1.75pt, -{Classical TikZ Rightarrow[length=3.5pt]}]
        (GACT)--(ACT);
        
        \node[] at (-0.125\mywidth,-0.138\myheight-0.002\myheight) (AS) {};
        
        \node[] at (0.125\mywidth,-0.138\myheight-0.002\myheight) (AP) {};
        
        \node[] at (0,-0.1439\myheight) (Acenter) {}; 
        
        \draw[line width=0.6pt, double=kit-blue50!70, double distance=1.75pt, {Classical TikZ Rightarrow[length=3.5pt]}-{Classical TikZ Rightarrow[length=3.5pt]},to path={|- (\tikztotarget)}]
        (AS.center)--(AP.center);
        
        \draw[line width=0.6pt, double=kit-blue50!70, double distance=1.75pt, -]
        (ACT)--(Acenter.center);
    \end{tikzpicture}
    \caption{\small Overview of the theorems proved in this paper, classified according to whether they are adaptive or independent. The theorems represented are the generalizations of sequential composition and the best bound (BB) for disjoint inputs (as in the parallel setting). In the figure, ``O'' denotes the original theorem, ``G'' our generalized version, and ``CD'' common domain. Arrows indicate that a result directly implies the other.}
    \label{fig:CompositionTheorems}
\end{figure}
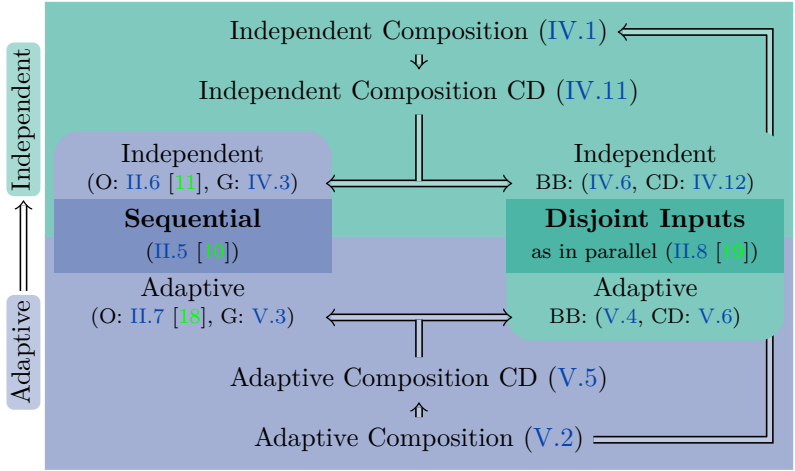

An overview of the generalized results is given in \Cref{fig:CompositionTheorems}. 
Our contributions are as follows:
\begin{itemize}
   \item We prove the independent composition (IC) and the adaptive composition (AC) theorems, two new results that allow for reducing the estimated privacy loss and designing improved DP mechanisms in general contexts. Moreover, our theorems make it possible to mix different granularity mechanisms while controlling the privacy guarantees offered. 
   \item We study particular cases of previous theorems that generalize the sequential and parallel composition to any granularity notion. This allows us to compute the minimum privacy loss for the bounded case when the mechanism processes disjoint parts of the database.
   \item We define $(d_{\D},\delta_{\D})$-privacy, $d^2_\D$-zCprivacy, and $d_\D$-Gprivacy, $d_\D$-private versions of $(\varepsilon,\delta)$-DP, $\rho$-zCDP, and $\mu$-GDP. Our definitions allow us to generalize to other domains and to provide general composition bounds. We also adapt our general composition results to $(d_{\D},\delta_{\D})$-privacy, $d^2_\D$-zCprivacy and $d_\D$-Gprivacy. Particularly, we show that the parallel composition metric bound can be improved in $d_\D$-Gprivacy.
\end{itemize}

The paper is organized as follows: Preliminaries are explored in \Cref{sec:preliminaries}, and we formalize the granularities and the generalization to $d_{\D}$-privacy in \Cref{sec:generalizingGranularityNotions}.
We present our independent composition theorem in \Cref{sec:IndependentComposition}, including interesting cases such as a generalization of the independent sequential composition and the setting where the mechanisms take as input disjoint parts of the database.
In \Cref{sec:GeneralizingAdaptive}, we discuss the analogous results for the adaptive scenario. Then we give the composition results for $(\varepsilon,\delta)$-DP, $\rho$-zCDP, and $\mu$-GDP in \Cref{sec:extending}. Finally, we discuss post-processing and the reciprocal theorems (\Cref{sec:reciprocal}) and conclude with a
brief summary of the results (\Cref{sec:conclusions}). 
All proofs of our statements can be found in Appendix~\ref{sec:proofs}.

\paragraph*{Related Work} \citeauthor{li2016Differential}~\cite{li2016Differential} analyze the composition theorems in unbounded and bounded DP, and find out that the parallel composition theorem does not necessarily hold for bounded DP mechanisms. However, they do not explore other granularities of the state of the art or attempt to provide a solution for the bounded problem.
\citeauthor{mcsherry2009privacy}~\cite{mcsherry2009privacy} gives the first distance-based formulation of DP, later generalized by \citeauthor{chatzikokolakis2013Broadening}~\cite{chatzikokolakis2013Broadening} with the definition of $d_{\D}$-privacy, which we use to set the general framework for composition. However, only sequential composition has been explored for $d_{\D}$-privacy~\cite{galli2023Group}.
Therefore, the generalization of other composition settings, such as parallel, to other granularities (metrics) is still an open question, and to the best of our knowledge, there is no work in the literature, either for DP or for $d$-privacy, that, in a general manner, computes an accurate privacy loss bound when we have other metrics, domains and composition rules.

\section{Preliminaries}\label{sec:preliminaries}
In this section, we introduce the main concepts relevant to this work. The main notation used throughout the manuscript is compiled in \Cref{tab:notation}.

\begin{table}[t]
    \centering
    \begin{tabular}{|c|c|}
        \hline
        Symbol & Meaning \\
        \hline\hline
        $\X$ & Set of possible data records \\
        $\DX$ & The universe of all databases drawn from $\X$ \\
        $\D$ & Database class \\
        $D,D'$ & A pair of databases \\
        $|D|$ & Size of $D$ (number of records) \\
        $x$ & Data record (element of $\X$)\\
        $m_{D}(x)$ & Multiplicity of $x$ in the multiset $D$\\
        $\M\colon\D\to\S$ & A randomized mechanism with domain $\D$\\
        $\S\coloneqq\Range(\M)$ & Set of possible outputs of all $\M(D)$ with $D\in\D$ \\
        $S$ & Measurable subset of $\S$ \\
        $s$ & Element of $\S$ \\
        $\G$ & Granularity notion/neighborhood definition \\
        $D\neigh_\G D'$ & $D$ and $D'$ are $\G$-neighboring \\
        $d_\D$ (or $d$) & Metric over $\D$ \\
        $d^\G_\D$ & Canonical metric of $\G$ over $\D$ \\
        $\U$, $\B$ & Unbounded and bounded granularity (resp.) \\
        $D\triangle D'$ & Symmetric difference ($(D\cup D')\backslash(D\cap D')$) \\
        $[k]$ & Set of indices $\{1,\dots,k\}$ \\
        $I_f(D,D')$ & For $f=\{f_i\}_{i\in[k]}$, $|\{i\in[k]\mid f_i(D)\neq f_i(D')\}|$ \\
        \hline
    \end{tabular}
    \caption{\small Summary of the notation used in this paper.}
    \label{tab:notation}
\end{table}

\subsection{Tabular Databases and Differential Privacy}

In the original formulation of DP, the database $D$ is assumed to be comprised of a finite number $n$ of rows, where the intuition is that each row contains data related to an individual, drawn from a universe of data records~$\X$~\cite{dwork2014algorithmic}. In this case, the data model is a tabular database, and we refer to a single data row as a \textit{record}.
We denote the universe of all the possible tabular databases drawn from $\X$ as $\DX$. In particular, $\DX$ contains the empty database $\varnothing$ and is closed under subsets (if $D'\subseteq D\in\DX$, then $D'\in\D_{\X}$) and under basic math operators: $D\cup D'$, $D\cap D', D\backslash D'\in\DX$ for all $D,D'\in\DX$. We consider all these operations as defined for multisets~\cite{syropoulos2001Mathematics} for the rest of the paper.

The first definition of $\varepsilon$-DP with precise formulation\footnote{For the literature definitions and theorems, we state them as they are defined in the cited reference but using the notation of this manuscript.} was introduced by \citeauthor{dwork2006Differential}~\cite{dwork2006Differential}.

\begin{definition}[Differential privacy~\cite{dwork2006Differential}]\label{def:firstDP}
    A randomized mechanism $\M$ with domain $\DX$ is \textit{$\varepsilon$-differentially private} ($\varepsilon$-DP) if for all $D,D'\in\DX$ differing on at most one element and all measurable $S\subseteq \Range(\M)$,
    \begin{equation}\label{eq:DPoriginal}
        \Prob\{\M(D)\in S\} \leq \e^{\varepsilon}\Prob\{\M(D')\in S\}.
    \end{equation}
\end{definition}

An important part of DP is the concept of \textit{neighborhood}, also referred to as the \textit{granularity notion} of DP~\cite{dwork2014algorithmic}. In \Cref{def:firstDP}, two databases $D,D'\in\DX$ are \textit{neighboring} if and only if they ``differ on at most one element'', i.e., $|D\triangle D'|=|(D\cup D')\backslash(D\cap D')|\leq1$. In other words, we obtain a neighboring database by removing or adding a single element or row. Assuming each row is linked to a single individual, we get the usual DP interpretation: DP aims to protect the participation of each individual in the original database up to $\varepsilon$.

Two parameters control the privacy of individuals in the DP definition, namely, the \textit{privacy budget} $\varepsilon$ and the universe of records $\mathcal{X}$. The former limits the amount of information that an attacker can extract with access to the mechanism's output. The latter encodes what information is considered public. For example, if $\mathcal{X}$ is the set of possible addresses of a city, we can discover (up to~$\varepsilon$) that a person lives in a particular city, while if $\mathcal{X}$ is the set of possible addresses of a country, we can discover (up to~$\varepsilon$) that an individual lives in the country, but not which exact city.

Furthermore, with the \textit{group privacy} property of DP~\cite{dwork2014algorithmic}, we also protect the participation of $n$ individuals with the protection degrading linearly with respect to $n$. More precisely, we have the following result:
\begin{proposition}[\cite{mcsherry2009privacy}]\label{prop:GroupPrivacy}
    A mechanism $\M$ is $\varepsilon$-DP if and only if for all $D,D'\in\D_{\X}$ and all measurable set $S\subseteq \Range(\M)$ 
    \begin{equation}\label{eq:McSherry}
        \Prob\{\M(D)\in S\} \leq \e^{\varepsilon
        |D\triangle D'|}\Prob\{\M(D')\in S\}.
    \end{equation}
\end{proposition}

In this case, $d^{\triangle}(D,D')\coloneqq|D\triangle D'|$ can be thought of as the distance (or metric) between $D$ and $D'$ in~$\DX$. In this regard,  \citeauthor{mcsherry2009privacy}~\cite{mcsherry2009privacy} provides the first statement of DP from a metric perspective, which stems from the group privacy property. 
This laid the foundations of the generalization to $d$-privacy~\cite{chatzikokolakis2013Broadening}, which we will explore in \Cref{sec:generalizingGranularityNotions}.

\subsection{Differential Privacy: Unbounded vs.\ Bounded}
Nowadays, many neighborhood definitions for DP exist.
A compilation of common granularities is provided in~\cite{desfontaines2020SoK}.
Among these, unbounded and bounded DP are the most popular ones~\cite{kifer2011No}. The unbounded notion corresponds to the original definition presented by \citeauthor{dwork2006Differential}~\cite{dwork2006Differential}. 

\begin{definition}[restate = DEunbounded, name = Unbounded]\label{unbounded}
    A pair of databases $D,D'\in\DX$ are \textit{unbounded neighboring} if $D$ can be obtained from $D'$ by either adding or removing one record (i.e., $|D\triangle D'|=1$).
\end{definition}

\begin{definition}[restate = DEbounded, name = Bounded]\label{bounded}
    A pair of databases $D,D'\in\DX$ are \textit{bounded neighboring} if $D$ can be obtained from $D'$ by changing the value of exactly one record (i.e., $|D\triangle D'|=2$ and $|D|=|D'|$).
\end{definition}

These two notions of neighborhood lead to different privacy guarantees. The clearest difference concerns the privacy of the number of records: the unbounded notion protects the number of records in the database, while the bounded notion does not.

\subsection{Introduction to the Composition Theorems}
One of the most useful properties of DP mechanisms relates to composition theorems. Sequential and parallel composition are considered key components of DP and are regularly used in the field. 

The composition theorems share a common foundation. Simply put, these theorems say that given $k$ $\varepsilon_i$-DP mechanisms $\M_i$, the composed mechanism $\M$ satisfies $\varepsilon$-DP, where $\varepsilon$ depends on $\varepsilon_1,\dots,\varepsilon_k$. In other words, these theorems estimate the privacy loss (i.e., the final privacy budget) of the mechanism $\M$ composed of $\M_i$.
However, there are different ways to compose a set of mechanisms, and thus different theorems. We distinguish the following:

\textbf{Independent vs.\ adaptive:} Composition is \textit{independent} if the outputs of each $\M_i$ are independent of each other. On the other hand, it is \textit{adaptive} if $\M_i$ can use the outputs of any $\M_j$ with $j<i$ as input. More intuitively, $\M$ computes the mechanisms in order (first $\M_1$, then $\M_2$, then $\M_3$, etc.) and can take the output of previous mechanisms as input.
Note that adaptive composition is more general than independent composition, i.e., the independent theorems are cases of adaptive results.

\textbf{Sequential vs.\ parallel:} Orthogonally,
if every $\M_i$ takes as input the whole database $D$ in its computation, the composition is \textit{sequential}. Alternatively, the composition is \textit{parallel} if each $\M_i$ uses only data from a subset $D_i\subseteq D$ that is not used by any other.

The combination of these variations leads to four clear cases (see \Cref{fig:CompositionTheorems}), which we will refer to as the independent/adaptive sequential/parallel composition settings, due to the lack of consensus\footnote{For example, ISC and IPC are referred to as independent and sequential composition in~\cite{kifer2020Guidelines}; and as sequential and adaptive composition in~\cite{desfontaines2020SoK}.}.
We will also refer to them by the corresponding acronyms: ISC, IPC, ASC, and APC. 
In the current literature, we frequently find ISC~\cite{dwork2014algorithmic}, ASC~\cite{li2016Differential}, and IPC~\cite{mcsherry2009privacy}; while APC remains heavily unused.

\subsection{The Classic Composition Theorems}

The sequential and parallel composition theorems were initially stated for the original DP definition~\cite{dwork2006Differential}, unbounded DP, before the introduction of any other granularity. Nevertheless, we specify it in the following theorems. 

The first composition result of DP appeared in~\cite{dwork2006Our}.

\begin{theorem}[Sequential composition~\cite{dwork2006Our}]\label{th:OriginalSequentialComposition}
    A mechanism that permits $T$ adaptive interactions with an [unbounded] $\varepsilon$-DP mechanism ensures [unbounded] $T\varepsilon$-DP\@.
\end{theorem}

The theorem corresponds to the adaptive definition and includes independent composition as a subcase. Nowadays, these results are sometimes formulated separately with precise hypotheses and allow for different privacy budgets.

\begin{theorem}[Independent sequential composition (ISC)~\cite{dwork2014algorithmic}]\label{th:ISCDwork}
    Let $\M_i\colon\DX\to\S_i$ be an [unbounded] $\varepsilon_i$-DP mechanism for each $i\in[k]$. Consider the mechanism $\M$ with domain $\D$ such that $\M(D)=(\M_1(D),\dots,\M_k(D))$ for all $D\in\D$. Then $\M$ is [unbounded] $(\sum_{i=1}^k \varepsilon_i)$-DP\@.
\end{theorem}

\begin{theorem}[Adaptive sequential composition (ASC)~\cite{li2016Differential}]\label{th:ASCLi}
    Let $\M_1,\dots,\M_k$ be $k$ mechanisms (that take auxiliary inputs) that satisfy [unbounded] $\varepsilon_1$-DP, \dots, $\varepsilon_k$-DP, respectively, with respect to the input database. Publishing $\textbf{t} = \scalar{t_1,t_2,\dots,t_k}$, where $t_1=\M_1(D)$, $t_2=\M_2(t_1, D)$, \dots, $t_k = \M_k(\scalar{t_1,\dots,t_{k-1}},D)$, satisfies [unbounded] $(\sum^k_{i=1}\varepsilon_i)$-DP.
\end{theorem}

In search of optimization, the literature has found circumstances for a better bound than the sequential one.  Databases can be composed of diverse information and most queries only need to compute values in a proper subset of data. It is these circumstances which, in fact, provide the better bound: The parallel composition theorem.

\begin{theorem}[Parallel composition~\cite{mcsherry2009privacy}]\label{th:ParallelMcSherry}
    Let $\M_i$ each provide [unbounded] $\varepsilon$-DP\@. Let $\X_i$ be arbitrary disjoint subsets of the universe of records $\X$. The sequence of $\M_i(D_i)$ provides [unbounded] $\varepsilon$-DP, where $D_i\subseteq D$ is the multiset such that element $x\in D$ has multiplicity $m_{D_i}(x)=\mathbf{1}_{\X_i}(x)\,m_{D}(x)$.
\end{theorem}

By abuse of notation, $D_i$ is also often denoted as $D\cap\X_i$. Nowadays, this formulation has also seen modifications. For example, \citeauthor{li2016Differential}~\cite{li2016Differential} use a partitioning function $p$ to define the disjoint subsets in the previous statement, i.e., $p_i(D)=D_i$ for all $i$ and $D\in\DX$. 

Even though \Cref{th:OriginalSequentialComposition,th:ASCLi,th:ISCDwork} were initially stated for the unbounded granularity notion, they can easily be translated for other granularities~\cite{galli2023Group}.
However, in \Cref{th:ParallelMcSherry}, if instead of unbounded, we impose $\M_i$ to be bounded $\varepsilon$-DP, then it is not generally true that the sequence of $\M_i(D_i)$ provides bounded $\varepsilon$-DP.  \citeauthor{li2016Differential}~\cite{li2016Differential} show why the proof is not applicable: even if $\M_i$ are bounded $\varepsilon$-DP, $\M'_i$ such that $\M'_i(D)=\M_i(D_i)=\M_i(D\cap\X_i)$ is not necessarily bounded $\varepsilon$-DP.
This fact is clear in the following counterexample, which we provide to complete \citeauthor{li2016Differential}'s claim~\cite{li2016Differential}:

\begin{example}[Parallel composition does not hold for bounded DP]\label{ex:BoundedParallelFail}
    Let $\D_{\X}$ be a database universe and $\X_i$ arbitrary disjoint subsets of $\X$. We show that given $k>1$ mutually independent bounded $\varepsilon_i$-DP mechanisms $\M_i\colon\DX\to\S$, it is not necessarily true that the composed mechanism $\M\colon\DX\to\S$ such that $\M(D)=(\M_1(D_1),\dots,\M_k(D_k))$ is bounded DP, where $D_i\subseteq D$ is the multiset such that element $x\in D$ has multiplicity $m_{D_i}(x)=\mathbf{1}_{\X_i}(x)\,m_{D}(x)$.
    
    To do so, we prove that we can select $k>1$ mutually independent bounded $\varepsilon_i$-DP mechanisms $\M_i\colon\DX\to\S$ such that mechanism $\M\colon\DX\to\S$ with $\M(D)=(\M_1(D_1),\dots,\M_k(D_k))$ is not bounded $\varepsilon$-DP for any $\varepsilon\geq 0$.
    
    For all $i\in[k]$, we choose $\M_i\colon\DX\to\S$ such that they output the number of elements of the input database, i.e., $\M_i(D)=\M^*(D)=|D|$ for all $D\in\DX$. It can easily be checked that this mechanism is bounded $0$-DP. Observe that in this case, $\M(D)=(\M^*(D_1),\dots,\M^*(D_k))=(|D_1|,\dots,|D_k|)$.
    
    Let $D,D'\in\D_{\X}$ be two bounded-neighboring databases such that $D\triangle D'=\{x,x'\}$ with $x\in D_j$ and $x'\in D'_l$, $j\neq l$. It is clear then that $\M^*(D_j)=|D_j|\neq|D'_j|=\M^*(D'_j)$ (analogously for $l$), so $\Prob\{\M^*(D_j) = |D_j|\}=1\not\leq 0=\Prob\{\M^*(D'_j) = |D_j|\}$. Note that this is not a contradiction with $\M^*$ being unbounded DP, since $D_j$ and $D'_j$ are not bounded-neighboring databases. 
    
    Consequently, taking $s=(|D_1|,\dots,|D_n|)\in\S$ we obtain $\Prob\{\M(D) = s\}=1$, but $\Prob\{\M^*(D'_j)=|D_j|\} = 0$. Therefore $\Prob\{\M(D)=s\}=1\not\leq 0=\e^\varepsilon\Prob\{\M(D')=s\}$ for all $\varepsilon\geq0$, so the mechanism $\M$ is not bounded~DP\@.
\end{example}

We want to showcase with this example that the composition results proved for unbounded DP in $\DX$ cannot be trivially generalized to other data domains or neighborhood definitions. The failure of bounded DP on satisfying $(\max_{i\in[k]}\varepsilon_i)$-DP when composed in parallel opens a new question about how to measure the privacy of composed mechanisms in general.

The main goal of this work is to answer this question by generalizing these composition theorems to more general scenarios, in which the domain of the mechanism is not necessary $\DX$, and the given granularity notion is not necessarily unbounded. To achieve so, we introduce new more-general composition rules that even allow composing DP mechanisms with different domains and granularity notions. These results are shown in \Cref{sec:IndependentComposition,sec:GeneralizingAdaptive}.

However, to carry out this extension of properties to general settings, we first need to define a formal structural model. Thus, we will begin by generalizing the data domain and the concept of granularity notions in the next section.

\section{Generalizing the Granularity Notion of DP}\label{sec:generalizingGranularityNotions}
As mentioned earlier, DP was designed to handle aggregated queries on tabular data. However, in many cases, mechanisms impose a maximum or minimum number of elements in the database, are only defined for databases of a fixed size, or are not defined for the empty database, which are incompatible conditions if $\DX$ is the mechanism domain. Also, the structure of the data is not necessarily tabular, such as graph databases~\cite{hay2009Accurate}. For instance, in a social network graph, each node is an individual in the database, while the edges represent the social relationship between the nodes. This means that information about individuals is not always encoded in rows or multiset elements. 

This motivates the need to generalize DP to different settings. In this section, we provide a mathematical formalization of the granularity notions and the data domain, establishing a framework in which privacy can always be measured and compared between different notions.

Databases are collections of data and can be defined as mathematical objects such as multisets (original case), sets, numbers, functions, streams, or graphs. A collection of databases forms a \textit{database class}\footnote{We define it as a mathematical class instead of a set because a collection of sets does not need to be a well-defined set~\cite{levy2002Basic}. We denote the usual inclusion of classes by $\D'\subseteq\D$.}, which we denote by $\D$.
In our setting, we will consider the cases where the domain of $\M$ is a generic class $\D$ instead of~$\DX$ (including the case where $\D\subseteq\DX$).

Moreover, DP allows many different neighborhood definitions~\cite{desfontaines2020SoK}, each with its own privacy implications and interpretability. We generalize the definition of granularity notion $\G$ as follows.

\begin{definition}[$\G$-neighborhood]\label{def:G-neighborhood}
    Given a database class $\D$, we define the \textit{$\G$-neighborhood relation} as a binary symmetric relation $\neigh_\G$ between elements in $\D$.
    We say that $D,D'\in\D$ are \textit{$\G$-neighboring} if $D\neigh_\G D'$.
\end{definition}

We will use calligraphic letters to denote certain granularity notions (e.g., $\U$ for unbounded, $\B$ for bounded). With \Cref{def:G-neighborhood}, we can establish a general framework for DP similar to that in~\cite{kifer2011No}. That is, a mechanism $\M$ with domain $\D$ is \textit{$\G$ $\varepsilon$-DP} ($\varepsilon\geq0$) if for all $\G$-neighboring $D, D'\in\D$ and all measurable $S\subseteq\Range(\M)$,
\[
    \Prob\{\M(D)\in S\} \leq \e^{\varepsilon}\Prob\{\M(D')\in S\}.
\]

Note that, given a data domain $\D$, we can construct~\cite{chatzikokolakis2013Broadening} a \textit{canonical metric $d^\G_\D$} for each granularity $\G$ over $\D$ by defining the distance between two databases $d^{\G}_\D(D,D')$ as the minimum number of neighboring databases in $\D$ you need to cross to obtain $D'$ from $D$ (with $d^\G_\D(D,D')=\infty$ if it is not possible). In particular, note that $d^{\G}_\D(D,D')=0$ if and only if $D=D'$, and $d^{\G}_\D(D,D')=1$ if and only if $D\neigh_\G D'$ (and $D\neq D'$).
See \Cref{prop:PrivacySpaceWellDefined} for more details and the proof of well-definition.

Then, from the group property of DP (\Cref{prop:GroupPrivacy}), $\M$ is $\G$ $\varepsilon$-DP if and only if for all $D,D'\in\D$ and all measurable $S\subseteq\Range(\M)$,
\[
    \Prob\{\M(D)\in S\} \leq \e^{\varepsilon d^\G_\D(D,D')}\Prob\{\M(D')\in S\}.
\]

This property motivates using metrics to measure privacy protection. As mentioned, this idea first appeared in~\cite{mcsherry2009privacy} with $d^\U_{\DX}(D,D')=|D\triangle D'|$ (note that $d^\U_{\DX}=d^\triangle_{\DX}$, but not generally over $\D$). Later, a formal generalization called $d_\D$-privacy~\cite{chatzikokolakis2013Broadening} was introduced. We consider the variant~\cite{chatzikokolakis2013Broadening} modeled by an \textit{extended pseudometric} $d_{\D}\colon\D^2\to[0,\infty]$, i.e., a metric in which the distance between two different databases can also be $0$ and $\infty$. To simplify the terminology, we will simply refer to $d_\D$ as \textit{metrics}. Note that having a metric $d_{\D}$ implies that $(\D,d_\D)$ is a (pseudo)metric space, which we will call \textit{privacy space}.
\begin{definition}[$d_{\D}$-privacy~\cite{chatzikokolakis2013Broadening}]\label{def:d-privacy}
    Let $(\D,d_{\D})$ be a privacy space. Then, a randomized mechanism $\M$ with domain $\D$ is \textit{$d_\D$-private} if for all $D, D'\in\D$ and all measurable $S\subseteq\Range(\M)$,
    \begin{equation}\label{eq:d-privacy}
        \Prob\{\M(D)\in S\} \leq \e^{d_{\D}(D,D')}\Prob\{\M(D')\in S\}.
    \end{equation}
\end{definition}

Observe that the metric absorbs the privacy budget ($\varepsilon$), i.e., $d_\D$ can be written as $d_\D=\varepsilon d'_\D$ where $d'_\D$ is also a metric. We will also denote it simply as $d$ when possible.

Additionally, we obtain the following result:

\begin{theorem}[restate = THdprivacytoDP, name = ]\label{th:d-privacyToDP}
    Let $\G$ be a granularity notion over the database class $\D$. Then, a mechanism $\M$ with domain $\D$ is $\varepsilon d^\G_\D$-private if and only if it is $\G$ $\varepsilon$-DP.
\end{theorem}

Given any granularity notion we can obtain a metric, but that not all metrics (even up to $\varepsilon$) are the canonical metric for a granularity notion. Therefore, the notion of $d_\D$-privacy is more general than $\G$ $\varepsilon$-DP. Besides, note that the restriction of $d^\G_\D$ to the subclass $\D'\subseteq\D$ is not always~$d^\G_{\D'}$
(see \Cref{RE:InducedVsIntrinsic}).
  
To understand the real privacy implications of a $d_{\D}$-privacy, we need to look at the domain and the distance.

The domain, $\D$, encodes what information we consider public knowledge and what we want to protect up to $d_\D$. The larger the domain, the greater the privacy, but it also comes with the cost of greater sensitivities and harder-to-achieve privacy protection. The distance, $d_\D$, encodes how hard it is to distinguish any pair of databases, and therefore what information we are protecting.

Additionally, it is important to select domains with compatible metrics. For example, information may be disclosed if $d_{\D}(D,D')=\infty$. Therefore, \textit{connected} privacy spaces (i.e., $d_\D(D,D')<\infty$ for all $D,D'\in\D$) are preferable because the change across connected components is not guaranteed to be protected by a DP mechanism. For example, when $\D$ is totally disconnected, we can end up with nonsensical privacy guarantees like in the following example.

\begin{example}\label{ex:UnboundedRestrictedDomainFail}
    Consider $\D\coloneqq\{D\in\DX\mid |D|=N\}$, the class of tabular databases of size $N$, and choose the unbounded granularity notion. It is clear that unbounded-neighboring databases always differ by one element. Therefore, there are no unbounded-neighboring databases in~$\D$ (i.e., the privacy space is totally disconnected). 
    
    This privacy space would imply, by \textit{reductio ad absurdum}, that \textit{any} mechanism is unbounded $\varepsilon$-DP for all $\varepsilon\geq 0$ since for all the neighbors (none) the definition holds. In particular, the identity mechanism (such that $\M(D)=D$) defined over $\D$ (which does not provide any protection) is unbounded $0$-DP. 
\end{example}

Note that choosing $\DX$ as the domain does not lead to the same problem, but as we mentioned before, relaxing the domain so that it is defined for subsets $\D$ of $\DX$ and other database types is usually more convenient, coherent, and necessary. 
Following the same line, the bounded granularity defines a connected privacy space over $\D\coloneqq\{D\in\DX\mid |D|=N\}$, but defines a disconnected one over $\DX$.

\subsection{Relationship between Metrics}\label{sec:RelationsMetrics}
So far, we have described a mathematical model to understand any metric and granularity notion. This will be necessary for the following sections to define general properties and theorems. However, we need to understand the real privacy implications of metrics given their formal definition.

The notion of $d_\D$-privacy allows us to compare the privacy level between metrics over the same domain, which also helps to extend composability notions proven for one to others. Consider two metrics, $d_1$ and $d_2$, over $\D$ such that $d_1\leq d_2$ (pointwise). In this case, we can say that $d_1$ offers more protection than $d_2$ because any mechanism $\M\colon\D\to\S$ that satisfies $d_1$-privacy also satisfies $d_2$-privacy~\cite{chatzikokolakis2013Broadening}. 

In particular, given two canonical metrics $d^{\G_1}_\D$ and $d^{\G_2}_\D$ such that
\[
    k = \TBD{\G_1}{\G_2}{\D} \coloneqq \max_{\substack{D,D'\in\D\\D\neigh_{\G_2}D'}}d^{\G_1}_\D(D,D')<\infty,
\]
we obtain $d^{\G_1}_\D\leq kd^{\G_2}_\D$ (see \Cref{prop:GranRelation}).
Therefore, if $\M\colon\D\to\S$ is $\G_1$ $\varepsilon$-DP, then $\M$ is $\G_2$ $k\varepsilon$-DP. This fact allows us to compare different granularity notions over the same domain, e.g., all information protected by $\G_1$ must also be protected by $\G_2$, while not necessarily the other way around. 

From this result, we can deduce the well-known fact that unbounded $\varepsilon$-DP implies bounded $2\varepsilon$-DP in $\DX$~\cite{li2016Differential} since $\TBD{\U}{\B}{\DX}=2$. However, $\TBD{\B}{\U}{\DX}=\infty$ because $d^{\B}_{\DX}(D,D')=\infty$ for all $D\neigh_{\U} D'$. Note that the privacy-level comparison between two granularity notions directly depends on which class we compare them in. While this result holds in $\DX$, we saw in \Cref{ex:UnboundedRestrictedDomainFail} that this is not the case for all database classes.

If the diameter of $(\D,d_1)$, $\diam(\D,d_1)\coloneqq\max d_1$, 
is bounded, we can always compare it to the other metrics over $\D$. For example, the \textit{free-lunch} granularity notion $\FL$~\cite{kifer2011No} is defined such that all pairs of databases are free-lunch neighboring, i.e., $d^{\FL}_\D(D,D')=1$ for all $D\neq D'$. Therefore, $d^{\FL}_\D\leq d^\G_\D$ verifies for any canonical metric $d^\G_\D$, and thus free-lunch DP implies all others.

\subsection{Changing the Privacy Space}
It is also interesting to understand how queries or other transformations can produce a transition from one privacy space to another and how this change can be reflected in our overall privacy.
\begin{definition}[Sensitivity~\cite{chatzikokolakis2013Broadening}]
    Let $(\D_1,d_{1})$ and $(\D_2,d_{2})$ be two privacy spaces and let $f\colon\D_1\to\D_2$ be a deterministic map. We define the \textit{sensitivity of $f$} with respect to $d_{1}$ and $d_{2}$ as the smallest value $\Delta f\in[0,\infty]$ such that 
    $d_{2}(f(D),f(D'))\leq \Delta f\,d_{1}(D,D')$ holds for all $D,D'\in\D_1$ with $d_{1}(D,D')<\infty$.
\end{definition}

\begin{proposition}[restate = PRpreprocessing, name = Preprocessing~\cite{chatzikokolakis2013Broadening}]\label{prop:preprocessing}
    Let $(\D_1,d_{1})$ and $(\D_2,d_{2})$ be two privacy spaces and let $f$ be a deterministic map with sensitivity $\Delta f<\infty$ with respect to $d_{1}$ and $d_{2}$, and let $\M\colon\D_2\to\S_2$ be a $d_{2}$-private mechanism. Then $\M\circ f$ satisfies $(\Delta f)d_{1}$-privacy.
\end{proposition}

In the case where the metrics are the canonical metrics of granularities $\G_1$ and $\G_2$, we obtain that the sensitivity is $\Delta f\coloneqq\max_{D\neigh_{\G_1}D'}d_{2}(f(D),f(D'))$. If we then choose $f=\id\colon\D\to\D$, we obtain that $\Delta\id = \TBD{\G_1}{\G_2}{\D}$.

\begin{remark}\label{re:ReciprocalPreprocessing}
    The reciprocal of \Cref{prop:preprocessing} is not true. For example, consider $(\D_1,d_{1})=(\D_2,d_{2})=(\R,d^{\FL})$, the free-lunch metric over $\R$. Take $\M\colon\R\to\R$ such as $\M(x)=x+Z$ where $Z\sim\Lap(\frac{1}{\varepsilon})$. We can easily verify that this mechanism is not free-lunch DP by selecting two numbers $x<<y$. In other words, the sensitivity of the identity map over the real numbers is infinite. However, if we take $f\colon\R\to\R$ such that $f(x)=\frac{1}{1+\e^x}$, then $\Delta f =\|f(x)-f(y)\|_1\leq 1$. Therefore, $\M\circ f$ corresponds to a Laplace mechanism, and is free-lunch $\varepsilon$-DP\@. In conclusion, there exist $\M$ and $f$ such that $\M\circ f$ is $(\varepsilon\Delta f)d^{\FL}_\R$-private but $\M$ is not $\varepsilon' d^{\FL}_\R$-private for any $\varepsilon'>0$.
\end{remark}  

We can also apply multiple preprocessing functions to a mechanism, obtaining the following bound:

\begin{proposition}[restate = PRsensitivitycomposition, name = Sensitivity of the composition]\label{prop:SensitivityComposition}
    Let $(\D_1,d_{1})$, $(\D_2,d_{2})$ and $(\D_3,d_{3})$ be privacy spaces and let $f\colon\D_1\to\D_2$ and $g\colon\D_2\to\D_3$ be two deterministic maps. Then $\Delta(g\circ f)\leq\Delta f\,\Delta g$.
\end{proposition}

\section{The Independent Composition Theorem}\label{sec:IndependentComposition}
Now that we have a general framework for DP in arbitrary privacy spaces, we can start to explore how we can extend the properties of DP from $(\DX,d^\U_{\DX})$ to the other privacy spaces. In this section, we focus this analysis on the independent composition. To this end, we present a theorem that models all possible independent compositions of mechanisms over arbitrary privacy spaces. To begin, we first need to understand composability, in its more general form. 

\subsection{Composing Mechanisms}
Assuming the role of the curator, we have a database $D\in\D$ and we want to publish certain extracted information $s\in\S$. However, we cannot publish $s$ directly because it would compromise privacy. Therefore, we want an attacker with access only to the output $\Tilde{s}$ of our mechanism to be unable to distinguish aspects of $D$ from other databases of $\D$.
Besides, the information we need to extract can be obtained as a function of some query answers. That is, $s=h(s_1,\dots,s_k)$ where $h$ is an arbitrary deterministic function and $s_i=f_i(D)$ is the output of an arbitrary query (where $f_i$ can even be the identity). Thus, by trying to get every $\Tilde{s}_i$ (private output of $s_i$) and computing $\Tilde{s}=h(\Tilde{s}_1,\dots,\Tilde{s}_k)$, we make it possible to discretize our problem.
To do this, we use $k$ $d_i$-private mechanisms
$\M^*_i\colon\D_i\to\S_i$ such that $\M^*_i(f_i(D))=\Tilde{s}_i$. Therefore, the question arises whether the composition of the mechanisms $\M$ such that $\M(D)=(\M^*_1(f_1(D)),\dots,\M^*_k(f_k(D)))$ for all $D\in\D$ is $d_{\D}$-private, and what privacy $d_{\D}$ implies. To answer this question, we state and prove the independent/adaptive composition theorems (\ref{th:ICTheoremVariableDomain} and~\ref{th:ACTheoremVariableDomain}). Note that $\M_i\coloneqq\M^*_i\circ f_i$ defines a mechanism over $\D$ for all $i\in[k]$. 

In \Cref{sec:CommonDomainIndependent}, we will explore the scenario where, instead of imposing $\M^*_i$ to be $d_i$-private, we directly impose $\M_i$ to be $d_i$-private. Since each $\M_i$ is defined over the same domain as $\M$, we call this scenario \textit{common domain}. This change is significant because it allows us to prove alternative theorems~(\ref{th:ICTheoremCommonDomain} and~\ref{th:ACTheoremCommonDomain}) to our composition results~(\Cref{th:ICTheoremVariableDomain,th:ACTheoremVariableDomain}, respectively), and it ensures that the composed mechanism does not completely lose the privacy guarantee, as it happens in \Cref{ex:BoundedParallelFail}.
As a result, we can provide tighter bounds on the privacy loss for cases such as bounded DP, which are not covered outside the common domain.

\subsection{Independent Composition}\label{sec:IndependentCompositionSubsection}

In this section, we introduce our generalized version of independent composition. We will explore its adaptive counterpart in \Cref{sec:GeneralizingAdaptive}. Note that adaptive composition includes independent composition, but we present the results for the independent case first to simplify the notation.

Formally, independent composition refers to the case where the mechanisms $\M_1,\dots,\M_k$ are \textit{mutually independent}, i.e., $\M_1(D),\dots,\M_k(D)$ are mutually independent random elements for all $D\in\D$. In other words, the output of each of these mechanisms does not depend on the others. The \textit{independent-composed mechanism} $\M\coloneqq(\M_1,\dots,\M_k)_{\ind}$ is then defined as the mechanism with domain $\D$ such that $\M(D)=(\M_1(D),\dots,\M_k(D))$ for all $D\in\D$.

With this definition, we can state the independent composition (IC) theorem. Since the theorem does not impose any condition on the privacy metric of the initial $\M_i$, our results can be used for any privacy space and any possible independent composition strategy. 

\begin{theorem}[restate = THictheoremvariabledomain, name = IC theorem]\label{th:ICTheoremVariableDomain}
    Let $\D$ be a database class and, for all $i\in[k]$, let $(\D_i,d_{i})$ be a privacy space, and let $f_i\colon\D\to\D_i$ be a deterministic map. 
    For all $i\in[k]$, let $\M^*_i\colon\D_i\to\S_i$ be mutually independent $d_i$-private mechanisms.
    Then mechanism $\M=(\M^*_1\circ f_1,\dots,\M^*_k\circ f_k)_{\ind}$ is $d_{\D}$-private with
    \[
        d_\D(D,D') \coloneqq \sum^k_{i=1} d_i(f_i(D),f_i(D'))\quad \text{for all $D,D'\in\D$}.
    \]
\end{theorem}

It is important to note that the IC theorem~(\ref{th:ICTheoremVariableDomain}) provides the privacy level of the resulting mechanism by construction. This means that we cannot generally impose the privacy level of the composed mechanism $\M$, but we can compute it as we see in the following example.

\begin{example}\label{ex:IC_mixdist}
    Let $\X=\X_1\cup\X_2$ be a set of locations in $\R^2$ of two districts $i\in[2]$, each associated with hospital $i$ in location $l_i$, and consider $\D=\DX$, consisting of databases of locations from ambulances in both districts. Assume that the maximum Euclidean distance between any two points in $\X_1$ and $\X_2$ in the districts is equal and finite, $\diam(\X_1)=\diam(\X_2)=L$.
    Our goal is to compute the number of locations in each district and determine the closest ambulance to each hospital. To do so, we will compose the following $d$-private mechanisms:
    A $d^{\U}_\D$-private mechanism $\M^*_a\colon\D\to\mathbb{N}$ that outputs the noisy count of records in $D\in\D$, and a $d^{\textrm{Eu}}_{\X}$-private mechanism $\M^*_b\colon\X\to\X$, with $d^{\textrm{Eu}}_{\X}$ the Euclidean distance over $\X$, that given $x\in\X$ outputs a perturbed version of it.
    
    For all $i\in[2]$ and $D\in\D$, let $p_i(D)=D\cap\X_i$ be the subset of locations of $D$ in district $i$, and let $f_i(D)=\argmin_{x\in p_i(D)}\{\|x-l_i\|_2\}$ be the closest ambulance to hospital $i$. Thus, we can obtain the wanted information through the composed mechanism $\M$ such that $\M(D)=(\M^*_a(p_1(D)),\M^*_a(p_2(D)),\M^*_b(f_1(D)),\M^*_b(f_2(D)))$.
    
    Now using the IC theorem (\ref{th:ICTheoremVariableDomain}), we can compute the privacy that $\M$ provides. For all $D,D'\in\D$, we have a protection of $d_\D(D,D')\coloneqq \sum^2_{i=1} (d^{\U}_\D(p_i(D),p_i(D'))+d^{\textrm{Eu}}_\X(f_i(D),f_i(D')))\leq (d^{\U}_\D+2d^{\infty}_\D)(D,D')\leq d^{\U}_\D(D,D')+2L$
    with $d^{\infty}_\D(D,D')=\max_{x\in D,x'\in D'}d^{\textrm{Eu}}_\X(x,x')$ the maximum distance.
\end{example}

Note that in the IC theorem (\ref{th:ICTheoremVariableDomain}), we can end up with extreme cases where $d_\D(D,D')=\infty$ for certain $D,D'\in\D$, which does not provide privacy between these databases. However, we can still obtain reasonable $d_\D$ in general cases where $d_\D$ possesses good privacy properties.

For $\D_i=\D$ and $f_i=\id$, we obtain a result reminiscent of the sequential composition theorem:
\begin{theorem}[restate = THgeneralizedISC, name = Generalized ISC]\label{th:GeneralizedISC}
    Let $\{(\D,d_i)\}_{i\in[k]}$ be a set of privacy spaces.
    For all $i\in[k]$, let $\M_i\colon\D\to\S_i$ be mutually independent $d_i$-private mechanisms.
    Then $\M=(\M_1,\dots,\M_k)_{\ind}$ is $(\sum_{i=1}^k d_i)$-private.
\end{theorem}

Note that by choosing $d_i=\varepsilon_i d$, we obtain that $\M$ is $\varepsilon d$-private with $\varepsilon=\sum^k_{i=1}\varepsilon_i$ (first proven in~\cite{galli2023Group}). Furthermore, by selecting $d$ as $d^\G_\D$, we obtain the sequential composition theorem for every granularity: If $\M_i\colon\D\to\S_i$ are mutually independent $\G$ $\varepsilon_i$-DP mechanisms, then $\M=(\M_1,\dots,\M_k)_{\ind}$ is $\G$ $(\sum_{i=1}^k \varepsilon_i)$-DP. This shows that sequential composition works as expected for every granularity. 

On the other hand, the setting in which the mechanisms take as input disjoint subsets of the initial database (as in parallel composition) does not generally yield analogous results to \Cref{th:ParallelMcSherry}.
We can model this setting by taking $f_i$ in the IC theorem~(\ref{th:ICTheoremVariableDomain}) so they define a partitioning function. More formally, we define a \textit{$k$-partitioning function} $p=\{p_1,\dots,p_k\}$ as a function where $p_i\colon\D\to p_i(\D)\eqqcolon \D_i$ such that $p_i(D)\subseteq D$ with $p_i(D)\cap p_j(D)\neq \varnothing$ for $i\neq j$\footnote{We do not require that $D=\bigcup^k_{i=1} p_i(D)$, i.e., our partition can be \textit{non-exhaustive}.}. Note, therefore, that the domains $\D_i$ of $\M_i$ might be different in this setting by construction. Let us see an example of a partitioning function, based on that of~\cite{li2016Differential}.  

\begin{example}[Partitioning function for $\D\subseteq\DX$]\label{ex:partition}
    Let $\D\subseteq\DX$. A partition $\{\X_i\}_{i\in[k]}$ of $\X$, extends naturally as a partition of the elements $D\in\D$, i.e., $p_i(D)\subseteq D$ is the multiset such that element $x\in D$ has multiplicity $m_{p_i(D)}(x)=\mathbf{1}_{\X_i}(x)\,m_{D}(x)$. In this case, the partitioning function $p$ uses only $x$ to compute the value of $p(x)$, and therefore the result is independent of the other records. 
\end{example}

In this setting, the IC theorem (\ref{th:ICTheoremVariableDomain}) yields that $\M$ is $d_\D$-private with $d_\D(D,D')=\sum^k_{i=1}d_{i}(p_i(D),p_i(D'))$ $\leq I_p(D,D')(\max_{i\in[k]}\Delta p_i)d_{i}(D,D')$ for all $D,D'\in\D$, where $I_p(D,D')\coloneqq\#\{i\mid p_i(D)\neq p_i(D')\}$. This fact is coherent with what we know: Assuming a partitioning function of \Cref{ex:partition}, if we select $\varepsilon_i d^{\U}_{\D_i}$ mechanisms then  $d_\D\leq(\max_{i\in[k]}\varepsilon_i)d^\U_\D$, since $\Delta p_i=\varepsilon_i$ and $I_p(D,D')=1$ for all $D\neigh_\U D'$. 
If we select $d_i=\varepsilon_i d^{\B}_{\D_i}$, there exist $D,D'\in\D$, as we saw in \Cref{ex:BoundedParallelFail}, such that $d_i(D,D')=d^{\B}_{\D_i}(p_i(D),p_i(D'))=\infty$ for some $i$ and therefore $d_\D(D,D')=\infty$.
In general, we have no better expression for $d_\D$ unless we add extra conditions. In \Cref{sec:GeneralizationOfIndependentParallelComposition,sec:CommonDomainIndependent}, we will explore conditions to achieve the best bound in this setting.

Furthermore, between accessing the whole database (\Cref{th:GeneralizedISC}) or a partition of it, the IC theorem (\ref{th:ICTheoremVariableDomain}) allows considering intermediate composition strategies that provide tighter, more-precise bounds, such as shown in the following example:

\begin{example}\label{ex:IntermediateSetting}
    We continue with the scenario presented in \Cref{ex:IC_mixdist}, but now we have $k>3$ hospitals and each ambulance has at least three associated hospital locations. The universe of records in this case is $\X'=(\X,[k]^{\leq3})$ and $\D=\D_{\X'}$, where $[k]^{\leq3}$ denotes the subsets of at least three elements of $[k]$. We consider the analogous $\M^*_a$ mechanism. We want to know the number of available ambulances for each hospital, so we consider $\M$ such that $\M(D)=(\M^*_a(f_1(D)),\dots,\M^*_a(f_k(D)))$ where $f_i(D)$ is the subdatabase of $D\in\D$ of ambulances assigned to hospital~$i$. Since each ambulance only collaborates with at most three hospitals, $I_f(D,D')\leq3d^\U_\D(D,D')$. Applying then the IC theorem~(\ref{th:ICTheoremVariableDomain}), we obtain that $\M$ is $d_\D$-private with $d_{\D}(D,D')=\sum_{i=1}^{k}d^{\U}_\D(f_i(D),f_i(D'))\leq3d^{\U}_\D(D,D')<kd^{\U}_\D(D,D')$. 
\end{example}

In particular, the last example showcases this intermediate setting. Function $f=\{f_i\}_{i\in[k]}$ does not define a partition, so we cannot apply \Cref{th:ParallelMcSherry}, but a single change to database $D$ affects at most three databases in $\{f_i(D)\}_{i\in[k]}$, hence the final budget is $d_\D\leq3d^\U_\D$ instead of $kd^{\U}_\D$ given by the sequential counterpart (\Cref{th:GeneralizedISC}).

\subsection{A Better Bound for Disjoint Inputs}\label{sec:GeneralizationOfIndependentParallelComposition}
Following the discussion in \Cref{sec:IndependentCompositionSubsection}, considering as input disjoint subsets of the initial database, we explore the possibility to obtain the best possible bound. For this section, we assume that mechanisms $\M_i$ are $d_i$-private, with $d_i$ ``proportional'' to a single metric type (e.g., $d_i=\varepsilon_i d^\triangle_{\D_i}$) or over a fixed granularity (i.e., $d_i=\varepsilon_i d^\G_{\D_i}$). 

\Cref{th:ParallelMcSherry} tells us that if $\M_i$ are $\varepsilon_i d^{\triangle}_{\DX}$-private, then the composed mechanism $\M=(\M_1,\dots,\M_k)_{\ind}$ is $(\max_{i\in[k]}\varepsilon_i)d^{\triangle}_{\DX}$-private. This privacy bound is the best possible bound we can get in this setting. Note that in the case where mechanisms $\M_i$ satisfy the same privacy guarantee ($\varepsilon d^{\triangle}_{\DX}$-privacy) for all $i\in[k]$, then $\M$ also satisfies it. Thus, the composition does not degrade the privacy level at all. However, as we mentioned before, the best bound cannot be obtained for all metrics. Therefore we explore in this section which additional conditions the partition must satisfy (with respect to the metric) to ensure that we obtain the best-case bound, the maximum privacy budget of $\M_i$.

The first case we consider is a metric-type $d^*$ that is well-defined over $\D$ and $\D_i$ for all $i\in[k]$\footnote{This means that metrics $d^*_\D$ and $d^*_{\D_i}$ are well-defined metrics and that $d^*(D,D')$ is constant for all domains containing $D,D'\in\D$. Examples include $d^\triangle$, which is well-defined for all $\D\subseteq\DX$.}. We can give a sufficient condition for obtaining the best bound: We say that metric $d^*$ \textit{commutes with the partition given by $p$} if, for all $D,D'\in\D$,
\begin{gather*}
    \sum^k_{i=1} d^*_{\D_i}(p_i(D),p_i(D'))=d^*_\D\bigg(\bigcup_{i=1}^k p_i(D),\bigcup_{i=1}^k p_i(D')\bigg) \leq d^*_\D(D,D').
\end{gather*}

By the IC theorem~(\ref{th:ICTheoremVariableDomain}), if $d^*$ commutes with $p$ and $\M_i$ are $\varepsilon_i d^*_{\D_i}$-private, then $\M$ is $(\max_{i\in[k]}\varepsilon_i) d^*_\D$-private.
For example, $d^\triangle$ commutes with all partitions $p$ of \Cref{ex:partition} (see \Cref{prop:commutative}),
which relates to the original result of \citeauthor{mcsherry2009privacy}~\cite{mcsherry2009privacy}.

Secondly, we can also focus on
a fixed granularity notion $\G$, and given $d_i=\varepsilon_i d^\G_{\D_i}$ for all $i\in[k]$, we study when we obtain that $\M$ is $\varepsilon d^\G_\D$-private with $\varepsilon=\max_{i\in[k]}\varepsilon_i$.
Recall that different domains define different canonical metrics, so the previous case does not apply, and checking commutativity is not an option. In this case, the corresponding equation translates to $\sum^k_{i=1} d^\G_{\D_i}(p_i(D),p_i(D'))=d^\G_\D(D,D')$. This equation can be hard to check in general, but it holds if the partition verifies:
\begin{itemize}
    \item \textit{$d^\G_\D$-compatibility}: For all $\G$-neighboring $D,D'\in\D$, there exists at most one $j\in[k]$ such that $p_i(D)=p_i(D')$ for all $i\neq j$, i.e., $I_p(D,D')=1$ for all $D\neigh_\G D'$; and
    \item $\G$ is also well-defined over $\D_i$ and the sensitivity of $p_i$ with respect to $d^\G_\D$ and $d^\G_{\D_i}$ is $\Delta p_i\leq 1$ (i.e., $d^\G_{\D_i}(p_i(D),p_i(D'))\leq 1$ if $d^\G_{\D}(D,D')=1$).
\end{itemize}

Under these conditions, we obtain the desired result (where $\M^*_i$ can have different domains):
\begin{theorem}[restate = THgeneralizedIPCvariabledomain, name = IC best bound for disjoint inputs]\label{th:GeneralizedIPCVariableDomain}
    Let $\D$ be a database class and $\G$ a granularity over $\D$. Let $p$ be a $d^\G_\D$-compatible $k$-partitioning function such that $\Delta p_i\leq 1$. 
    For all $i\in[k]$, let $\M^*_i\colon\D_i\to\S_i$ be mutually independent $\varepsilon_i d^\G_{\D_i}$-private mechanisms.
    Then mechanism $\M=(\M^*_1\circ p_1,\dots,\M^*_k\circ p_k)_{\ind}$ is $\varepsilon d^\G_{\D}$-private with $\varepsilon=\max_{i\in[k]} \varepsilon_i$.
\end{theorem}

As discussed, all partitions $p$ of \Cref{ex:partition} are $d^{\U}_{\D}$-compatible since the addition/removal of one record can only affect the partition this record belongs to, so $I_p(D,D')=1$ for all $D\neigh_{\U} D'$, and additionally $\Delta p_i\leq 1$. Therefore, \Cref{th:GeneralizedIPCVariableDomain} can be applied to obtain \Cref{th:ParallelMcSherry}.

Even though \Cref{th:GeneralizedIPCVariableDomain} is stated for any granularity, $d^\G_{\D}$-compatibility is a strict condition.
For example, no partitioning function of \Cref{ex:partition} (with $k>1$) is $d_{\DX}^{\B}$-compatible
(see \Cref{prop:CompatiblePartitionBounded}).
Nevertheless, we can construct compatible partitioning functions to certain bounded metrics $d^\B_\D$, as shown in the following example:

\begin{example}[A $d^\B_\D$-compatible partition]\label{ex:compatible}
    Consider a database $D$ with ordered elements, i.e., every element $(n,x)\in D$ consists of a record value $x\in\X$ and a unique identifier $n\in[|D|]$. Let $\DX^{\mathrm{ord}}$ denote the class of all such databases.
    
    Let $p$ be a $k$-partitioning function of $\N$, which induces a partition of the elements of $\D\subseteq\DX^{\mathrm{ord}}$ that divides the databases only taking the order into account, i.e., such that $p(n,x)=p(n,y)$ for all $x,y\in\X$. Then $p$ is $d_{\D}^{\B}$-compatible and verifies $\Delta p_i\leq 1$ (see proof in \Cref{prop:OrderBasedCompatiblePartitionBounded}).
    Therefore, we can obtain the best bound for bounded in this case using \Cref{th:GeneralizedIPCVariableDomain}. 
\end{example}

\subsection{Common-Domain Setting}\label{sec:CommonDomainIndependent}
The common-domain setting relates to the perspective in  which $\M_i=\M^*_i\circ f_i$ are $d_i$-private instead of $\M^*_i$, i.e., $\M_i$ and $\M$ have the same ``common'' domain $\D$. This change provides new composition rules that allow us to obtain better privacy bounds. Importantly,  when we impose the privacy constraints in $\M^*_i$, in the case where $d_i(D,D')$ are well-defined and finite, we can still end up with $d_{\D}(D,D')=\infty$, as we saw in \Cref{ex:BoundedParallelFail}. However, if $\M_i$ are $d_i$-private, we can bound the privacy loss by at least $d_\D(D,D')=\sum^k_{i=1}d_i(D,D')<\infty$, avoiding this problem.

In this scenario, while $\M_i$ can protect any database of $D\in\D$, the computation of $\M_i$ depends exclusively on the information of contained $f_i(D)$ and not the total information of $D$. 
The exclusive dependence of $\M_i$ on specific information improves the privacy guarantee and gives better privacy-loss bounds. To be able to analyze the composition in this setting, we present a coherent formalization of ``depending exclusively on $f_i(D)$'' under the notion of \textit{dependency}:

\begin{definition}[Dependency]
    Let $\M\colon\D\to\S$ be a randomized mechanism, and let $f$ be a deterministic map with domain $\D$.
    We say that $\M$ is \textit{$f$-dependent} if there exists $\M^*\colon f(\D)\to\S$ such that $\M=\M^*\circ f$.
\end{definition}
This definition implies that $\Prob\{\M^*(f(D))\in S\}=\Prob\{\M(D)\in S\}$ for all measurable $S\subseteq\S$. Since $\M^*(f(D))$ depends exclusively on $f(D)$, consequently $\M(D)$ depends exclusively on the information in $f(D)$ for all $D\in\D$ (i.e., only data in $f(D)$ affects the output of $\M(D)$). 
\begin{example}[Dependency]\label{ex:dependency}
    Let us revisit the scenario of \Cref{ex:IC_mixdist}. 
    For $i\in[2]$, we define $\M_{i}\colon\DX\to \N$ such that it outputs the noisy participants count from district $i$ in $D$, i.e., $\M_i(D)=\sum_{x\in\X_i} m_{D}(x)+z$
    with $z\sim\Lap(\frac{1}{\varepsilon_i})$ (note it is the Laplace mechanism). 
    Mechanisms $\M_{i}$ are $p_{i}$-dependent, since there exists $\M^*_i(D)=|D|+z$ such that $\M_{i}=\M^*_i\circ p_{i}$. This means that, even though $\M$ takes as input the whole database $D$, it just needs to see the information contained in subset $p_i(D)$ to know how many locations belong to district $i$. 
\end{example}

Under this definition, we arrive at the following result:
\begin{proposition}[restate = PROPmininumprivacy, name = Minimum privacy]\label{prop:MinimumPrivacy}
    Let $(\D,d_\D)$ be a privacy space, let $f$ be a deterministic map with domain $\D$, and let $\M\colon\D\to\S$ be a $d_{\D}$-private mechanism. If $\M$ is $f$-dependent, then $\M$ is $d_{\D}^{f}$-private* with
    \[
        d_{\D}^f(D,D') \coloneqq \min_{\substack{\tilde{D},\tilde{D}'\in\D \\ f(\tilde{D})=f(D)\\ f(\tilde{D}')=f(D')}} d_\D(\tilde{D},\tilde{D}').
    \]
\end{proposition}

Note that $d^f_\D$ is not necessarily a metric\footnote{It does not generally fulfill the triangle inequality.} (thus we call it \textit{$d$-privacy*}).
However, it gives an accurate value for the distance between the probability distributions of the output given two input databases. Since $d_{\D}^f\leq d_\D$,  having the dependency constraint in a mechanism can imply more privacy. This way, the privacy loss is chosen as the minimum with respect to the dependent data $f(D)$, and not $D$. In particular, if $f(D)=f(D')$ for a pair $D,D'\in\D$, then $d_{\D}^{f}(D,D')=0$. Furthermore, it is possible to find metrics $d$ in-between these, i.e., $d^f_\D\leq d\leq d_\D$.

Applying \Cref{prop:MinimumPrivacy} to the IC theorem~(\ref{th:ICTheoremVariableDomain}), we obtain:

\begin{theorem}[restate = THictheoremcommondomain, name = IC theorem for common domain]\label{th:ICTheoremCommonDomain}
    For $i\in[k]$, let $(\D,d_i)$ be a privacy space, and let $f_i$ be a deterministic map over $\D$.
    For all $i\in[k]$, let $\M_i\colon\D\to\S_i$ be mutually independent mechanisms satisfying $d_i$-privacy and $f_i$-dependency.
    Then mechanism $\M=(\M_1,\dots,\M_k)_{\ind}$ is $d_\D$-private* with $d_\D\coloneqq\sum^k_{i=1} d_i^{f_i}$.
\end{theorem}

We can also bound the result with
\[
    d_\D(D,D') = \sum^k_{i=1} d_i^{f_i}(D,D') \leq 
    \sum_{i\,:\,f_i(D)\neq f_i(D')} d_i(D,D'),
\]
which are not metrics, but are better bounds than $\sum^k_{i=1} d_i$ given by the IC theorem~(\ref{th:ICTheoremVariableDomain}). Translating this result to the case of granularities, if we take $\M_i$ to be $\G$ $\varepsilon_i$-DP (i.e., $\varepsilon_i d^\G_\D$-private), we obtain that $\M$ is $\G$ $\varepsilon$-DP (i.e., $\varepsilon d^\G_\D$-private) with
\[
    \varepsilon = \max_{D\neigh_\G D'} 
    \sum_{i\,:\,f_i(D)\neq f_i(D')} \varepsilon_i.
\]

\Cref{th:ICTheoremCommonDomain} allows us to obtain the corresponding cases, corollaries, and examples to those we obtained from the IC theorem~(\ref{th:ICTheoremVariableDomain}) for this new setting. In some cases, such as taking $f_i=\id$ for all $i\in[k]$, correspond to the same result (\Cref{th:GeneralizedISC}), since $d^{\id}_{\D}=d_{\D}$. In others, however, the change of setting leads to a different scenario and results, such as when trying to find the best bound for disjoint inputs (i.e., the counterpart of \Cref{sec:GeneralizationOfIndependentParallelComposition}).

The corresponding question of \Cref{sec:GeneralizationOfIndependentParallelComposition} translates as follows: Given $k$ mechanisms $\M_i\colon\D\to\S$ that are $d_i$-private with $d_i=\varepsilon_i d$ for a metric $d$ over $\D$ and $p_i$-dependent with $p$ an arbitrary partitioning function, we are interested in studying the conditions such that $\M=(\M_1,\dots,\M_k)_{\ind}$ is $d_\D$-private with
$d_\D=(\max_{i\in[k]}\varepsilon_i)d$. 

The natural approach is to check when metric $d$ verifies
\begin{equation}\label{eq:CommutativityCommon}
    \sum^k_{i=1} d^{p_i}(D,D') = d(D,D')
\end{equation}
for all $D,D'\in\D$, since then $d_\D=\max_{i\in[k]}\varepsilon_i d$ follows from \Cref{th:ICTheoremCommonDomain}.

\Cref{eq:CommutativityCommon} can be hard to check directly, but we can give sufficient conditions for it when $d=d^\G_\D$, the canonical distance of a granularity notion. Here, it is sufficient to ask that the partition is $d^\G_\D$-compatible.

\begin{theorem}[restate = THgeneralizedIPCcommondomain, name = IC best bound for disjoint inputs (common domain)]\label{th:GeneralizedIPCCommonDomain}
    Let $\D$ be a database class and $\G$ a granularity over $\D$. Let $p$ be a $d^\G_\D$-compatible $k$-partitioning function. For all $i\in[k]$, let $\M_i\colon\D\to\S_i$ be mutually independent mechanisms satisfying $\varepsilon_id^\G_\D$-privacy and $p_i$-dependency. 
    Then mechanism $\M=(\M_1,\dots,\M_k)_{\ind}$ is $\varepsilon d^\G_\D$-private with $\varepsilon=\max_{i\in[k]} \varepsilon_i$.

\end{theorem}

Note that in this case, it is not necessary to impose ``$\Delta p_i\leq 1$'', which was necessary for our previous  theorem~(\ref{th:GeneralizedIPCVariableDomain}). \Cref{th:GeneralizedIPCCommonDomain} is therefore also a consequence of preprocessing (\Cref{prop:preprocessing}) applied to \Cref{th:GeneralizedIPCVariableDomain}.

\subsection{A Better Composition for the Bounded Case over Disjoint Databases}\label{sec:BestBounded}

The strict conditions necessary  to obtain the $\max_{i\in[k]}\varepsilon_i$ bound in \Cref{th:GeneralizedIPCCommonDomain,th:GeneralizedIPCVariableDomain} cannot be achieved in the bounded case for partitions of \Cref{ex:partition}, because they are not $d^{\B}_\D$-compatible in general. This is also true for other granularities, especially those based on the bounded notion. However, even if \Cref{th:GeneralizedIPCCommonDomain,th:GeneralizedIPCVariableDomain} do not apply,
we can still compute the best-case bound when considering a partition of the database.

In this subsection, we briefly discuss how we can bound the minimum privacy budget consumed when taking a partition of the databases using \Cref{th:ICTheoremCommonDomain}. We thus provide a solution to the problem posed by \citeauthor{li2016Differential}~\cite{li2016Differential}, obtaining a tight bound for composition over disjoint databases in bounded DP (when taking a partition of \Cref{ex:partition}), which was previously missing.

\begin{corollary}[restate = COboundedparallel, name = ]\label{th:BoundedParallel}
    Let $p$ be a $k$-partitioning function of \Cref{ex:partition}. For all $i\in[k]$, let $\M_i\colon\D\to\S_i$ be mutually independent mechanisms satisfying bounded $\varepsilon_i$-DP and $p_i$-dependent.
    Then mechanism $\M=(\M_1,\dots,\M_k)_{\ind}$ with domain $\D$ is bounded $\varepsilon$-DP with $\varepsilon=\max_{i,j\in[k];\,i\neq j} (\varepsilon_i+\varepsilon_j)$.
\end{corollary}

Note that this result is stated for common domain, and that the non--common-domain counterpart cannot be defined as we prove in \Cref{ex:BoundedParallelFail}. 

Also, note that returning to the DP formulation increases the tightness of the bound with respect to the direct statement using $d^{\B}_\D$-privacy. In this case, the best bound is $\sum^k_{i=1} \varepsilon_i d^{\B,p_i}_\D$. To showcase the improvement we add the following example:

\begin{example}\label{ex:better-bound}
    We continue from \Cref{ex:dependency} but considering $k>2$ districts instead of two. We already showed that $\M_{i}$ are $p_{i}$-dependent. Besides, $\M_{i}$ are $\varepsilon_id^{\B}_{\DX}$-private because they are Laplace mechanisms. Applying \Cref{th:BoundedParallel}, we have that $\M=(\M_1,\dots,\M_n)_{\ind}$ is $\varepsilon d^{\B}_{\DX}$-private for $\varepsilon=\max_{i,j\in[k];\,i\neq j} (\varepsilon_i+\varepsilon_j)$. Particularly, given $D\neigh_{\B} D'$ with $D\triangle D'=\{x,x'\}$, $D$ and $D'$ are indistinguishable up to $\varepsilon_i$ if $x$ and $x'$ are in the same district~$i$, and up to $\varepsilon_i+\varepsilon_j$ if they are in different districts $i\neq j$. Note that this is a much better bound than applying sequential composition directly, which would give us that $D\neigh_{\B}D'$ are indistinguishable up to $\sum^k_{i=1}\varepsilon_i>\max_{i,j\in[k];\,i\neq j} (\varepsilon_i+\varepsilon_j)$.
\end{example}

We conclude this section with a small result obtained by applying \Cref{prop:MinimumPrivacy} to bounded DP in~$\DX$.
\begin{corollary}[restate = COdependencyunboundedbounded, name = ]\label{co:DependencyUnboundedBounded}
    Let $\DX$ be a database universe, $\mathcal{Y}\subsetneq\X$ and $f\colon\DX\to\D_{\mathcal{Y}}$ such that $f(D)=D\cap\mathcal{Y}$. Let $\M\colon\D_{\X}\to\S$ be a $d^{\B}_{\DX}$-private mechanism that is $f$-dependent. Then,
    $\M$ is $d_{\DX}$-private* with
    \begin{align*}
        d_{\DX}(D,D')\coloneqq{}\!\min\{d^{\B}_{\DX}(D,D'), |f(D)\triangle f(D')|\}
        \leq \!\min\{d^{\B}_{\D_{\X}}(D,D'),d^{\U}_{\D_{\X}}(D,D')\}.
    \end{align*}
\end{corollary}

\section{The Adaptive Composition Theorem}\label{sec:GeneralizingAdaptive}
In the previous section, we elaborated on composability when we apply mechanisms that work independently from each other, obtaining the IC theorem~(\ref{th:ICTheoremVariableDomain}). However, the question remains open on how composition works in the adaptive scenario, where each mechanism can also take as input the output of the previous mechanisms. In this section, we discuss adaptive composition, which is a generalization of independent composition, and provide the adaptive counterparts to the theorems of \Cref{sec:IndependentComposition}. 

To be precise, we formalize the adaptive-composed mechanism as follows:
\begin{definition}[Adaptive-composed mechanism]\label{def:AdaptiveComposedMechanism}
    For $i\in[k]$, let $\overS_{i}\coloneqq\S_1\times\cdots\times\S_{i-1}$ (where $\overS_1=\varnothing$), and let $\M_i\colon\overS_{i}\times\D\to\S_i$ be randomized mechanisms. We define the \textit{adaptive-composed mechanism} $\M\coloneqq(\M_1,\dots,\M_k)_{\adapt}$ as the mechanism with domain $\D$ such that $\M(D)=(\MM_1(D),\dots,\MM_k(D))$ for all $D\in\D$, where $\MM_i(D)$ are defined recursively as $\MM_{i}(D)=\M_{i}(\MM_{1}(D),\dots,\MM_{i-1}(D),D)$ for $i\in[k]$ (where $\MM_1=\M_1$).
\end{definition}

In other words, given $D\in\D$, $\M$ first draws $D_1$ following the distribution of $\M_1(D)$; then, $\M$ draws $D_i$ following the distribution of $\M_i(D_{1},\dots,D_{i-1},D)$ for each $i=2,\dots,k$ in order. At the end, $\M$ outputs $\M(D)=(D_1,\dots,D_k)$.

Note that adaptive-composed mechanisms are more general than independent-composed mechanisms, corresponding to the case where $\M_i$ are mutually independent and, in particular, constant over $\overS_i$.

We directly define the adaptive composition (AC) theorem. Similar to the independent results, this result does not impose any conditions on the privacy level of the initial mechanisms $\M_i$.

\begin{theorem}[restate = THactheoremvariabledomain, name = AC theorem]\label{th:ACTheoremVariableDomain}
    Let $\D$ be a database class, and, for all $i\in[k]$, let $(\D_i,d_{i})$ be a privacy space, $f_i\colon\D\to\D_i$ a deterministic map and $f^*_i = \id_{\overS_i}\times f_i$ (with $f^*_1=f_1$). 
    
    For $i\in[k]$, let $\M^*_i\colon\overS_{i}\times\D_i\to\S_i$ be a mechanism such that $\M^*_i(\overs_{i},\cdot)\colon\D_i\to\S_i$ satisfies $d_{i}$-privacy for any $\overs_{i}\in\overS_{i}$.
    
    Then mechanism $\M=(\M^*_1\circ f^*_1,\dots,\M^*_k\circ f^*_k)_{\adapt}$ is $d_{\D}$-private with
    \[
        d_\D(D,D') \coloneqq \sum^k_{i=1} d_i(f_i(D),f_i(D'))\quad \text{for all $D,D'\in\D$}.
    \]
\end{theorem}

Observe that the expression of $d_\D$ does not change with respect to the IC theorem (\ref{th:ICTheoremVariableDomain}). Therefore using adaptive composition, which is more general than independent composition, does not affect the privacy bound of the resulting mechanism; or, alternatively, no improvement is gained by considering mechanisms $\M_1,\dots,\M_k$ mutually independent. 

Analogously to the independent case, particular composition rules can be derived from \Cref{th:ACTheoremVariableDomain}, as well as translated to the common domain. The same consequences are extracted from these adaptive results. We present such results, which also generalize their respective independent cases. 

First, if we impose $f_i=\id$ and $\D=\D_i$ for all $i\in[k]$, we obtain a generalization of the sequential setting as expected:

\begin{theorem}[restate = THgeneralizedASC, name = Generalized ASC]\label{th:GeneralizedASC}
    Let $\{(\D,d_i)\}_{i\in[k]}$ be a set of privacy spaces.
    For $i\in[k]$, let $\M_i\colon\overS_{i}\times\D\to\S_i$ be a mechanism such that $\M_k(\overs_{i},\cdot)\colon\D\to\S_i$ is $d_i$-private for all $\overs_{i}\in\overS_{i}$.
    Then $\M=(\M_1,\dots,\M_k)_{\adapt}$ is $(\sum_{i=1}^k d_i)$-private.
\end{theorem}

Second, if we study what happens when we apply adaptive composition over disjoint subsets of the input, we obtain the analogous adaptive counterpart of \Cref{th:GeneralizedIPCVariableDomain}:

\begin{theorem}[restate = THgeneralizedAPCvariabledomain, name = AC best bound for disjoint inputs]\label{th:GeneralizedAPCVariableDomain}
    Let $\D$ be a database class and $\G$ a granularity over $\D$. Let $p$ be a $d^\G_\D$-compatible $k$-partitioning function such that $\Delta p_i\leq 1$,
    and $p^*_i = \id_{\overS_i}\times p_i$ (with $p^*_1=p_1$). 
    For $i\in[k]$, let $\M^*_i\colon\overS_{i}\times\D_i\to\S_i$ be a mechanism such that $\M^*_i(\overs_{i},\cdot)\colon\D_i\to\S_i$ satisfies $\varepsilon_i d^\G_{\D_i}$-privacy for any $\overs_{i}\in\overS_{i}$.
    Then mechanism $\M=(\M^*_1\circ p^*_1,\dots,\M^*_k\circ p^*_k)_{\adapt}$ is $\varepsilon d^\G_{\D}$-private with $\varepsilon=\max_{i\in[k]} \varepsilon_i$.
\end{theorem}

Note that we cannot get around the problem that $\M_i^*$ being $d$-privacy does not imply that $\M_i=\M^*_i\circ f_i$ is $d$-private in the adaptive setting.
Therefore, we show the common-domain results that show what happens if we impose $\M_i$ to be $d_i$-private directly (i.e., the counterparts to \Cref{th:ICTheoremCommonDomain,th:GeneralizedIPCCommonDomain}):

\begin{theorem}[restate = THactheoremcommondomain, name = AC theorem for common domain]\label{th:ACTheoremCommonDomain}
    For $i\in[k]$, let $(\D,d_i)$ be a privacy space, and let $f_i$ be a deterministic map over $\D$.
    For $i\in[k]$, let $\M_i\colon\overS_{i}\times\D\to\S_i$ be a mechanism such that $\M_k(\overs_{i},\cdot)\colon\D\to\S_i$ satisfies $d_i$-privacy and $f_i$-dependency for any $\overs_{i}\in\overS_{i}$.
    Then mechanism $\M=(\M_1,\dots,\M_k)_{\adapt}$ is $d_\D$-private* with $d_\D\coloneqq\sum^k_{i=1} d_i^{f_i}$.
\end{theorem}

\begin{theorem}[restate = THgeneralizedAPCcommondomain, name = AC best bound for disjoint inputs (common domain)]\label{th:GeneralizedAPCCommonDomain}
    Let $\D$ be a database class and $\G$ a granularity over $\D$. Let $p$ be a $d^\G_\D$-compatible $k$-partitioning function. 
    For $i\in[k]$, let $\M_i\colon\overS_{i}\times\D\to\S_i$ be a mechanism such that $\M_k(\overs_{i},\cdot)\colon\D\to\S_i$ satisfies $\varepsilon_id^\G_\D$-privacy and $p_i$-dependency for any $\overs_{i}\in\overS_{i}$.
    Then mechanism $\M=(\M_1,\dots,\M_k)_{\adapt}$ is $\varepsilon d^\G_\D$-private with $\varepsilon=\max_{i\in[k]} \varepsilon_i$.
\end{theorem}

Observe that, since all results are a consequence of the AC theorem (\ref{th:ACTheoremVariableDomain}), which has the same bound as its IC counterpart, none of the results degrade their bound with respect to their IC versions.

\section{Extending to Other DP-Based Notions}\label{sec:extending}

Given that composability is not an exclusive property of $\varepsilon$-DP, but also of other DP-based notions, it is interesting to understand how composition extends to other DP-based notions. In this section, we present $d_\D$-privacy formulations of approximate DP~\cite{dwork2006Our}, zero-concentrated DP~\cite{bun2016Concentrated}, and Gaussian DP~\cite{dong2022Gaussian}, and study the corresponding adaptive composition theorems. Note that since each notion has its own group property, each extension behaves differently than that of $d_\D$-privacy, although similar patterns are present.

\subsection{Extending to Approximate DP}\label{sec:ApproximateDP}
Approximate DP~\cite{dwork2006Our}, also known as $(\varepsilon,\delta)$-DP, is an important and popular extension of DP. In this section, we introduce an adapted version of $d_\D$-privacy for the approximate scenario, $(d_\D,\delta_\D)$-privacy, which generalizes $(\varepsilon,\delta)$-DP in the same way that $d_\D$-privacy generalizes $\varepsilon$-DP. Afterward, we present the composition results for this notion.

From the original definition of $(\varepsilon,\delta)$-DP~\cite{dwork2006Our,dwork2014algorithmic}, defined for unbounded neighboring databases, we present the definition of approximate DP for any granularity: 

\begin{definition}[$\G$ $(\varepsilon,\delta)$-DP]\label{def:ApproximateDP}
    Let $\varepsilon,\delta\geq0$. A randomized mechanism $\M$ with domain $\DX$ is \textit{$\G$ $(\varepsilon,\delta)$-DP} if for all measurable $S\subseteq \Range(\M)$ and for all $\G$-neighboring $D,D'\in\DX$,
    \[
        \Prob\{\M(D)\in S\} \leq \e^{\varepsilon}\Prob\{\M(D')\in S\}+\delta.
    \]
\end{definition}

Different privacy interpretations of $\delta$ can be found in~\cite{dwork2014algorithmic,meiser2018Approximate,canonne2020discrete,desfontaines2020privacy}.
Note that having $\delta\geq1$ is meaningless and provides no privacy since any mechanism, including one that releases the raw data, is $\G$ $(\varepsilon,\delta)$-DP for $\delta\geq1$.  

Our definition of $(d_\D,\delta_\D)$-privacy is formulated so that \Cref{th:Approximated-privacytoDP} verifies, which is analogous to \Cref{th:d-privacyToDP} for the pure-DP case. The construction of $d_\D$-privacy from $\varepsilon$-DP uses the fact that the privacy budget $\varepsilon$ scales linearly with respect to distance $d_\D$. 
\begin{definition}[$(d_{\D},\delta_{\D})$-privacy]\label{def:approximate d-privacy}
    Let $d_\D$ be a metric over $\D$ and $\delta_\D\colon\D^2\to[0,\infty]$. Then, a randomized mechanism $\M$ with domain $\D$ is \textit{$(d_\D,\delta_\D)$-private} if for all $D, D'\in\D$ and all measurable $S\subseteq\Range(\M)$,
    \[
        \Prob\{\M(D)\in S\} \leq \e^{d_{\D}(D,D')}\Prob\{\M(D')\in S\} + \delta_\D(D,D').
    \]
\end{definition}

Analogously to $(\varepsilon,\delta)$-DP, if $\delta_\D(D,D')\geq1$, the indistinguishability (up to $\varepsilon$) between $D,D'\in\D$ is no longer guaranteed. Moreover, we recover $d_{\D}$-privacy when $\delta_\D=0$. 

Note that $\delta_{\D}$ does not need to be a metric. Furthermore, in $(\varepsilon,\delta)$-DP, $\delta$ does not scale linearly under group privacy, but rather ends up as $\delta\frac{\e^{\varepsilon n}-1}{\e^\varepsilon-1}$ (which can be larger than~$1$). Parameter $\delta_\D$ scales in the same way, which is shown in our next result, where we denote as $[d]_\varepsilon\colon\D^2\to[0,\infty]$ the function such that $[d]_\varepsilon(D,D')=\frac{1}{\e^\varepsilon-1}(\e^{\varepsilon d(D,D')}-1)$. 

\begin{theorem}[restate = THapproximatedprivacytoDP, name = ]\label{th:Approximated-privacytoDP}
    Let $\G$ be a granularity notion over the database class $\D$. Then, a mechanism $\M$ with domain $\D$ is $(\varepsilon d^{\G}_{\D},\delta[d^{\G}_{\D}]_{\varepsilon})$-private if and only if it is $\G$ $(\varepsilon,\delta)$-DP.
\end{theorem}

However, please note that $\delta[d^{\G}_{\D}]_\varepsilon$ can scale to numbers greater than $1$. This can lead to weak privacy models since such values result in no privacy, as we said before. For instance with $\delta=10^{-5}$ and $\varepsilon=1$ we have $\delta[d^{\B}_{\D}]_{\varepsilon}(D,D')>1$ for all $D,D'\in\D$ such that $d^{\B}_{\D}(D,D')\geq 13$. Therefore a $( d^{\B}_{\D},10^{-5}[d^{\B}_{\D}]_{1})$-private mechanism can allow an attacker to likely distinguish outputs of two databases in which we have changed more than thirteen records.

We now present the AC result for $(d_\D,\delta_\D)$-privacy.

\begin{theorem}[restate = THapproximateACtheoremvariabledomain, name = Approximate AC theorem]\label{th:ApproximateACTheoremVariableDomain}
    Let $\D$ be a database class, and, for all $i\in[k]$, let $(\D,d_i)$ be a privacy space and $\delta_i\colon\D^2\to[0,\infty]$. Let $f_i\colon\D\to\D_i$ be a deterministic map and $f^*_i = \id_{\overS_i}\times f_i$ (with $f^*_1=f_1$). 
    
    For $i\in[k]$, let $\M^*_i\colon\overS_{i}\times\D_i\to\S_i$ be a mechanism such that $\M^*_i(\overs_{i},\cdot)\colon\D_i\to\S_i$ satisfies $(d_{i},\delta_i)$-privacy for any $\overs_{i}\in\overS_{i}$.
    
    Then mechanism $\M=(\M^*_1\circ f^*_1,\dots,\M^*_k\circ f^*_k)_{\adapt}$ is $(d_{\D},\delta_\D)$-private with
    \begin{align*}
        d_\D(D,D') \coloneqq \sum^k_{i=1} d_i(f_i(D),f_i(D'))\quad \text{and} \quad
        \delta_\D(D,D') \coloneqq \sum^k_{i=1} \delta_i(f_i(D),f_i(D'))\quad \text{for all $D,D'\in\D$}.
    \end{align*}
\end{theorem}

Note that from this result here, we are able to derive all the results so far in this paper. In addition, \Cref{th:ApproximateACTheoremVariableDomain} can be used to define the approximate variations of all our main composition results, where the same consequences can be extracted as in \Cref{sec:IndependentComposition}:

\begin{theorem}[restate = THapproximateASC, name = Generalized approximate ASC]\label{th:ApproximateASC}
    Let $\D$ be a database class, and, for all $i\in[k]$, let $(\D,d_i)$ be a privacy space and $\delta_i\colon\D^2\to[0,\infty]$.
    For $i\in[k]$, let $\M_i\colon\overS_{i}\times\D\to\S_i$ be a mechanism such that $\M_k(\overs_{i},\cdot)\colon\D\to\S_i$ is $(d_i, \delta_i)$-private for any $\overs_{i}\in\overS_{i}$.
    Then $\M=(\M_1,\dots,\M_k)_{\adapt}$ is $(\sum_{i=1}^k d_i,\sum_{i=1}^k \delta_i)$-private.
\end{theorem}
It is important to remark that $\sum_{i=1}^k d_i<\infty$ if and only if $d_i<\infty$, but we can still end up with no privacy guarantee if $\sum_{i=1}^k \delta_i\geq 1$, which can happen even if all $\delta_i<1$. This fact motivates further the search for tighter bounds and the introduction of the approximate counterpart of \Cref{th:GeneralizedIPCVariableDomain}:

\begin{theorem}[restate = THapproximateAPCvariabledomain, name = Approximate best bound for disjoint inputs]\label{th:ApproximateAPCVariableDomain}
    Let $\D$ be a database class and $\G$ a granularity over $\D$. Let $p$ be a $d^\G_\D$-compatible $k$-partitioning function such that
    $\Delta p_i\leq 1$,
    and $p^*_i = \id_{\overS_i}\times p_i$ (with $p^*_1=p_1$). 
    For $i\in[k]$, let $\M^*_i\colon\overS_{i}\times\D_i\to\S_i$ be a mechanism such that $\M^*_i(\overs_{i},\cdot)\colon\D_i\to\S_i$ satisfies $(\varepsilon_i d^\G_{\D_i},\delta_i [d^\G_{\D_i}]_{\varepsilon_i})$-privacy for any $\overs_{i}\in\overS_{i}$.
    Then mechanism $\M=(\M^*_1\circ p^*_1,\dots,\M^*_k\circ p^*_k)_{\adapt}$ is $(\varepsilon d^\G_{\D},\delta [d^\G_\D]_\varepsilon)$-private with $\varepsilon=\max_{i\in[k]} \varepsilon_i$ and $\delta=\max_{i\in[k]} \delta_i$.
\end{theorem}

Furthermore, the common-domain setting imposing $\M^*_i\circ f_i$ to be $(d_i,\delta_i)$-private also leads to interesting composition results for this DP variation:

\begin{theorem}[restate = THapproximateACtheoremcommondomain, name = Approximate AC theorem for common domain]\label{th:ApproximateACTheoremCommonDomain}
    For $i\in[k]$, let $(\D,d_i)$ be a privacy space, $\delta_i\colon\D^2\to[0,\infty]$, and let $f_i$ be a deterministic map over $\D$. 
    For $i\in[k]$, let $\M_i\colon\overS_{i}\times\D\to\S_i$ be a mechanism such that $\M_k(\overs_{i},\cdot)\colon\D\to\S_i$ satisfies $(d_i,\delta_i)$-privacy and $f_i$-dependency for any $\overs_{i}\in\overS_{i}$.
    Then mechanism $\M=(\M_1,\dots,\M_k)_{\ind}$ is $(\sum^k_{i=1} d_i^{f_i},\sum^k_{i=1} \delta_i^{f_i})$-private* with
    \[
        \delta_i^f(D,D') \coloneqq \min_{\substack{\tilde{D},\tilde{D}'\in\D \\ d_i(\tilde{D},\tilde{D}')=d_i^f(D,D')}} \delta_i(\tilde{D},\tilde{D}').
    \]
    
\end{theorem}

\begin{example}
    We continue from \Cref{ex:IntermediateSetting}, where we want to know the number of available ambulances for each hospital. 
    However, we instead consider $\M=(\M_1,\dots,\M_k)_{\ind}$ such that $\M_i=\M^*_a\circ f_i$ are bounded $(1,10^{-5})$-DP (i.e., $(d^{\B}_{\D},\delta_0[d^{\B}_{\D}]_{1})$-private with $\delta_0=10^{-5}$). By construction, for all $i\in[k]$, mechanism $\M_i$ outputs the perturbed number of ambulances linked to $i$ and is $f_i$-dependent (where $f_i(D)$ is the subdatabase of $D\in\D$ of ambulances assigned to hospital $i$).
    
    Applying \Cref{th:ApproximateACTheoremCommonDomain}, we obtain that mechanism $\M$ is $(\sum^k_{i=1} (d^{\B}_{\D})^{f_i},\sum^k_{i=1} (\delta_0[d^{\B}_{\D}]_{1})^{f_i})$-private*. Note that under the bounded metric, we have that $I_f(D,D')\leq 6$ for all $D\neigh_\B D'$. Therefore, we can bound the privacy parameters as follows: $\sum^k_{i=1} (d^{\B}_{\D})^{f_i}(D,D')\leq \sum_{i\in I_f(D,D')} d^{\B}_{\D}(D,D')\leq 6 d^{\B}_{\D}(D,D')$ and, analogously, $\sum^k_{i=1} (\delta_0[d^{\B}_{\D}]_{1})^{f_i}(D,D')\leq \delta_0\sum_{i\in I_f(D,D')} [d^{\B}_{\D}]_{1}(D,D')\leq 6 \delta_0[d^{\B}_{\D}]_{1}(D,D')$. 

    In conclusion, $\M$ is $(6d^{\B}_{\D},6 \delta_0[d^{\B}_{\D}]_{1})$-private (i.e., bounded $(6,6\cdot 10^{-5})$-DP).
\end{example}

The approximate variant of \Cref{th:GeneralizedAPCCommonDomain} can also be enunciated (see \Cref{th:approximateAPCCommonDomain}).

\subsection{Extending to Zero-Concentrated DP}\label{sec:ExtendingTozCDP}
Another common adaptation of DP is \textit{zero-concentrated DP} (zCDP)~\cite{bun2016Concentrated}. This privacy metric is based on a bound on the Rényi divergence:

\begin{definition}[Rényi divergence~\cite{vanerven2014Renyi,bun2016Concentrated}]
    Given two probability distributions $P$ and $Q$ defined over $\S$, the \textit{Rényi divergence of order $\alpha\in(1,\infty)$} is defined as
    \[
        D_{\alpha}(P\dline Q)\coloneqq\frac{1}{\alpha-1}\ln\int_{\S}p^{\alpha}q^{1-\alpha}\,\diff\mu 
    \]
    where $p$ and $q$ are the densities of $P$ and $Q$ with respect to measure $\mu$\footnote{Measure $\mu$ always exists in this case and its choice does not affect the results~\cite{vanerven2014Renyi}.}, respectively. For order $\alpha=\infty$, it is defined as
    \[
        D_{\infty}(P\dline Q) \coloneqq \lim_{\alpha\to\infty} D_{\alpha}(P\dline Q) = \ln\sup_{\text{$S$ meas.}}\frac{P(S)}{Q(S)}     
    \]
\end{definition}

The previous integral notation will be useful to represent both continuous and discrete cases, i.e., if $P$ and $Q$ are continuous, the integral equals $\int_{\S}p(s)^{\alpha}q(s)^{1-\alpha}\,\diff s$ with $p$ and $q$ the corresponding density functions, and if $P$ and $Q$ are discrete, it equals $\sum_{s\in\S}p(s)^{\alpha}q(s)^{1-\alpha}$ with $p$ and $q$ the corresponding probability mass functions.

Note that the Rényi divergence is not a metric for $\alpha\in(1,\infty)$, since it does not satisfy the symmetry property and the triangle inequality. It does, however, satisfy the weaker triangle inequality~\cite{bun2016Concentrated}:
For all probability distributions $P$, $Q$ and~$R$, and all $\alpha,k\geq 1$, we have
\[
    D_{\alpha}(P\dline R)\leq \frac{k\alpha}{k\alpha-1}D_{\frac{k\alpha-1}{k-1}}(P\dline Q)+ D_{k\alpha}(Q\dline R).
\]

In the subsequent results, we denote $D_\alpha(\M(D)\dline\M(D'))$ as the Rényi divergence of the distributions of $\M(D)$ and $\M(D')$. Observe that the case $\alpha=\infty$ can be used to define $d_\D$-privacy (and DP), i.e., $\M$ with domain $\D$ is $d_\D$-private if and only if for all $D,D'\in\D$
\[
    D_{\infty}(\M(D)\dline\M(D'))\leq d(D,D').
\]

We can state now the definition of zero-concentrated DP~\cite{bun2016Concentrated} directly extended for any possible granularity~$\G$. 

\begin{definition}[Zero-concentrated DP] 
    Let $\rho\geq0$. A randomized mechanism $\M$ with domain $\DX$ is $\G$ $\rho$-\textit{zero-concentrated DP} ($\G$ $\rho$-zCDP) if, for all $\G$-neighboring $D,D'\in\D$ and all $\alpha\in(1,\infty)$:
    \[
        D_{\alpha}(\M(D)\dline\M(D'))\leq \rho\alpha.
    \]
\end{definition}

The extension to metric zCDP is not trivial, since the bound of the Rényi divergence does not scale linearly for group privacy, but instead quadratically (i.e., $D_{\alpha}(\M(D)\dline\M(D'))\leq (d^{\G}_\D(D,D'))^2\rho\alpha$). In this case, bounding the divergence by a metric would be too restrictive with regard to the original notion. In particular, known zCDP mechanisms, such as the Gaussian mechanism, would not satisfy a linear privacy degradation. Therefore, knowing that the Rényi divergence scales quadratically, we define the following notion:
\begin{definition}[$d^2_{\D}$-zCprivacy]
    Let $(\D,d_\D)$ be a privacy space. Then, a randomized mechanism $\M$ with domain $\D$ is \textit{$d^2_{\D}$-zCprivacy} if for all $D, D'\in\D$ and all $\alpha\in(1,\infty)$,
    \begin{equation}\label{eq:d-zCprivacy}
        D_{\alpha}(\M(D)\dline\M(D'))\leq {d}^2_{\D}(D,D')\alpha
    \end{equation}
    where $d^2_\D(D,D')\coloneqq(d_\D(D,D'))^2$.
\end{definition}

With this definition, we obtain once again the analogous to \Cref{th:d-privacyToDP} for zCDP:

\begin{theorem}
    [restate = THdzCprivacytoDP, name  = ]\label{th:d-zCprivacyToDP}
    Let $\G$ be a granularity notion over the database class $\D$. Then, a mechanism $\M$ with domain $\D$ is $\rho (d^\G_\D)^2$-zCprivate if and only if it is $\G$ $\rho$-zCDP.
\end{theorem}

We now present the AC theorem for $d^2_\D$-zCprivacy:

\begin{theorem}[restate = THzeroconcentratedactheoremvariabledomain, name = Zero-concentrated AC theorem]\label{th:ZeroConcentratedACTheoremVariableDomain}
    Let $\D$ be a database class, and, for all $i\in[k]$, let $(\D_i,d_{i})$ be a privacy space, and $f_i\colon\D\to\D_i$ a deterministic map and $f^*_i = \id_{\overS_i}\times f_i$ (with $f^*_1=f_1$). 
    
    For $i\in[k]$, let $\M^*_i\colon\overS_{i}\times\D_i\to\S_i$ be a mechanism such that $\M^*_i(\overs_{i},\cdot)\colon\D_i\to\S_i$ satisfies $d^2_{i}$-zCprivacy for any $\overs_{i}\in\overS_{i}$.
    
    Then mechanism $\M=(\M^*_1\circ f^*_1,\dots,\M^*_k\circ f^*_k)_{\adapt}$ is $d^2_{\D}$-zCprivate with
    \[
        d^2_\D(D,D') \coloneqq \sum^k_{i=1} d^2_i(f_i(D),f_i(D')) \quad \text{for all $D,D'\in\D$}.
    \]
\end{theorem}

As in the previous cases, \Cref{th:ZeroConcentratedACTheoremVariableDomain} can be used to formulate the corresponding corollaries. 

\begin{theorem}[restate = THzeroconcentratedACtheoremcommondomain, name = Zero-concentrated AC theorem for common domain]\label{th:ZeroConcentratedACTheoremCommonDomain}
    For $i\in[k]$, let $(\D,d_i)$ be a privacy space, and let $f_i$ be a deterministic map over $\D$.
    For $i\in[k]$, let $\M_i\colon\overS_{i}\times\D\to\S_i$ be a mechanism such that $\M_k(\overs_{i},\cdot)\colon\D\to\S_i$ satisfies $d^2_i$-zCprivacy and $f_i$-dependency for any $\overs_{i}\in\overS_{i}$.
    Then mechanism $\M=(\M_1,\dots,\M_k)_{\adapt}$ is $d^2_\D$-zCprivate* with $d^2_\D\coloneqq\sum^k_{i=1} (d_i^{f_i})^2$.
\end{theorem}

\begin{theorem}[restate = THzeroconcentratedASC, name = Generalized zero-concentrated ASC]\label{th:ZeroConcentratedASC}
    Let $\D$ be a database class, and, for all $i\in[k]$, let $(\D,d_i)$ be a privacy space.
    For $i\in[k]$, let $\M_i\colon\overS_{i}\times\D\to\S_i$ be a mechanism such that $\M_k(\overs_{i},\cdot)\colon\D\to\S_i$ is $d^2_i$-zCprivate for any $\overs_{i}\in\overS_{i}$.
    Then $\M=(\M_1,\dots,\M_k)_{\adapt}$ is $(\sum_{i=1}^k d_i^2)$-zCprivate.
\end{theorem}

When $d_i=\rho_i (d^\U_\D)^2$ we recover the original composition bound $\sum^k_{i=1}\rho_i$ established for unbounded zCDP in~\cite{bun2016Concentrated}, which generalizes to all granularities. However, to the best of our knowledge, no analysis on the privacy loss has previously been performed for zCDP when mechanism $\M_i$ input disjoint subsets. Therefore, we give the two first results about how zCDP degrades when composed, similar to parallel composition:

\begin{theorem}[restate = THzeroconcentratedAPCvariabledomain, name = Zero-concentrated best bound for disjoint inputs]\label{th:ZeroConcentratedAPCVariableDomain}
    Let $\D$ be a database class and $\G$ a granularity over $\D$. Let $p$ be a $d^\G_\D$-compatible $k$-partitioning function such that
    $\Delta p_i\leq 1$,
    and $p^*_i = \id_{\overS_i}\times p_i$ (with $p^*_1=p_1$). For $i\in[k]$, let $\M^*_i\colon\overS_{i}\times\D_i\to\S_i$ be a mechanism such that $\M^*_i(\overs_{i},\cdot)\colon\D_i\to\S_i$ satisfies $\rho_i (d^\G_{\D_i})^2$-zCprivacy for any $\overs_{i}\in\overS_{i}$. Then mechanism $\M=(\M^*_1\circ p^*_1,\dots,\M^*_k\circ p^*_k)_{\adapt}$ is $\rho (d^\G_{\D})^2$-zCprivate with $\rho=\max_{i\in[k]} \rho_i$.
\end{theorem}

For the common-domain setting, we find the analogous theorem (see \Cref{th:ZeroConcentratedAPCCommonDomain}).

\subsection{Extending to Gaussian DP}\label{sec:GDP}
Finally, we extend our results to Gaussian DP (GDP)~\cite{dong2022Gaussian}. GDP uses the hypothesis testing interpretation of DP to bound the privacy loss. This way, we understand that an attacker is trying to solve a hypothesis testing problem for two neighboring databases $D$ and $D'$ as~\cite{dong2022Gaussian}
\[
    \begin{cases}
        H_0\colon \text{The input database is $D$,}\\ 
        H_1\colon \text{The input database is $D'$.}
    \end{cases}
\]

Specifically, given an output $s$, an attacker will use a rejection rule $\phi$ to decide whether $D$ or $D'$ was the initial database. The difficulty in distinguishing between the two hypotheses is then described by the optimal trade-off between the \textit{type~I error} (i.e., rejecting $H_0$ when it is true) and the \textit{type~II error} (i.e., failing to reject $H_0$ when it is false). If $P$ and $Q$ are the distribution functions of $\M(D)$ and $\M(D')$ respectively, then the type~I and type~II errors are defined respectively as $\alpha_\phi \coloneqq \mathbb{E}_P[\phi]$ and $\beta_\phi\coloneqq1-\mathbb{E}_Q[\phi]$, given a rejection rule $0\leq\phi\leq1$. This motivates the definition of trade-off function~\cite{dong2022Gaussian}.

\begin{definition}[Trade-off function~\cite{dong2022Gaussian}]
    Let $P$ and $Q$ be two probability distributions on the same measurable space. A \textit{trade-off function} is defined as $T(P,Q)\colon[0,1]\to[0,1]$	such that
    \[
        T(P, Q)(\alpha) = \inf_{\phi}\{\beta_{\phi}\mid\alpha_{\phi}\leq\alpha\},
    \]
    where the infimum is taken over all (measurable) rejection rules $\phi$.
\end{definition}
A trade-off function $T(P,Q)(\alpha)$ represents the minimum achievable type~II error $\beta$ for a given level of type~I error $\alpha$. Note that the minimum $\beta_\phi$ can be achieved by the likelihood-ratio test, since it is the test with the highest \textit{power} (i.e., lowest type~II error for a prespecified type~I error~$\alpha$) according to the Neyman--Pearson lemma~\cite{lehmann2005Testing}. 
 The larger the trade-off function, the harder it is to distinguish between the two hypotheses. 
This idea of ``hard to distinguish'' leads us to the definition of Gaussian DP (GDP)~\cite{dong2022Gaussian}, which we directly define for any neighborhood notion:
\begin{definition}[Gaussian DP]
    Let $\mu\geq0$. A mechanism $\M$ with domain $\D$ is said to be \textit{$\G$ $\mu$-GDP} if, for all $\G$-neighboring $D,D'\in\D$, 
    \[
        T(\M(D),\M(D'))(\alpha)\geq T(\mathcal{N}(0,1),\mathcal{N}(\mu,1))(\alpha)
    \]
    for all $\alpha\in [0,1]$. We denote $G_\mu\coloneqq T(\mathcal{N}(0,1),\mathcal{N}(\mu,1))$.
\end{definition}
 
First, note that $T(\M(D),\M(D'))$ is the trade-off function of the distribution of $\M(D)$ and $\M(D')$ (by abuse of notation). GDP establishes that distinguishing between $\M(D)$ and $\M(D')$ is at least as hard as distinguishing between the normal distributions $\mathcal{N}(0,1)$ and $\mathcal{N}(\mu,1)$. By the Neyman--Pearson lemma, we can explicitly express $G_{\mu}$ as $G_{\mu}(\alpha) = \Phi( \Phi^{-1}(1- \alpha)-\mu)$ for all $\alpha\in[0,1]$, where $\Phi$ is the distribution function of $\mathcal{N}(0,1)$. Note that this trade-off function decreases with respect to $\mu$, i.e., $G_{\mu}\leq G_{\mu'}$ if $\mu\geq\mu'$. 

GDP satisfies a group privacy property that establishes that privacy degrades linearly with respect to the number of changes between the two databases~\cite{dong2022Gaussian}. Consequently, we use this property to define the $d_\D$-privacy adaptation of GDP:
\begin{definition}[$d_\D$-Gaussian privacy]
    Let $d_\D\colon\D^2\to[0,\infty]$ be a metric. A mechanism $\M$ with domain $\D$ is said to be \textit{$d_\D$-Gprivate} if, for all $D,D'\in\D$, 
    \[
        T(M(D),M(D')) \geq G_{d_\D(D,D')},
    \]
    where $G_\infty(\alpha)\coloneqq\lim_{\mu\to\infty}G_\mu(\alpha)=0$.
\end{definition}

Our definition of $d_\D$-Gprivacy generalizes the original notion of Gaussian DP:

\begin{theorem}[restate = PRdGDPtoGDP, name  = ]\label{th:d-gdpToDP}
    Let $\G$ be a granularity notion over the database class $\D$. Then, a mechanism $\M$ with domain $\D$ is $\mu d^\G_\D$-Gprivate if and only if it is $\G$ $\mu$-GDP.   
\end{theorem}

We can now present the AC theorem for $d_\D$-Gprivacy.

\begin{theorem}[restate = THgaussianactheorem, name = Gaussian AC theorem]\label{th:GaussianACTheorem}
    Let $\D$ be a database class and, for all $i\in[k]$, let $(\D_i,d_{i})$ be a privacy space, and $f_i\colon\D\to\D_i$ a deterministic map and $f^*_i = \id_{\overS_i}\times f_i$ (with $f^*_1=f_1$). 
    
    For $i\in[k]$, let $\M^*_i\colon\overS_{i}\times\D_i\to\S_i$ be a mechanism such that $\M^*_i(\overs_{i},\cdot)\colon\D_i\to\S_i$ satisfies $d_{i}$-Gprivacy for any $\overs_{i}\in\overS_{i}$. 
    Then mechanism $\M=(\M^*_1\circ f^*_1,\dots,\M^*_k\circ f^*_k)_{\adapt}$ is $d_{\D}$-Gprivate with
    \[
        d_\D(D,D') \coloneqq \sqrt{\sum^k_{i=1} d_i(f_i(D),f_i(D'))^2}\quad \text{for all $D,D'\in\D$}.
    \]
\end{theorem}

Note that unlike the AC theorem (\ref{th:ACTheoremVariableDomain}), $d_\D$ is not the sum of the distances (i.e., the $\ell_1$-norm), but actually the sum of the squares of the distances (i.e., the $\ell_2$-norm). Recall that $\norm{(d_1,\dots,d_k)}_2 \leq \norm{(d_1,\dots,d_k)}_1$. In this case, we can notice improvements in GDP to the composition results. We also see the same improvements in the common-domain counterpart.

\begin{theorem}[restate = THgaussianACtheoremcommondomain, name = Gaussian AC theorem for common domain]\label{th:GaussianACTheoremCommonDomain}
    For $i\in[k]$, let $(\D,d_i)$ be a privacy space, and let $f_i$ be a deterministic map over $\D$. 
    For $i\in[k]$, let $\M_i\colon\overS_{i}\times\D\to\S_i$ be a mechanism such that $\M_k(\overs_{i},\cdot)\colon\D\to\S_i$ satisfies $d_i$-Gprivacy and $f_i$-dependency for any $\overs_{i}\in\overS_{i}$.
    Then mechanism $\M=(\M_1,\dots,\M_k)_{\adapt}$ is $d_\D$-Gprivate* with $d_\D\coloneqq\sqrt{\sum^k_{i=1} (d_i^{f_i})^2}$.   
\end{theorem}

As in the previous subsections, we recover the generalized ASC results when $f_i=\id$:

\begin{theorem}[restate = THgaussianASC, name = Generalized Gaussian ASC]\label{th:GaussianASC}
    Let $\D$ be a database class, and $d$ a metric defined in $\D$. For $i\in[k]$, let $\M^*_i\colon\overS_{i}\times\D_i\to\S_i$ be a mechanism such that $\M^*_i(\overs_{i},\cdot)\colon\D\to\S_i$ satisfies $d_i$-DP for any $\overs_{i}\in\overS_{i}$. 
    Then mechanism $\M=(\M^*_1,\dots,\M^*_k)_{\adapt}$ is $d_{\D}$-Gprivate with with $d_\D=\sqrt{d_1^2+\dots+d_k^2}$.
\end{theorem}

Choosing $d_i=\mu_id^{\G}_{\D}$, we obtain from this theorem the already-known~\cite{dong2022Gaussian} sequential bound $\|(\mu_1,\dots,\mu_k)\|_2$.

For $d$-Gprivacy, as for the other notions, it is interesting to find cases where we can obtain better bounds than the sequential one using our result. We explore these cases in the following corollaries. For example, we can also obtain the best bound for when $f$ defines a partitioning function:

\begin{theorem}[restate = THgaussianparallel, name = ]\label{th:GaussianParallel}
      Let $\D$ be a database class, and let $p$ be $k$-partitioning function of $\D$ in $\D_i$ and $p^*_i = \id_{\overS_i}\times p_i$ (with $p^*_1=p_1$). Let $d^*$ be well-defined over $\D$ and $\D_i$. For $i\in[k]$, let $\M^*_i\colon\overS_{i}\times\D_i\to\S_i$ be a mechanism such that $\M^*_i(\overs_{i},\cdot)\colon\D_i\to\S_i$ satisfies $\mu_i d^*_{\D_i}$-Gprivacy. If $d^*$ commutes with $p$ then mechanism $\M=(\M^*_1\circ p^*_1,\dots,\M^*_k\circ p^*_k)_{\adapt}$ is $\Tilde{d}_{\D}$-Gprivate with 
      \begin{equation}\label{eq:GaussianParallel}
        \Tilde{d}_{\D}(D,D')\coloneqq\sqrt{\sum^k_{i=1} (\mu_id^*_{\D_i}(p_i(D),p_i(D')))^2}\leq \max_{i\in[k]} \mu_i d^*_\D(D,D').
      \end{equation}
\end{theorem}

Note that the inequality is in fact an equality when $d^*_{\D_i}(p_i(D),p_i(D'))=0$ for all but one $i\in[k]$. Therefore in some cases, the Gaussian AC theorem (\ref{th:GaussianACTheorem}) can give us a tighter bound than $\max_{i\in[k]}\mu_id^*_\D$. We see this in the following example:
\begin{example}\label{ex:ultra-parallel}
    Let $\D\subseteq\D_{\X}$, let $\D_i=\D_{\X_i}$ where $\{\X_i\}_{i\in[k]}$ defines a partition, and consider $d^\triangle$, which commutes with the previous partition (see \Cref{prop:commutative}).
    If $\M_i\colon\D_i\to\S_i$ are $d^{\triangle}_{\D_i}$-Gprivate, then mechanism $\M=(\M_1,\dots,\M_k)_{\adapt}$ is $\Tilde{d}_\D$-Gprivate with $\Tilde{d}_\D\leq d^{\triangle}_\D$.
    For instance,
    if $D=D'\backslash\{x_i,x_j\}$ with $x_i\in\X_i$ and $x_j\in\X_j$ ($i\neq j$), we have that $d^{\triangle}_\D(D,D')=2$, while
    \begin{gather*}
        \Tilde{d}_\D(D,D')=\sqrt{d^{\triangle}_{\D_i}(p_i(D),p_i(D'))^2+d^{\triangle}_{\D_j}(p_j(D),p_j(D'))^2}\\
        =\sqrt{|\{x_i\}|^2+|\{x_j\}|^2}=\sqrt{1+1}=\sqrt{2}<2.
    \end{gather*}
\end{example}

The Gaussian version of \Cref{th:GeneralizedAPCVariableDomain} also holds. However, in this case, a compatible partition implies $d^\G_{\D_i}(p_i(D),p_i(D'))=0$ for all but one $i\in[k]$, so the inequality in \Cref{eq:GaussianParallel} becomes an equality and the AC theorem does not provide a tighter bound.

\begin{theorem}[restate = THgaussianAPCvariabledomain, name = Gaussian best bound for disjoint inputs]\label{th:GaussianAPCVariableDomain}
    Let $\D$ be a database class and $\G$ a granularity over $\D$. Let $p$ be a $d^\G_\D$-compatible $k$-partitioning function such that
    $\Delta p_i\leq 1$,
    and $p^*_i = \id_{\overS_i}\times p_i$ (with $p^*_1=p_1$). 
    For $i\in[k]$, let $\M^*_i\colon\overS_{i}\times\D_i\to\S_i$ be a mechanism such that $\M^*_i(\overs_{i},\cdot)\colon\D_i\to\S_i$ satisfies $\mu_i d^\G_{\D_i}$-Gprivacy for any $\overs_{i}\in\overS_{i}$.
    Then mechanism $\M=(\M^*_1\circ p^*_1,\dots,\M^*_k\circ p^*_k)_{\adapt}$ is $\mu d^\G_{\D}$-Gprivate with $\mu=\max_{i\in[k]} \mu_i$.
\end{theorem}

The common-domain setting of this theorem for GDP is analogous (see \Cref{th:GaussianAPCCommonDomain}).

\section{Post-Processing and Reciprocal Results}\label{sec:reciprocal}
Finally, we study post-processing in the privacy notions we have introduced that leads to reciprocal results. 
All the $d_\D$-privacy adaptations of DP notions we introduced, as well as $d_\D$-privacy, are robust to post-processing:

\begin{theorem}[restate = THpostprocessing, name = Post-processing]\label{th:Post-Processing}
    The privacy notions of $d_\D$-privacy, $(d_\D,\delta_\D)$-privacy, $d^2_\D$-zCprivacy and $d_\D$-Gprivacy are robust to post-processing. 
\end{theorem}

Moreover, we obtain reciprocal results for the composition theorems for common domain for any privacy notion $\PN$ that is robust to post-processing. More precisely, \Cref{th:ICTheoremCommonDomain} has a reciprocal result.

\begin{theorem}[restate = THReciprocalIndependent, name = Reciprocal to the IC theorem (common domain)]\label{th:reciprocalindependent}
    Let $\PN$ be a privacy notion that is robust to post-processing. 
    For all $i\in[k]$, let $\M_i\colon\D\to\S_i$ be mutually independent randomized mechanisms. Let $\M=(\M_1,\dots,\M_k)_{\ind}$ be a mechanism that satisfies $\PN$. Then $\M_i$ must satisfy $\PN$ for all $i\in[k]$.
\end{theorem}

Even though it is not useful in constructing new mechanisms, this result makes it clear that we cannot obtain a $\PN$ mechanism by independently composing mechanisms that do not satisfy $\PN$, and can serve as a first check to ensure whether a mechanism satisfies $\PN$ or not. For instance, \Cref{ex:BoundedParallelFail} fails because $\M_i=\M^*_i\circ f_i$ do not satisfy $\PN$. Also, for the adaptive case, we have the following result:

\begin{theorem}[restate = THReciprocalAdaptive, name = ``Reciprocal'' to the AC theorem (common domain)]\label{th:reciprocaladaptive}
    Let $\PN$ be a privacy notion that is robust to post-processing. Let $\M_i\colon\overS_{i}\times\D\to\S_i$ for $i\in[k]$ be randomized mechanisms. Let $\M=(\M_1,\dots,\M_k)_{\adapt}$ be a mechanism satisfying $\PN$. Recall that by definition $\M(D)=(\MM_1(D),\dots,\MM_k(D))$ for all $D\in\D$, where $\MM_i(D)$ are defined recursively as $\MM_{i}(D)=\M_{i}(\MM_{i-1}(D),\dots,\MM_1(D),D)$ for $i\in[k]$. Then $\MM_i$ must  satisfy $\PN$ for all $i\in[k]$.
\end{theorem}

Note that this result tells us that all $\MM_i$ satisfy $\PN$, but this is not the exact reciprocal of \Cref{th:ACTheoremCommonDomain}. Given the same hypotheses, it is not necessarily true that $\M_i(\overs_i,\cdot)$ satisfy $\PN$ for all $\overs_{i}\in\overS_i$. 

Furthermore, no result for $\M^*_i$ can be generally stated. For example, in \Cref{re:ReciprocalPreprocessing}, we provide a case where $\M^*_i\circ f_i$ is free-lunch DP while $\M^*_i$ is not.

\section{Conclusions}\label{sec:conclusions}

In this paper, we study the composability properties of DP in the new settings of the literature, including new granularities and data domains.
We show that composability can be defined independently of the neighborhood definition. Our results can be used to directly obtain specific composition rules when new granularity notions (or metrics) are proposed,
without having to prove these same rules for each case. 

Moreover, our IC and AC theorems~(\ref{th:ICTheoremVariableDomain} and~\ref{th:ACTheoremVariableDomain}) are defined for $d_\D$-privacy. 
The notion of $d_\D$-privacy not only generalizes the original DP setting, but also provides more precise information about the protection given.
Therefore, we facilitate the computation of the final privacy guarantee of any composed mechanism over any desired data domain and even under mixed privacy requirements, which was not previously defined. 
In particular, we prove the existence of a significantly better bound to the privacy loss for bounded DP when the composed mechanisms are applied to disjoint databases (\Cref{th:BoundedParallel}).

Besides, we study particularly interesting composition settings in the literature such as the case in which each composed mechanism inputs the whole database or just disjoint subsets, and we compare them with the original sequential and parallel composition results.
Since the original parallel composition theorem~\cite{mcsherry2009privacy} does not generalize to all metrics, we also investigate the additional hypotheses necessary to obtain the best possible privacy loss when we work over a partitioned database. We provide the hypotheses under which we obtain the best bound and conclude that these conditions are easily satisfied for some metrics, such as $d^{\triangle}$; while others metrics only work for specific partitions, such as the bounded metric.

Furthermore, we extend our results to other DP-based privacy notions: namely, approximate DP, zero-concentrated DP, and Gaussian DP. To this end, we present $d_\D$-privacy variants that simultaneously include both the original definition and their group privacy property. Also, we provide the corresponding composition theorems for each of these notions.

Finally, we discuss reciprocal versions of the composition, which can be used to check when a mechanism fails to guarantee DP.

\paragraph*{Future work} In this paper, we limit ourselves to some DP-based notions that can be directly expressed with a metric. Extending our composition theorems to other DP-based semantic privacy notions, such as Rényi DP~\cite{mironov2017Renyi} or $f$-DP~\cite{dong2022Gaussian}, could be interesting future work. Moreover, it will be interesting to explore the advanced composition versions of the presented theorems for such semantic notions that allow advanced composition.

\section*{Acknowledgments}
Javier Parra-Arnau is a ``Ramón y Cajal'' fellow (ref.\ RYC2021-034256-I) funded by the MCIN\slash AEI\slash 10.13039\slash 501100011033 and the EU ``NextGenerationEU''/PRTR. This work was also supported by the COMPROMISE (PID2020-113795RB-C31) and MOBILYTICS (TED2021-129782B-I00) projects, funded by the same two institutions above.
The authors at KIT are supported by the KASTEL Security Research Labs (Topic 46.23 of Helmholtz Association) and EXC 2050/1 `CeTI' (ID 390696704), as well as the BMBF project ``PROPOLIS'' (16KIS1393K).

The authors also thank the inhouse textician at KASTEL Security Research Labs.

\Urlmuskip=0mu plus 1mu\relax
\printbibliography

@inproceedings{bun2016Concentrated,
  title = {Concentrated Differential Privacy: Simplifications, Extensions, and Lower Bounds},
  shorttitle = {Concentrated Differential Privacy},
  author = {Bun, Mark and Steinke, Thomas},
  editor = {Hirt, Martin and Smith, Adam},
  year = {2016},
  series = {Lecture {{Notes}} in {{Computer Science}}},
  pages = {635--658},
  publisher = {{Springer}},
  address = {{Berlin, Heidelberg}},
  doi = {10.1007/978-3-662-53641-4_24},
  isbn = {978-3-662-53641-4},
  langid = {english},
  booktitle = {Proc. {{Theory Cryptogr}}. {{Conf}}. ({{TCC}})}
}

@book{burago2022course,
  title = {A Course in Metric Geometry},
  author = {Burago, Dmitri and Burago, Yuri and Ivanov, Sergei},
  year = {2022},
  volume = {33},
  publisher = {{American Mathematical Society}}
}

@inproceedings{canonne2020discrete,
  title = {The Discrete {{Gaussian}} for Differential Privacy},
  author = {Canonne, Cl{\'e}ment L and Kamath, Gautam and Steinke, Thomas},
  year = {2020},
  volume = {33},
  pages = {15676--15688},
  publisher = {{Curran Associates, Inc.}},
  doi = {10.29012/jpc.784},
  booktitle = {Proc. {{Int}}. {{Conf}}. {{Neural Inform}}. {{Process}}. {{Syst}}. ({{NeurIPS}})}
}

@inproceedings{chatzikokolakis2013Broadening,
  title = {Broadening the Scope of Differential Privacy Using Metrics},
  author = {Chatzikokolakis, Konstantinos and Andr{\'e}s, Miguel E. and Bordenabe, Nicol{\'a}s Emilio and Palamidessi, Catuscia},
  editor = {De Cristofaro, Emiliano and Wright, Matthew},
  year = {2013},
  series = {Lecture {{Notes}} in {{Computer Science}}},
  pages = {82--102},
  publisher = {{Springer}},
  address = {{Berlin, Heidelberg}},
  doi = {10.1007/978-3-642-39077-7_5},
  isbn = {978-3-642-39077-7},
  langid = {english},
  booktitle = {Proc. {{Priv}}. {{Enhanc}}. {{Technol}}. ({{PETS}})}
}

@misc{desfontaines2020privacy,
  title = {The Privacy Loss Random Variable},
  author = {Desfontaines, Damien},
  year = {2020},
  month = mar,
  journal = {Ted is writing things (personal blog)},
  url = {https://desfontain.es/privacy/privacy-loss-random-variable.html},
  urldate = {2023-08-22}
}

@article{desfontaines2020SoK,
  title = {{{SoK}}: Differential Privacies},
  shorttitle = {{{SoK}}},
  author = {Desfontaines, Damien and Pej{\'o}, Bal{\'a}zs},
  year = {2020},
  journal = {Proc. Priv. Enhanc. Technol.},
  volume = {2020},
  number = {2},
  pages = {288--313},
  issn = {2299-0984},
  doi = {10.2478/popets-2020-0028}
}

@misc{dong2019Gaussian,
  title = {Gaussian Differential Privacy},
  author = {Dong, Jinshuo and Roth, Aaron and Su, Weijie J.},
  year = {2019},
  month = may,
  number = {arXiv:1905.02383},
  eprint = {1905.02383},
  primaryclass = {cs, stat},
  publisher = {{arXiv}},
  doi = {10.48550/arXiv.1905.02383},
  archiveprefix = {arxiv}
}

@article{dong2022Gaussian,
  title = {Gaussian Differential Privacy},
  author = {Dong, Jinshuo and Roth, Aaron and Su, Weijie J.},
  year = {2022},
  month = feb,
  journal = {J. Royal Stat. Soc. Ser. B},
  volume = {84},
  number = {1},
  pages = {3--37},
  issn = {1369-7412},
  doi = {10.1111/rssb.12454}
}

@inproceedings{dwork2006Differential,
  title = {Differential Privacy},
  author = {Dwork, Cynthia},
  editor = {Bugliesi, Michele and Preneel, Bart and Sassone, Vladimiro and Wegener, Ingo},
  year = {2006},
  series = {Lecture {{Notes}} in {{Computer Science}}},
  pages = {1--12},
  publisher = {{Springer}},
  address = {{Berlin, Heidelberg}},
  doi = {10.1007/11787006_1},
  isbn = {978-3-540-35908-1},
  langid = {english},
  booktitle = {Proc. {{Int}}. {{Colloq}}. {{Automata}}, {{Lang}}., {{Program}}. ({{ICALP}})}
}

@inproceedings{dwork2006Our,
  title = {Our Data, Ourselves: Privacy via Distributed Noise Generation},
  shorttitle = {Our Data, Ourselves},
  author = {Dwork, Cynthia and Kenthapadi, Krishnaram and McSherry, Frank and Mironov, Ilya and Naor, Moni},
  editor = {Vaudenay, Serge},
  year = {2006},
  series = {Lecture {{Notes}} in {{Computer Science}}},
  pages = {486--503},
  publisher = {{Springer}},
  address = {{Berlin, Heidelberg}},
  doi = {10.1007/11761679_29},
  isbn = {978-3-540-34547-3},
  langid = {english},
  booktitle = {Proc. {{Adv}}. {{Cryptology}} \textendash{} {{Annual Int}}. {{Conf}}. {{Theory Appl}}. {{Cryptogr}}. {{Techniques}} ({{EUROCRYPT}})}
}

@article{dwork2014algorithmic,
  title = {The Algorithmic Foundations of Differential Privacy},
  author = {Dwork, Cynthia and Roth, Aaron},
  year = {2014},
  month = aug,
  journal = {Found. Trends Theor. Comput. Sci.},
  volume = {9},
  number = {3\textendash 4},
  pages = {211--407},
  issn = {1551-305X},
  doi = {10.1561/0400000042}
}

@inproceedings{galli2023Group,
  title = {Group Privacy for Personalized Federated Learning},
  author = {Galli, Filippo and Biswas, Sayan and Jung, Kangsoo and Cucinotta, Tommaso and Palamidessi, Catuscia},
  year = {2023},
  month = apr,
  pages = {252--263},
  doi = {10.5220/0011885000003405},
  isbn = {978-989-758-624-8},
  booktitle = {Proc. {{Int}}. {{Conf}}. {{Inform}}. {{Syst}}. {{Security Priv}}. ({{ICISSP}})}
}

@inproceedings{hay2009Accurate,
  title = {Accurate Estimation of the Degree Distribution of Private Networks},
  author = {Hay, Michael and Li, Chao and Miklau, Gerome and Jensen, David},
  year = {2009},
  month = dec,
  pages = {169--178},
  publisher = {{IEEE}},
  issn = {2374-8486},
  doi = {10.1109/ICDM.2009.11},
  booktitle = {Proc. {{IEEE Int}}. {{Conf}}. {{Data Min}}. ({{ICDM}})}
}

@inproceedings{kifer2011No,
  title = {No Free Lunch in Data Privacy},
  author = {Kifer, Daniel and Machanavajjhala, Ashwin},
  year = {2011},
  month = jun,
  series = {{{SIGMOD}} '11},
  pages = {193--204},
  publisher = {{Association for Computing Machinery}},
  address = {{New York, NY, USA}},
  doi = {10.1145/1989323.1989345},
  isbn = {978-1-4503-0661-4},
  booktitle = {Proc. {{ACM SIGMOD Int}}. {{Conf}}. {{Manage}}. {{Data}} ({{MOD}})}
}

@misc{kifer2020Guidelines,
  title = {Guidelines for Implementing and Auditing Differentially Private Systems},
  author = {Kifer, Daniel and Messing, Solomon and Roth, Aaron and Thakurta, Abhradeep and Zhang, Danfeng},
  year = {2020},
  month = may,
  number = {arXiv:2002.04049},
  eprint = {2002.04049},
  primaryclass = {cs},
  publisher = {{arXiv}},
  doi = {10.48550/arXiv.2002.04049},
  archiveprefix = {arxiv}
}

@book{lehmann2005Testing,
  title = {Testing Statistical Hypotheses},
  author = {Lehmann, E. L. and Romano, Joseph P.},
  year = {2005},
  series = {Springer Texts in Statistics},
  edition = {3rd ed},
  publisher = {{Springer}},
  address = {{New York}},
  isbn = {978-0-387-98864-1},
  langid = {english},
  lccn = {QA277 .L425 2005}
}

@book{levy2002Basic,
  title = {Basic Set Theory},
  author = {L{\'e}vy, Azriel},
  year = {2002},
  publisher = {{Dover Publications}},
  googlebooks = {TCIX3qis9pUC},
  isbn = {978-0-486-42079-0},
  langid = {english}
}

@book{li2016Differential,
  title = {Differential Privacy: From Theory to Practice},
  shorttitle = {Differential {{Privacy}}},
  author = {Li, Ninghui and Lyu, Min and Su, Dong and Yang, Weining},
  year = {2016},
  month = oct,
  series = {Synthesis Lectures on Information Security, Privacy, and Trust},
  volume = {8},
  publisher = {{Morgan \& Claypool}},
  address = {{San Rafael, California}},
  isbn = {1-62705-297-6},
  langid = {english}
}

@inproceedings{mcsherry2009privacy,
  title = {Privacy Integrated Queries},
  author = {McSherry, Frank},
  year = {2009},
  month = jun,
  publisher = {Association for Computing Machinery, Inc.},
  doi = {10.1145/1559845.1559850},
  langid = {american},
  booktitle = {Proc. {{ACM SIGMOD Int}}. {{Conf}}. {{Manage}}. {{Data}} ({{MOD}})}
}

@misc{meiser2018Approximate,
  title = {Approximate and Probabilistic Differential Privacy Definitions},
  author = {Meiser, Sebastian},
  year = {2018},
  number = {2018/277},
  eprint = {2018/277},
  publisher = {{Cryptology ePrint Archive}},
  url = {https://eprint.iacr.org/2018/277},
  urldate = {2023-08-22},
  archiveprefix = {Cryptology ePrint Archive}
}

@inproceedings{mironov2017Renyi,
  title = {R\'enyi Differential Privacy},
  author = {Mironov, Ilya},
  year = {2017},
  month = aug,
  pages = {263--275},
  issn = {2374-8303},
  doi = {10.1109/CSF.2017.11},
  booktitle = {Proc. {{IEEE Comput}}. {{Security Found}}. {{Symp}}. ({{CSF}})}
}

@misc{siegristProbability,
  title = {Probability, Mathematical Statistics, and Stochastic Processes},
  author = {Siegrist, Kyle},
  journal = {LibreTexts},
  url = {https://stats.libretexts.org/Bookshelves/Probability_Theory/Probability_Mathematical_Statistics_and_Stochastic_Processes_(Siegrist)},
  urldate = {2023-04-25},
  langid = {english}
}

@article{soria-comas2016Big,
  title = {Big Data Privacy: Challenges to Privacy Principles and Models},
  shorttitle = {Big {{Data Privacy}}},
  author = {{Soria-Comas}, Jordi and {Domingo-Ferrer}, Josep},
  year = {2016},
  month = mar,
  journal = {Data Sci. Eng.},
  volume = {1},
  number = {1},
  pages = {21--28},
  issn = {2364-1541},
  doi = {10.1007/s41019-015-0001-x},
  langid = {english}
}

@inproceedings{syropoulos2001Mathematics,
  title = {Mathematics of Multisets},
  author = {Syropoulos, Apostolos},
  editor = {Calude, Cristian S. and P{\u a}un, Gheorghe and Rozenberg, Grzegorz and Salomaa, Arto},
  year = {2001},
  series = {Lecture {{Notes}} in {{Computer Science}}},
  pages = {347--358},
  publisher = {{Springer}},
  address = {{Berlin, Heidelberg}},
  doi = {10.1007/3-540-45523-X_17},
  isbn = {978-3-540-45523-3},
  langid = {english},
  booktitle = {Proc. {{Multiset Process}}.}
}

@article{vanerven2014Renyi,
  title = {R\'enyi Divergence and {{Kullback-Leibler}} Divergence},
  author = {{van Erven}, Tim and Harremo{\"e}s, Peter},
  year = {2014},
  month = jul,
  journal = {IEEE Trans. Inform. Theory},
  volume = {60},
  number = {7},
  pages = {3797--3820},
  issn = {1557-9654},
  doi = {10.1109/TIT.2014.2320500}
}




\begin{appendices}

\section{List of Proofs and Additional Theorems}\label{sec:proofs}
In this section, we recompile all the mathematical proofs of the theorems, propositions, and corollaries stated in this paper, as well as additional supporting results.

\subsection*{Proofs of \Cref{sec:generalizingGranularityNotions}: \nameref*{sec:generalizingGranularityNotions}}

\begin{proposition}\label{prop:PrivacySpaceWellDefined}
    Let $\D$ be a database class and $\G$ a granularity notion over $\D$. Then the canonical metric $d^\G_\D$ is a well-defined extended metric.
\end{proposition}
\begin{proof}
    The canonical metric $d^\G_\D\colon\D^2\to[0,\infty]$ is defined as the minimum number of neighboring databases in $\D$ you need to cross to to obtain $D$ from $D'$ (with $d^\G_\D(D,D')=\infty$ if it is not possible). More formally, we define a \textit{relational chain between elements $D,D'\in\D$} as an ordered finite sequence of $D_i\in\D$ such that $D_0\neigh_{\G}D_1\neigh_{\G}\cdots\neigh_{\G}D_n$ with $D=D_0$ and $D'=D_n$, and define $d^\G_\D(D,D')$ as the minimum length of any relation chain connecting $D$ and $D'$ (with $d^\G_\D(D,D')=\infty$ if no chain exists).

    We need to prove that $d^\G_\D$ is a well-defined extended metric. By construction, the image of $d^\G_\D$ is $[0,\infty]$, and $d^\G_\D(D,D')=0$ if and only if $D=D'$. Symmetry also follows from the fact that $\neigh_\G$ is a symmetric relation, i.e., any chain from $D$ to $D'$ can also be seen as a chain from $D'$ to $D$. 
    
    Finally, concatenating the chains gives us the triangle inequality. Let $D,D',D''\in\D$ such that $d^\G_\D(D,D')=m$ and $d^\G_\D(D',D'')=n$. The triangle inequality holds if $n=\infty$ or $m=\infty$, so suppose $n,m<\infty$. Then, by definition, there exists a relational chain of length $n$ connecting $D$ and $D'$, and a relational chain of length $m$ connecting $D'$ and $D''$. Joining the chains at $D'$ gives us a relational chain of length $n+m$. By definition of $d^\G_\D$, we obtain the triangle inequality $d^\G_\D(D,D'')\leq n+m = d^\G_\D(D,D')+d^\G_\D(D',D'')$.
    
    In conclusion, $d^\G_\D$ is a extended metric and $(\D,d^\G_{\D})$ is a metric space.\qedhere
\end{proof}

\THdprivacytoDP*
\begin{proof}
    First, we see that $\varepsilon d^\G_\D$-privacy implies $\G$ $\varepsilon$-DP. Suppose that $\M\colon\D\to\S$ is $\varepsilon d^\G_\D$-private. Then, for any $\G$-neighboring databases $D,D'\in\D$ and any measurable $S\subseteq\S$, we have that
    \[
        \Prob\{\M(D)\in S\} \leq \e^{\varepsilon d^\G_\D(D,D')}\Prob\{\M(D')\in S\}.
    \]
    
    By construction of the canonical metric, $d^\G_\D(D,D')=1$ since $D$ and $D'$ are $\G$-neighboring, and therefore $\M$ is $\G$ $\varepsilon$-DP.
    
    Now we prove the other implication. Suppose $\M\colon\D\to\S$ is $\G$ $\varepsilon$-DP. We want to see that, for all $D,D'\in\D$ and all measurable $S\subseteq\S$,
    \[
        \Prob\{\M(D)\in S\} \leq \e^{\varepsilon d^\G_\D(D,D')}\Prob\{\M(D')\in S\}.
    \]
    
    The result clearly holds if $d^\G_\D(D,D')=\infty$, so suppose $d^\G_\D(D,D')=n<\infty$. Since the distance is finite, there exists $D_0,\dots,D_{n}\in\D$, such that $D=D_0$, $D'=D_n$ and
    \[
        D_0 \neigh_{\G} D_1 \neigh_{\G} \cdots \neigh_{\G} D_{n-1} \neigh_{\G} D_n.
    \]
    
    Since $D_{i-1}$ and $D_i$ are $\G$-neighboring, for all measurable $S\subseteq\S$ and $i\in[n]$ we have that 
    \[
        \Prob\{\M(D_{i-1})\in S\} \leq \e^{\varepsilon}\Prob\{\M(D_i)\in S\}, 
    \]
    and, by applying the inequalities in order, we obtain
    \begin{align*}
        \Prob\{\M(D)\in S\} &\leq \e^{\varepsilon}\Prob\{\M(D_1)\in S\} \\
        &\leq \e^{2\varepsilon}\Prob\{\M(D_2)\in S\} \\
        &\leq \cdots \\
        &\leq \e^{n\varepsilon}\Prob\{\M(D')\in S\} \\
        &= \e^{\varepsilon d^\G_\D(D,D')}\Prob\{\M(D')\in S\}.
    \end{align*}
    
    In conclusion, $\M$ is $\varepsilon d^\G_\D$-private.
\end{proof}

\begin{remark}\label{RE:InducedVsIntrinsic}
    The \textit{induced metric} of $d\colon\D^2\to[0,\infty]$ to a subclass $\D'\subseteq\D$ is defined as the metric $d|_{\D'}$ such that $d|_{\D'}(D,D')=d(D,D')$ for all $D,D'\in\D'$. 
    
    Note that the induced metric of $d^\G_\D$ to the subclass $\D'\subseteq\D$ is not $d^\G_{\D'}$. Mathematically speaking, the $d^\G_\D$ is a \textit{intrinsic metric}~\cite{burago2022course}, i.e., defined as the infimum of the lengths of all paths from the first database to the second. However, the induced metric to $\D'$ is not necessarily the intrinsic metric over $\D'$~\cite{burago2022course}.
    Therefore, the distance between two databases in $\D'\subseteq\D$ can be different over $\D'$ and~$\D$.
    
    As an example, consider the original definition of DP (\Cref{def:firstDP}) in which the privacy space is $(\DX,d^\U_{\DX})$ with the unbounded metric $d^\U_{\DX}(D,D')=|D\triangle D'|$. However, note that $d^\U_{\D}(D,D')\neq|D\triangle D'|$ in general for $\D\subseteq\DX$, e.g., in the class of databases of size $N$, $\D\coloneqq\{D\in\DX\mid |D|=N\}$. Therefore, there exist $\D\subseteq\DX$ such that $d^\U_{\D}\neq d^\triangle_\D$, even though $d^\U_{\DX}=d^\triangle_{\DX}$.
\end{remark}

\begin{proposition}[Relation between granularities]\label{prop:GranRelation}
    Let $d^{\G_1}_\D$ and $d^{\G_2}_\D$ be two canonical metrics of granularities $\G_1$ and $\G_2$, such that 
    \[
        k = \TBD{\G_1}{\G_2}{\D} \coloneqq \max_{\substack{D,D'\in\D\\D\neigh_{\G_2}D'}}d^{\G_1}_\D(D,D')<\infty.
    \]
    
    Then, $d^{\G_1}_\D\leq kd^{\G_2}_\D$.
\end{proposition}
\begin{proof}
    We need to see that $d^{\G_1}_\D(D,D')\leq kd^{\G_2}_\D(D,D')$ for all $D,D'\in\D$. If $d^{\G_2}_\D(D,D')=\infty$, then the result holds, so we consider $d^{\G_2}_\D(D,D')=n<\infty$.
    
    Since the distance is finite, there exists $D_0,\dots,D_{n}\in\D$, such that $D=D_0$, $D'=D_n$ and
    \[
        D_0 \neigh_{\G_2} D_1 \neigh_{\G_2} \cdots \neigh_{\G_2} D_{n-1} \neigh_{\G_2} D_n.
    \]
    
    Since $D_{i-1}$ and $D_i$ are $\G_2$-neighboring, $d^{\G_1}_\D(D_{i-1},D_i)\leq \dist_\D(\G_1,\G_2) = k$. Therefore, applying the triangle inequality with $d^{\G_1}_\D$ over the chain, we obtain
    \begin{gather*}
        d^{\G_1}_\D(D,D')\leq \sum^n_{i=1} d^{\G_1}_\D(D_{i-1},D_i) \leq \sum^n_{i=1} k = kn = kd^{\G_2}_\D(D,D'). \qedhere
    \end{gather*}
\end{proof}

\PRsensitivitycomposition*
\begin{proof}
    The result verifies if $\Delta f=\infty$ or $\Delta g=\infty$. Suppose then that the sensitivities are finite. By definition of sensitivity, we have for all $D,D'\in\D_1$,
    \begin{gather*}
        d_3((g\circ f)(D),(g\circ f)(D')) = d_3(g(f(D)),g(f(D')))
        \leq \Delta g\,d_2(f(D),f(D')) \leq \Delta g\,\Delta f\,d_1(D,D').
    \end{gather*}
    
    Therefore, since the sensitivity $\Delta(g\circ f)$ is defined as the smallest value such that the inequality holds, we have that $\Delta(g\circ f)\leq \Delta f\,\Delta g$.
\end{proof}

\subsection*{Proofs of \Cref{sec:IndependentComposition}: \nameref*{sec:IndependentComposition}}
 
\THictheoremvariabledomain*
\begin{proof}
    Note that $d_\D$ is a well-defined metric since it is the sum of metrics. Thus we only need to check that \Cref{eq:d-privacy} holds. 
    
    Let $D,D'\in\D$. Then, for all measurable $S_i\subseteq\S_i$, $i\in[k]$,
    \begin{align*}
        \Prob\{\M(D)\in(S_1,\dots,S_k) \}
        \overset{\textrm{(i)}}{=}{}& \prod_{i=1}^{k}\Prob\{\M_i(f_i(D))\in S_i\}\\
        \overset{\textrm{(ii)}}{\leq}{}&\prod_{i=1}^{k}\e^{d_i(f_i(D),f_i(D'))}\Prob\{\M_i(f_i(D'))\in S_i\}\\
        \overset{\textrm{(i)}}{=}{}&\e^{d_\D(D,D')}\Prob\{\M(D')\in S\},
    \end{align*}
    where $d_\D(D,D')=\sum^k_{i=1}d_{i}(f_i(D),f_i(D'))$ and
    \begin{enumerate}[(i)]
        \item is direct from the construction of $\M$, since $\M_i$ are mutually independent,
        \item uses the fact that $\M_i$ are $d_i$-private.
    \end{enumerate}
    
    This completes the proof. \qedhere
\end{proof}

\THgeneralizedISC*
\begin{proof}
    Direct from the IC theorem~(\ref{th:ICTheoremVariableDomain}) by taking $\D_i=\D$ and $f_i=\id$.
\end{proof}

\begin{proposition}\label{prop:commutative}
    Let $\D\subseteq\DX$. Any $k$-partitioning function $p$ of \Cref{ex:partition} is $d^\triangle_\D$-compatible.
\end{proposition}
\begin{proof}
    Consider a partition of \Cref{ex:partition} and fix $D,D'\in\D$. Since $d^\triangle(D,D')=|D\triangle D'|$, we need to prove that 
    \begin{align}\label{eq:CommutativityProof}
        \sum^k_{i=1} |p_i(D) \triangle p_i(D')| =\bigg|\bigg(\bigcup_{i=1}^k p_i(D)\bigg) \triangle \bigg(\bigcup_{i=1}^k p_i(D')\bigg)\bigg| \leq |D\triangle D'|.
    \end{align}
    
    We prove first the equality, which corresponds to seeing that 
    \[
        \bigcup_{i=1}^k (p_i(D) \triangle p_i(D')) = \bigg(\bigcup_{i=1}^k p_i(D)\bigg) \triangle \bigg(\bigcup_{i=1}^k p_i(D')\bigg).
    \] 
    
    Since $p_i(D)$ and $p_j(D')$ are always disjoint for all $D,D'\in\D$ and $i\neq j$, it is sufficient to see that 
    \[
        (A\triangle A')\cup (B\triangle B')=(A\cup B)\triangle (A'\cup B')
    \]
    for any arbitrary multisets such that $A\cap (B\cup B')=\varnothing$ and $A'\cap (B\cup B')=\varnothing$ (this case then extends by induction to $k$ pairs of disjoint multisets). 
    We denote by $\mathcal{A}$, $\mathcal{A}'$, $\mathcal{B}$ and $\mathcal{B}'$ the underlying set of the multisets of $A$, $A'$, $B$ and $B'$, respectively. We use the multiset notation $A\coloneqq\scalar{\mathcal{A}, m_A}$ where $m_A(a)$ corresponds to the multiplicity of $a\in\mathcal{A}$ in~$A$. Under this notation we have for arbitrary multisets $A$ and $B$ that
    \begin{enumerate}[(i)]
        \item $A\cup B = \scalar{\mathcal{A}\cup\mathcal{B}, \max\{m_A,m_B\}}$. 
        \item $A\cup B = \scalar{\mathcal{A}\cup\mathcal{B}, m_A+m_B}$ when $A\cap B=\varnothing$. 
        \item $A\cap B = \scalar{\mathcal{A}\cap\mathcal{B}, \min\{m_A,m_B\}}$. 
        \item $A\backslash B = \scalar{\mathcal{A}, m_A-m_B}$ if $B\subseteq A$.
        \item $A\triangle B = \scalar{\mathcal{A}\cup\mathcal{B}, |m_A-m_B|}$ using (i), (iii) and (iv).
    \end{enumerate}
    
    Therefore, we have that
    \begin{align*}
        (A\cup B)\triangle (A'\cup B') \overset{\text{(ii)}}{=}{}& \langle\mathcal{A}\cup\mathcal{B},m_{A}+m_{B}\rangle\triangle\langle\mathcal{A'}\cup\mathcal{B'},m_{A'}+m_{B'}\rangle\\
        \overset{\text{(v)}}{=}{}& \langle\mathcal{A}\cup\mathcal{B}\cup\mathcal{A'}\cup\mathcal{B'},|m_{A}+m_{B}-m_{A'}-m_{B'}|\rangle,
    \end{align*}
    and
    \begin{align*}
        (A\triangle A')\cup (B\triangle B')
        \overset{\text{(v)}}{=}{}& \langle\mathcal{A}\cup\mathcal{A'},|m_{A}-m_{A'}|\rangle\cup\langle\mathcal{B}\cup\mathcal{B'},|m_{B}-m_{B'}|\rangle \\
        \overset{\text{(ii)}}{=}{}& \langle\mathcal{A}\cup\mathcal{A'}\cup\mathcal{B}\cup\mathcal{B'},|m_{A}-m_{A'}|+|m_{B}-m_{B'}|\rangle.
    \end{align*}

    Since $A\cap (B\cup B')=\varnothing$ and $A'\cap (B\cup B')=\varnothing$, we obtain that $|m_{A}+m_{B}-m_{A'}-m_{B'}|=|m_A-m_{A'}|+|m_B-m_{B'}|$. Therefore, 
    \[
        (A\triangle A')\cup (B\triangle B')=(A\cup B)\triangle (A'\cup B'),
    \]
    and by induction we obtain the equality of \Cref{eq:CommutativityProof}.

    The inequality in \Cref{eq:CommutativityProof} follows from the equality we just proved. Take $A=\bigcup_{i=1}^k p_i(D)$, $A'=\bigcup_{i=1}^k p_i(D')$, $B=D\backslash A$ and $B'=D'\backslash A'$, that obviously verify $A\cap B=\varnothing$ and $A'\cap B'=\varnothing$. Therefore, 
    \begin{align*}
        \sum^k_{i=1} |p_i(D) \triangle p_i(D')| + |B\triangle B'| &=\bigg|\bigg(\bigcup_{i=1}^k p_i(D)\bigg) \triangle \bigg(\bigcup_{i=1}^k p_i(D')\bigg)\bigg| + |B\triangle B'| \\ 
        &=|(A \triangle A') \cup (B\triangle B')| \\
        &=|(A \cup B) \triangle (A'\cup B')| \\
        &= |D\triangle D'|.
    \end{align*}

    Then, we obtain the inequality since $|B\triangle B'|\geq 0$. \qedhere
\end{proof}

\THgeneralizedIPCvariabledomain*
\begin{proof}
    From \Cref{th:d-privacyToDP}, it is equivalent to see that $\M$ is $\G$ $\varepsilon$-DP with $\varepsilon=\max_{i\in[k]}\varepsilon_i$, i.e., that for all $\G$-neighboring $D,D'\in\D$ and measurable $S\subseteq\S$,
    \[
        \Prob\{\M(D)\in S\} \leq \e^{\varepsilon}\Prob\{\M(D')\in S\}.
    \]
    
    Applying the IC theorem~(\ref{th:ICTheoremVariableDomain}), we obtain that $\M$ is $d$-private with
    \[
        d(D,D') = \sum^k_{i=1} \varepsilon_i d^\G_{\D_i}(p_i(D),p_i(D')).
    \]
    
    Now suppose that $D,D'\in\D$ are $\G$-neighboring. By definition of $d^\G_\D$-compatibility, there exist $j\in[k]$ such that $p_i(D)=p_i(D')$ for all $i\neq j$. Consequently, for all $i\neq j$, $d^\G_{\D_i}(p_i(D),p_i(D')) = 0$. Moreover, by preprocessing (\Cref{prop:preprocessing}), we have that $d^\G_{\D_j}(p_j(D),p_j(D'))\leq \Delta p_j d^\G_{\D}(D,D') \leq 1$ since $D\neigh_\G D'$ and $\Delta p_j\leq 1$. Therefore,
    \begin{gather*}
        d(D,D') = \sum^k_{i=1} \varepsilon_i d^\G_{\D_i}(p_i(D),p_i(D')) = \varepsilon_j d^\G_{\D_j}(p_j(D),p_j(D'))\leq \varepsilon_j.
    \end{gather*}
    
    Consequently, since $\M$ is $d$-private, for all measurable $S\subseteq\S$,
    \[
        \Prob\{\M(D)\in S\} \leq \e^{\varepsilon_j}\Prob\{\M(D')\in S\}.
    \]
    
    Since $j\in[k]$ depends on the choice of the $\G$-neighboring $D,D'\in\D$, it is sufficient to choose $\varepsilon=\max_{i\in[k]}\varepsilon_i$ to cover all cases. In conclusion, $\M$ is $\G$ $\varepsilon$-DP.
\end{proof}

\begin{proposition}[Compatible partitions for bounded DP]\label{prop:CompatiblePartitionBounded}
    If $p$ is a $k$-partitioning function, $k>1$, of \Cref{ex:partition}, then $p$ is not $d^\B_{\DX}$-compatible.
\end{proposition}
\begin{proof}
    Since the definition of compatible partition applies to any pair of neighboring databases, we just need to prove that there exists a pair of bounded-neighboring $D,D'\in\DX$ such that the condition of $d^\B_{\DX}$-compatibility is not satisfied.
    
    In particular, we take $x_j\in\X_j$ and $x_i\in\X_i$ with $i\neq j$, and we build the two following bounded-neighboring databases: $D=\{x_i,x_i,x_i\}$ and $D'=\{x_i,x_i,x_j\}$. Then, $p_i(D)=D \neq D\backslash\{x_i\}=p_i(D')$ and $p_j(D)=\varnothing\neq \{x_j\}=p_j(D')$. Therefore there exist more than one $r\in[k]$ (particularly two: $i,j$), such that $p_r(D)\neq p_r(D')$, and thus $p_r(D)\not\neigh_\B p_r(D')$. Therefore, $p$ is not $d^\B_{\DX}$-compatible.
\end{proof}

\begin{proposition}[Compatible order-based partitions for bounded DP]\label{prop:OrderBasedCompatiblePartitionBounded}
    Consider a database $D$ with ordered elements, i.e., every element $(n,x)\in D$ consists of a record value $x\in\X$ and an unique identifier $n\in[|D|]$. Let $\DX^{\mathrm{ord}}$ denote class of all such databases.
    
    Let $p$ be a $k$-partitioning function of $\N$, which induces a partition of the elements of $\D\subseteq\DX^{\mathrm{ord}}$ that divides the databases only taking the order into account, i.e., such that $p(n,x)=p(n,y)$ for all $x,y\in\X$. Then $p$ is $d_{\D}^{\B}$-compatible and $\Delta p_i\leq 1$ for all $i\in[k]$.
\end{proposition}
\begin{proof}
    Due to the databases being ordered, two databases $D,D'\in\D$ are bounded neighboring if and only if we obtain one from the other by changing the record with identifier $n\in[|D|]=[|D'|]$.
    
    Let $D,D'\in\D$ be bounded neighboring databases. Since $p(n,x)=p(n,y)$ for all $n\in[k]$, there exists $j\in[k]$ such that $p_i(D)=p_i(D')$ for all $i\neq j$. In conclusion, $p$ is $d_{\D}^{\B}$-compatible.
    
    Moreover, $p_j(D)\triangle p_j(D')=\{(j,x),(j,y)\}$, so in particular $p_j(D)$ and $p_j(D')$ are also bounded neighboring. Therefore, 
    \[
        \Delta p_i \coloneqq \max_{D\neigh_\B D'} d^\B(p_i(D),p_i(D')) \leq 1
    \]
    for all $i\in[k]$. Since it holds independently of the choice of $D$ and $D'$, $\Delta p_i\leq 1$ for all $i\in[k]$. 
\end{proof}

\begin{remark}\label{re:equalRandomElements}
    Let $f$ be a deterministic map with domain $\D$, and let $\M$ with domain $\D$ be an $f$-dependent mechanism. If $f(D)=f(\tilde{D})$ for some $D,\tilde{D}\in\D$, then $\M(D)$ and $\M(\tilde{D})$ are equal random elements, i.e., $\Prob\{\M(D)\in S\}=\Prob\{\M(\tilde{D})\in S\}$ for all measurable $S\subseteq\S$. 
    
    This is because, by definition of $f$-dependency, there exists a mechanism $\M^*$ such that $\M=\M^*\circ f$. Therefore
    \begin{gather*}
        \Prob\{\M(D)\in S\} = \Prob\{\M^*(f(D))\in S\} = \Prob\{\M^*(f(\tilde{D}))\in S\} = \Prob\{\M(\tilde{D})\in S\}
    \end{gather*}
    for all measurable $S\subseteq\S$. 
\end{remark}

\PROPmininumprivacy*
\begin{proof}

    We fix $D,D'\in\D$ and choose $\tilde{D},\tilde{D}'\in\D$ such that $f(D)=f(\tilde{D})$, $f(D')=f(\tilde{D}')$, and $d_\D(\tilde{D},\tilde{D}')$ is minimum. In this case, $d_\D(\tilde{D},\tilde{D}')=d_\D^f(D,D')$. Then, by definition of $d_\D$-privacy and \Cref{re:equalRandomElements},
    \begin{gather*}
        \Prob\{\M(D)\in S\} = \Prob\{\M(\tilde{D})\in S\} 
        \leq \e^{d_\D(\tilde{D},\tilde{D}')} \Prob\{\M(\tilde{D}')\in S\} = \e^{d_\D^f(D,D')}\Prob\{\M(D')\in S\}
    \end{gather*}
    
    Since this holds for all $D,D'\in\D$  for all measurable $S\subseteq\S$, $\M$ is $d_\D^f$-private*.\qedhere 
\end{proof}

\THictheoremcommondomain*
\begin{proof}
    Applying \Cref{prop:MinimumPrivacy}, we obtain that $\M_i$ are $d_i^{f_i}$-private*. Then the result follows from an analogous proof of the IC theorem~(\ref{th:ICTheoremVariableDomain}).
\end{proof}

\THgeneralizedIPCcommondomain*
\begin{proof}
    From \Cref{th:d-privacyToDP}, it is equivalent to see that $\M$ is $\G$ $\varepsilon$-DP with $\varepsilon=\max_{i\in[k]}\varepsilon_i$, i.e., that for all $\G$-neighboring $D,D'\in\D$ and measurable $S\subseteq\S$,
    \[
        \Prob\{\M(D)\in S\} \leq \e^{\varepsilon}\Prob\{\M(D')\in S\}.
    \]
    
    Applying \Cref{th:ICTheoremCommonDomain}, we obtain that $\M$ is $d$-private* with
    \[
        d(D,D') = \sum^k_{i=1} (\varepsilon_i d^\G_\D)^{p_i}(D,D') = \sum^k_{i=1} \varepsilon_i d^{\G,p_i}_\D(D,D').
    \]
    
    Now suppose $D,D'\in\D$ are $\G$-neighboring. By definition of $d^\G_\D$-compatibility, there exist $j\in[k]$ such that $p_i(D)=p_i(D')$ for all $i\neq j$. Consequently, for all $i\neq j$, $d^{\G,p_i}_\D(D,D')\leq d^\G_\D(D,D) = 0$, since we can select $D$ as both $\tilde{D}$ and $\tilde{D}'$ in the definition (see \Cref{prop:MinimumPrivacy}). Therefore, 
    \begin{gather*}
        d(D,D') = \sum^k_{i=1} \varepsilon_i d^{\G,p_i}_\D(D,D') = \varepsilon_j d^{\G,p_j}_\D(D,D') \leq \varepsilon_j d^\G_\D(D,D') \leq \varepsilon_j,
    \end{gather*}
    where the last inequality comes from the fact that $D$ and $D'$ are $\G$-neighboring. Consequently, since $\M$ is $d$-private*, for all measurable $S\subseteq\S$,
    \[
        \Prob\{\M(D)\in S\} \leq \e^{\varepsilon_j}\Prob\{\M(D')\in S\}.
    \]
    
    Since $j\in[k]$ depends on the choice of the $\G$-neighboring $D,D'\in\D$, it is sufficient to choose $\varepsilon=\max_{i\in[k]}\varepsilon_i$ to cover all cases. In conclusion, $\M$ is $\G$ $\varepsilon$-DP.
\end{proof}

\COboundedparallel*
\begin{proof}
    From \Cref{th:d-privacyToDP}, it is equivalent to see that $\M$ is bounded $\varepsilon$-DP with $\varepsilon=\max_{i,j\in[k];\,i\neq j}(\varepsilon_i+\varepsilon_j)$, i.e., that for all bounded-neighboring $D,D'\in\D$ and measurable $S\subseteq\S$,
    \[
        \Prob\{\M(D)\in S\} \leq \e^{\varepsilon}\Prob\{\M(D')\in S\}.
    \]
    
    Applying \Cref{th:ICTheoremCommonDomain}, we obtain that $\M$ is $d$-private* with
    \[
        d(D,D') = \sum^k_{i=1} (\varepsilon_i d^{\B}_\D)^{p_i}(D,D') = \sum^k_{i=1} \varepsilon_i d^{\B,p_i}_\D(D,D').
    \]
    
    Now suppose $D,D'\in\D$ are bounded-neighboring. We know there exists $x\in D$ and $x'\in D'$ such that $D\triangle D'=\{x,x'\}$. Then, we have the following possibilities:
    \begin{enumerate}[(a)]
        \item $x,x'\in\X_j$ for a $j\in[k]$. This implies that $p_i(D)=p_i(D')$ for all $i\neq j$.
        \item $x\in\X_j$ and $x'\in\X_l$ for different $j,l\in[k]$. This implies that $p_i(D)=p_i(D')$ for all $i\neq j,l$.
        \item $x\in\X_j$ for $j\in[k]$ and $x'\not\in\X_l$ for any $l\in[k]$ (or vice-versa). This implies that $p_i(D)=p_i(D')$ for all $i\neq j$.
        \item $x,x'\not\in\X_l$ for any $l\in[k]$. Then $p_i(D)=p_i(D')$ for all $i\in[k]$.
    \end{enumerate} 
    
    In the worst case scenario, there are at most two subindices $j,l\in[k]$ such that $p_i(D)=p_i(D')$ for all $i\neq j,l$. For these subindices, $d^{\B,p_j}_\D(D,D'),d^{\B,p_l}_\D(D,D')\leq d^{\B}_\D(D,D')\leq1$, since $D$ and $D'$ are bounded-neighboring. Therefore, 
    \begin{gather*}
        d(D,D') = \sum^k_{i=1} \varepsilon_i d^{\B,p_i}_\D(D,D') \leq \max_{j,l\in[k];\,j\neq l} (\varepsilon_j+\varepsilon_l) = \varepsilon
    \end{gather*}
    for all bounded-neighboring $D,D'\in\D$. In conclusion, $\M$ is bounded $\varepsilon$-DP since it is $d$-private*.
\end{proof}

\begin{lemma}\label{LE:CorollaryBoundedUnbounded}
    Let $A,B\in\DX$ such that $|A|\leq |B|$ and $d^{\triangle}_{\DX}(A,B)=n$. Let $k=|B|-|A|$. Then, for any $\{x_i\}_{i\in[k]}\in\DX$, $C=A+\{x_i\}_{i\in[k]}$ verifies $d^{\B}_{\DX}(C,B)\leq n$ (where $+$ denotes the sum of multisets).
\end{lemma}
\begin{proof}
    Take $A,B\in\DX$ such that $r=|A|\leq |B|=s$ and $d^{\triangle}_{\DX}(A,B)=|A\triangle B|=n<\infty$. Observe that if $A\cap B = \{b_1,\dots,b_l\}$ with $0\leq l\leq r$, then we can express $A$ and $B$ as
    \begin{align*}
        B &= \{\overbrace{b_1,\dots,b_l}^{A\cap B},\overbrace{b_{l+1},\dots,b_r,b_{r+1},\dots,b_s}^{B\backslash (A\cap B)}\}, 
        \\
        A &= \{\underbrace{b_1,\dots,b_l}_{A\cap B},\underbrace{a_{l+1},\dots,a_r}_{A\backslash(A\cap B)}\}.
    \end{align*}

    In this case, note that 
    \begin{align*}
        A\triangle B = (A\backslash(A\cap B))\cup (B\backslash(A\cap B))
        =  \{b_{l+1},\dots,b_r,b_{r+1},\dots,b_s,a_{l+1},\dots,a_r\},
    \end{align*}
    which has size $n$ by hypothesis.
     
    Consider the case where $|A|=|B|$. Then $|A\backslash(A\cap B)|=|B\backslash(A\cap B)|$ and $n=d^{\triangle}_{\DX}(A,B)=2|A\backslash(A\cap B)|$ is even. 
    In particular, $A\triangle B$ has the same number of elements of $A$ and $B$, and therefore we can obtain $B$ from $A$ in $\frac{n}{2}$ bounded changes ($a_i\to b_i$ for $i\in\{l+1,\dots,r\}$). That is, $d^\B_{\DX}(A,B)=\frac{n}{2}$. Therefore, if $|A|=|B|$, the statement verifies taking $A=C$ (since $k=0$).

    Suppose now $|A|<|B|$ (where $k\coloneqq|B|-|A|$) and define $C=A+\{x_i\}_{i\in[k]}$ for arbitrary $x_i\in\X$, i.e., 
    \[
        C = \{\underbrace{b_1,\dots,b_l}_{A\cap B},\underbrace{a_{l+1},\dots,a_r}_{A\backslash(A\cap B)},x_1,\dots,x_k\}.
    \]

    In particular, $|B|=|C|$, so we can apply the previous case. Thus, we obtain that $m\coloneqq d^\triangle_{\DX}(C,B)$ is even and $d^\B_{\DX}(C,B)=\frac{m}{2}$. Furthermore, 
    \[
        C\triangle B \subseteq (A\triangle B) + \{x_i\}_{i\in[k]},
    \]
    so
    \[
        m = |C\triangle B| \leq |A\triangle B| + k = n+k.
    \]

    Note that $k\leq n$ since
    \begin{gather*}
        n-k = |A\triangle B| - (|B| - |A|) =  |A\backslash(A\cap B)| + |B\backslash(A\cap B)| - |B| + |A| \\
        = |A|-|A\cap B|+|B|-|A\cap B|-|B|+|A|
        = 2|A|-2|A\cap B|\geq 0.
    \end{gather*}

    Thus, $m\leq n+k \leq 2n$. In conclusion, $d^\B_{\DX}(C,B)\leq n$. Since the proof does not depend on the choice of $x_i$, the proof is complete. \qedhere
\end{proof}

\COdependencyunboundedbounded*
\begin{proof}
    For all $D,D'\in\D_{\X}$, we need to prove that 
    \[
        \min_{\substack{\tilde{D},\tilde{D}'\in\D \\ f(\tilde{D})=f(D)\\ f(\tilde{D}')=f(D')}} d^{\B}_{\D_{\X}}(\tilde{D},\tilde{D}')=\min\{d^{\B}_{\D_{\X}},|f(D)\triangle f(D')|\}.
    \]
    
    We have that $d^{\B,f}_{\D_{\X}}\leq d^{\B}_{\D_{\X}}$ by definition of minimum privacy. Therefore, we just need to prove that $d^{\B,f}_{\D_{\X}}\leq|f(D)\triangle f(D')|$ and we obtain the result.

    First, note that since $f(D)=D\cap{\mathcal{Y}}$, $f(f(D))=f(D)$ for all $D\in\DX$. Suppose without lost of generality that $|f(D)|\leq |f(D')|$, and let $k=|f(D')|-|f(D)|$. 
    We take $x\in\mathcal{X}\backslash\mathcal{Y}$ and define $C\coloneqq f(D)+\{x,\overset{(k)}{\ldots}\,,x\}$. We see it verifies $f(D)=f(C)$. Then, $d^{\B,f}_{\D_{\X}}(D,D')\leq d^{\B}_{\D_{\X}}(C,D')\leq |f(D)\triangle f(D')|$ by the definition of minimum privacy and \Cref{LE:CorollaryBoundedUnbounded}. \qedhere
\end{proof}

\subsection*{Proofs of \Cref{sec:GeneralizingAdaptive}: \nameref*{sec:GeneralizingAdaptive}}

\begin{remark}\label{re:probability stuff}
    For the proofs on the adaptive composition, we will need some basic probability results~\cite{siegristProbability}, which we will recompile in this remark. 
    
    Let $\M\colon\D\to\S$. As we mentioned earlier, $\M(D)$ for all $D\in\D$ are random elements (e.g., random variables, continuous or discrete; random vectors; random matrices). 
    Note that for every $\M(D)$ and measurable set $S\subseteq\S$, $P_{\M(D)}(S)=\Prob\{\M(D)\in S\}$ defines a measure. This can also be defined with an integral, i.e., 
    \[
        P_{\M(D)}(S) = \Prob\{\M(D)\in S\} = \int_S \diff{P_{\M(D)}},
    \]
    known as the \textit{Lebesgue--Stieltjes integral}. It can be evaluated over any \mbox{(Lebesgue--Stieltjes-)}integrable function $g\colon\S\to\R$ as $\int_S g\,\diff{P_{\M(D)}}$. This is also denoted as $\int_S g(s)\,\diff{P_{\M(D)}(s)}$, or as $\int_S g\,\diff{F_{\M(D)}}=\int_S g\,\diff{F_{\M(D)}}(s)$ with $F_{\M(D)}$ the distribution function of $\M(D)$. We will use the Lebesgue--Stieltjes integral because it allows us to generalize our results to any random element, such as discrete, continuous, and mixed random variables or random vectors. Specifically, the integral can be written as
    \[
        \int_S g\,\diff{P_{\M(D)}} = \sum_{s\in S} g(s)\Prob\{\M(D)=s\}
    \]
    if $\M(D)$ is a discrete random variable, and as
    \[
        \int_S g\,\diff{P_{\M(D)}} = \int_{S} g(s)\,f_{\M(D)}(s)\,\diff{s}
    \]
    if $\M(D)$ is a continuous random variable with density function $f_{\M(D)}$.
    
    Some of the well-known properties of the integrals that we will use in the proofs are linearity: for any integrable functions $f,g\colon\S\to\R$ and $\alpha,\beta\in\R$, 
    \[
        \int_S (\alpha f + \beta g) \,\diff{P_{\M(D)}} = \alpha\int_S f\,\diff{P_{\M(D)}} + \beta \int_S g\,\diff{P_{\M(D)}},
    \]
    and order: for any integrable functions $f,g\colon\S\to\R$ such that $f \leq g$,
    \[
        \int_S f \,\diff{P_{\M(D)}} \leq \int_S g \,\diff{P_{\M(D)}}.
    \]
    
    Additionally, from the probability properties, we have that
    \[
        \int_{\S} \diff{P_{\M(D)}} = \Prob\{\M(D)\in\S\} = 1,
    \]
    and the \textit{law of total probability}: for any event $A$,
    \[
        \Prob\{A\} = \int_{\S} \Prob\{A\mid \M(D)=s\}\,\diff{P_{\M(D)}(s)}.
    \]
    
    The last result that we will use concerns the sum of measures. Given two measures, $\mu$ and $\nu$, over the same measure space and $a,b\geq0$, we obtain that $a\mu+b\nu$ is also a measure over the same space. Extending to any $a,b\in\R$ gives us that $a\mu+b\nu$ is a signed measure. In either case, we have that
    \[
        \int_S g\,\diff(a\mu + b\nu) = a\int_S g\,\diff\mu + b\int_S g\,\diff\nu 
    \]
    for all measurable $S$ and integrable $g$.
\end{remark}

\begin{lemma}\label{lemma:InequalityIntegral}
    Let $\M\colon\D\to\S$ be a $d$-private mechanism. Then,
    \[
        \int_S g\,\diff{P_{\M(D)}} \leq \e^{d(D,D')}\int_S g\,\diff{P_{\M(D')}}
    \]
    for any integrable function~$g\colon\S\to[0,1]$.
\end{lemma}
\begin{proof}
    Fix $D,D'\in\D$. The result is clear if $d(D,D')=\infty$, so we see the finite case. 
    
    Define the signed measure $\alpha \coloneqq P_{\M(D)}-\e^{d(D,D')} P_{\M(D')}$. Note that $\alpha\leq0$ because $\M$ is $d$-private. Then, we have (see \Cref{re:probability stuff}) that
    \[
        \int_S g\,\diff\alpha = \int_S g\,\diff{P_{\M(D)}} - \e^{d(D,D')}\int_S g\,\diff{P_{\M(D')}},  
    \]
    for all measurable $S\subseteq\S$ and any integrable function~$g\colon\S\to[0,1]$. Since $g\leq 1$, we have that
    \[
        \int_S g\,\diff\alpha \leq \int_S \diff\alpha = \alpha(S) \leq 0,
    \]
    and therefore
    \[
        \int_S g\,\diff{P_{\M(D)}} \leq \e^{d(D,D')}\int_S g\,\diff{P_{\M(D')}}. \qedhere
    \]
\end{proof}

\THactheoremvariabledomain*
\begin{proof}
    Note that $d_\D$ is a well-defined metric since it is the sum of metrics. 
    
    We prove the statement by induction over $k$. The result is trivial for $k=1$, and we consider the case $k=2$.
    
    Denote $\S\coloneqq\S_1\times\S_2$ and fix $D,D'\in\D$. For $i\in[2]$, denote $D_i=f_i(D)$, $D'_i=f_i(D')$, and $d_i\coloneqq d_i(f_i(D),f_i(D'))$ to simplify the notation. By the law of total probability (see \Cref{re:probability stuff}), we have for any measurable $S\subseteq\S$ that   
    \begin{align*}
        \Prob\{\M(D)\in S\}
        ={}& \Prob\{(\M^*_1(D_1),\M^*_2(\M^*_1(D_1),D_2))\in S\} \\
        ={}& \int_{\S_1}\Prob\{(\M^*_1(D_1),\M^*_2(\M^*_1(D_1),D_2))\in S \mid \M^*_1(D_1)=s_1\}\,\diff{P_{\M^*_1(D_1)}(s_1)} \\
        ={}& \int_{\S_1}\Prob\{(s_1,\M^*_2(s_1,D_2))\in S\}\,\diff{P_{\M^*_1(D_1)}(s_1)} \\
        \overset{\textrm{(i)}}{\leq}{}& \e^{d_1}\!\int_{\S_1}\Prob\{(s_1,\M^*_2(s_1,D_2))\in S\}\,\diff{P_{\M^*_1(D'_1)}(s_1)}\\
        \overset{\textrm{(ii)}}{\leq}{}&\e^{d_1}\!\int_{\S_1}\e^{d_2}\Prob\{(s_1,\M^*_2(s_1,D'_2))\in S\}\,\diff{P_{\M^*_1(D'_1)}(s_1)} \\
        ={}&\e^{d_1+d_2}\Prob\{\M(D')\in S\},
    \end{align*}
    where
    \begin{enumerate}[(i)]
        \item uses \Cref{lemma:InequalityIntegral},
        \item uses the fact that $\M^*_2$ is $d_2$-private.
    \end{enumerate}
    
    Taking $d_\D=d_1+d_2$ proves the case $k=2$. Now suppose the statement is true for $k-1$ fixed and we prove it for $k$. Consider the mechanism $\M'=(\M^*_1,\dots,\M^*_{k-1})_{\adapt}$ with domain $\D$. By the induction hypothesis, $\M'$ is $d'_\D$-private with
    \[
        d'_\D(D,D') =
        \sum_{i=1}^{k-1}
        d_i(f_i(D),f_i(D'))
    \]
    for all $D,D'\in\D$. Then, we have that 
    \begin{gather*}
        \M(D)=(\M'(D),\M^*_k(\M'(D),f(D)_k))
    \end{gather*}
    for all $D\in\D$. We can easily check that we are in the conditions of the case $k=2$ by taking $\M'$ as $\M^*_1$ and $\M_k$ as $\M^*_2$. Therefore, we obtain that $\M$ is $d_\D$-private with 
    \begin{gather*}
        d_\D(D,D') =
        \sum_{i=1}^{k}
        d_i(f_i(D),f_i(D'))
    \end{gather*}
    for all $D,D'\in\D$.\qedhere
\end{proof}

\THgeneralizedASC*
\begin{proof}
    Direct from \Cref{th:ACTheoremVariableDomain} by taking $\D_i=\D$ and $f_i=\id$.
\end{proof}

\THgeneralizedAPCvariabledomain*
\begin{proof}
    The proof is analogous to the proof of \Cref{th:GeneralizedIPCVariableDomain} using \Cref{th:ACTheoremVariableDomain} instead.
\end{proof}

\THactheoremcommondomain*
\begin{proof}
    Applying \Cref{prop:MinimumPrivacy}, we obtain that $\M_i$ are $d_i^{f_i}$-private*. Then the result follows from an analogous proof of \Cref{th:ACTheoremVariableDomain}.
\end{proof}

\THgeneralizedAPCcommondomain*
\begin{proof}
    The proof is analogous to the proof of \Cref{th:GeneralizedIPCCommonDomain} using \Cref{th:ACTheoremCommonDomain} instead.
\end{proof}

\subsection*{Proofs of \Cref{sec:ApproximateDP}: \nameref*{sec:ApproximateDP}}

\THapproximatedprivacytoDP*
\begin{proof}
    First, we see that $(\varepsilon d^\G_\D,\delta [d^{\G}_{\D}]_\varepsilon)$-privacy implies $\G$ $(\varepsilon,\delta)$-DP. Suppose that $\M\colon\D\to\S$ is $(\varepsilon d^\G_\D,\delta [d^{\G}_{\D}]_\varepsilon)$-private. Then, for any $\G$-neighboring databases $D,D'\in\D$ and any measurable $S\subseteq\S$, we have that
    \[
        \Prob\{\M(D)\in S\}
        \leq \e^{\varepsilon d^\G_\D(D,D')}\Prob\{\M(D')\in S\} + \delta [d^{\G}_{\D}]_\varepsilon(D,D').
    \]
    
    By construction of the canonical metric, we have that $d^\G_\D(D,D')=1$ and $[d^{\G}_{\D}]_\varepsilon(D,D')=\frac{1}{\e^\varepsilon-1}(\e^{\varepsilon }-1)=1$ since $D$ and $D'$ are $\G$-neighboring. Therefore $\M$ is $\G$ $(\varepsilon,\delta)$-DP.
    
    Now we prove the other implication. Suppose $\M\colon\D\to\S$ is $\G$ $(\varepsilon,\delta)$-DP. We want to see that, for all $D,D'\in\D$ and all measurable $S\subseteq\S$,
    \[
        \Prob\{\M(D)\in S\}
        \leq \e^{\varepsilon d^\G_\D(D,D')}\Prob\{\M(D')\in S\} + \delta [d^{\G}_{\D}]_\varepsilon(D,D').
    \]
    
    The result clearly holds if $d^\G_\D(D,D')=\infty$, so suppose $d^\G_\D(D,D')=n<\infty$. Since the distance is finite, there exists $D_0,\dots,D_{n}\in\D$, such that $D=D_0$, $D'=D_n$ and
    \[
        D_0 \neigh_{\G} D_1 \neigh_{\G} \cdots \neigh_{\G} D_{n-1} \neigh_{\G} D_n.
    \]
    
    Since $D_{i-1}$ and $D_i$ are $\G$-neighboring, for all measurable $S\subseteq\S$ and $i\in[n]$ we have that 
    \[
        \Prob\{\M(D_{i-1})\in S\} \leq \e^{\varepsilon}\Prob\{\M(D_i)\in S\}+\delta, 
    \]
    and, by applying the inequalities in order, we obtain
    \begin{align*}
        \Prob\{\M(D)\in S\}\leq{}& \e^{\varepsilon}\Prob\{\M(D_1)\in S\}+\delta\\
        \leq{}& \e^{\varepsilon}(\e^{\varepsilon}\Prob\{\M(D_2)\in S\}+\delta)+\delta\\
        ={}&\e^{2\varepsilon}\Prob\{\M(D_2)\in S\}+\e^{\varepsilon}\delta+\delta\\
        \leq{}&\cdots\\
        \leq{}&\e^{(n-1)\varepsilon}(\e^{\varepsilon}\Prob\{\M(D_k)\in S\}+\delta)+\sum_{i=0}^{n-2}\e^{i\varepsilon}\delta\\
        ={}&\e^{n\varepsilon}\Prob\{\M(D')\in S\}+\delta\sum_{i=0}^{n-1}\e^{i\varepsilon}.
    \end{align*}
    
    Note that the term $\sum^{n-1}_{i=0}\e^{i\varepsilon}$ is a finite geometric sum, which is known to have value $\frac{\e^{n\varepsilon}-1}{\e^{\varepsilon}-1}$. Substituting $n=d^\G_\D(D,D')$ into the resulting expression gives the desired inequality. In conclusion, $\M$ is $(\varepsilon d^\G_\D,\delta [d^{\G}_{\D}]_\varepsilon)$-private.
\end{proof}

\begin{lemma}\label{lemma:InequalityIntegralApproximate}
    Let $\M\colon\D\to\S$ be a $(d,\delta)$-private mechanism. Then,
    \[
        \int_S g\,\diff{P_{\M(D)}} \leq \e^{d(D,D')}\int_S g\,\diff{P_{\M(D')}} + \delta(D,D')
    \]
    for any integrable function~$g\colon\S\to[0,1]$.
\end{lemma}
\begin{proof}
    The proof is essentially the same as in \Cref{lemma:InequalityIntegral}, but obtaining $\alpha\leq\delta(D,D')$. 
\end{proof}

\THapproximateACtheoremvariabledomain*
\begin{proof}
    Note that $d_\D$ is a well-defined metric since it is the sum of metrics. In addition, $\delta_\D$ is also a well-defined function. 
    
    We prove the statement by induction over $k$. The result is trivial for $k=1$, and we consider case $k=2$.
    
    Denote $\S\coloneqq\S_1\times\S_2$ and fix $D,D'\in\D$. To simplify the notation, for any measurable $S\coloneqq(S_1,S_2)\subseteq\S$, we denote the probabilities by
    \begin{gather*}
        P_1(S_1)\coloneqq\Prob\{\M^*_1(f(D)_1)\in S_1\}, \\ P'_1(S_1)\coloneqq\Prob\{\M^*_1(f(D')_1)\in S_1\}, \\
        P|_{s_1}(S)\coloneqq\Prob\{(s_1,\M^*_2(s_1,f(D)_2))\in S\}, \\ P'|_{s_1}(S)\coloneqq\Prob\{(s_1,\M^*_2(s_1,f(D')_2))\in S\},
    \end{gather*}
    and $d_i\coloneqq d_i(f_i(D),f_i(D'))$ and $\delta_i\coloneqq \delta_i(f_i(D),f_i(D'))$.
    By definition of $d$-privacy and since probabilities are bounded by $1$, we obtain that
    \begin{align*}
        P|_{s_1}(S) \leq \min\{1,\e^{d_2}P'|_{s_1}(S) + \delta_2\} \leq \min\{1,\e^{d_2}P'|_{s_1}(S)\} + \delta_2.
    \end{align*}

    Then, by the law of total probability,
    \begin{align*}
        \Prob\{\M(D)\in S\}
        ={}&\Prob\{(\M^*_1(f_1(D)),\M^*_2(\M^*_1(f_1(D)),f_2(D)))\in S\}\\
        ={}& \int_{\S_1} P|_{s_1}(S)\,\diff{P_{1}(s_1)} \\
        \overset{\textrm{(i)}}{\leq}{}& \int_{\S_1} (\min\{1,\e^{d_2}P'|_{s_1}(S)\} + \delta_2)\,\diff{P_{1}(s_1)}\\
        ={}& \int_{\S_1} \min\{1,\e^{d_2}P'|_{s_1}(S)\}\,\diff{P_{1}(s_1)}  + \delta_2\int_{\S_1}\diff{P_{1}(s_1)}\\
        \overset{\textrm{(ii)}}{=}{}& \int_{\S_1} \min\{1,\e^{d_2}P'|_{s_1}(S)\}\,\diff{P_{1}(s_1)}  + \delta_2\\
        \overset{\textrm{(iii)}}{\leq}{}&\e^{d_1}\!\int_{\S_1}\min\{1,\e^{d_2}P'|_{s_1}(S)\}\,\diff{P'_{1}(s_1)}+\delta_1+\delta_2 \\
        \overset{\textrm{(iv)}}{\leq}{}&\e^{d_1}\!\int_{\S_1}\e^{d_2}P'|_{s_1}(S)\,\diff{P'_{1}(s_1)}+\delta_1+\delta_2 \\
        ={}&\e^{d_1+d_2}\Prob\{\M(D')\in S\} + \delta_1 + \delta_2,
    \end{align*}
    where
    \begin{enumerate}[(i)]
        \item uses the previous inequality,
        \item uses the fact that $\int_{\S_1}\diff{P_{1}(s_1)}=1$ (see \Cref{re:probability stuff}).
        \item uses \Cref{lemma:InequalityIntegralApproximate} since $\min\{1,\e^{d_2}P'|_{s_1}(S)\}\leq1$,
        \item uses the fact that $\min\{1,\e^{d_2}P'|_{s_1}(S)\}\leq\e^{d_2}P'|_{s_1}(S)$.
    \end{enumerate}
    
    This proves the result for $k=2$. Now suppose the statement is true for $k-1$ fixed and we prove it for $k$. Consider the mechanism $\M'=(\M^*_1,\dots,\M^*_{k-1})_{\adapt}$ with domain $\D$. By the induction hypothesis, $\M'$ is $(d'_\D,\delta'_\D)$-private with
    \begin{gather*}
        d'_\D(D,D') =
        \sum_{i=1}^{k-1}
        d_i(f_i(D),f_i(D')) \quad\text{and} \quad
        \delta'_\D(D,D') =
        \sum_{i=1}^{k-1}
        \delta_i(f_i(D),f_i(D')) 
    \end{gather*}
    for all $D,D'\in\D$. Then, we have that 
    \begin{gather*}
        \M(D)=(\M'(D),\M^*_k(\M'(D),f_k(D)))
    \end{gather*}
    for all $D\in\D$. We can easily check that we are in the conditions of the case $k=2$ by taking $\M'$ as $\M^*_1$ and $\M^*_k$ as $\M^*_2$. Therefore, we obtain that $\M$ is $(d_\D,\delta_\D)$-private as in the statement. This completes the proof.\qedhere
\end{proof}

\THapproximateASC*
\begin{proof}
    Direct from \Cref{th:ApproximateACTheoremVariableDomain} by taking $\D_i=\D$ and $f_i=\id$.
\end{proof}

\THapproximateAPCvariabledomain*
\begin{proof}
    From \Cref{th:Approximated-privacytoDP}, it is equivalent to see that $\M$ is $\G$ $(\varepsilon,\delta)$-DP with $\varepsilon=\max_{i\in[k]}\varepsilon_i$ and $\delta=\max_{i\in[k]}\delta_i$, i.e., that for all $\G$-neighboring $D,D'\in\D$ and measurable $S\subseteq\S$,
    \[
        \Prob\{\M(D)\in S\} \leq \e^{\varepsilon}\Prob\{\M(D')\in S\}+\delta.
    \]
    
    Applying \Cref{th:ApproximateACTheoremVariableDomain}, we obtain that $\M$ is $(d,\delta)$-private with
    \[
        d(D,D') = \sum^k_{i=1} \varepsilon_i d^\G_{\D_i}(p_i(D),p_i(D'))
    \quad
    \text{and}
    \quad
        \delta(D,D') = \sum^k_{i=1} \delta_i [d^\G_{\D_i}]_{\varepsilon_i}(p_i(D),p_i(D')).
    \]
    
    Now suppose $D,D'\in\D$ are $\G$-neighboring. By definition of $d^\G_\D$-compatibility, there exist $j\in[k]$ such that $p_i(D)=p_i(D')$ for all $i\neq j$. Consequently, for all $i\neq j$, $d^\G_{\D_i}(p_i(D),p_i(D')) = 0$ and $[d^{\G}_{\D_i}]_{\varepsilon_i}(p_i(D),p_i(D')) = 0$. Moreover, by preprocessing (\Cref{prop:preprocessing}), we have that $d^\G_{\D_j}(p_j(D),p_j(D'))\leq \Delta p_j d^\G_{\D}(D,D') \leq 1$ since $D\neigh_\G D'$ and $\Delta p_j\leq 1$. Hence, $[d^\G_{\D_j}]_{\varepsilon_j}(p_j(D),p_j(D'))\leq 1$ too. Therefore, 
    \begin{gather*}
        d(D,D') = \sum^k_{i=1} \varepsilon_i d^\G_{\D_i}(p_i(D),p_i(D')) = \varepsilon_j d^\G_{\D_j}(p_j(D),p_j(D'))\leq \varepsilon_j,
    \end{gather*}
    and
    \begin{gather*}
        \delta(D,D') = \sum^k_{i=1} \delta_i [d^\G_{\D_i}]_{\varepsilon_i}(p_i(D),p_i(D')) = \delta_j [d^{\G}_{\D_j}]_{\varepsilon_j}(p_j(D),p_j(D'))\leq \delta_j.
    \end{gather*}
    
    Consequently, since $\M$ is $(d,\delta)$-private, for all measurable $S\subseteq\S$,
    \[
        \Prob\{\M(D)\in S\} \leq \e^{\varepsilon_j}\Prob\{\M(D')\in S\}+\delta_j.
    \]
    
    Since $j\in[k]$ depends on the choice of the $\G$-neighboring $D,D'\in\D$, it is sufficient to choose $\varepsilon=\max_{i\in[k]}\varepsilon_i$ and $\delta=\max_{i\in[k]}\delta_i$ to cover all cases. In conclusion, $\M$ is $\G$ $(\varepsilon,\delta)$-DP.
\end{proof}

\begin{proposition}[Minimum privacy for approximate DP]\label{prop:MinimumPrivacyApproximate}
    Let $(\D,d_\D)$ be a privacy space, $\delta_\D\colon\D^2\to[0,\infty]$, let $f$ be a deterministic map with domain $\D$ such that $f(D)\subseteq D$ for all $D\in\D$, and let $\M\colon\D\to\S$ be a $(d_{\D},\delta_\D)$-private mechanism. If $\M$ is $f$-dependent, then $\M$ is $(d_{\D}^{f},\delta^f_\D)$-private* with
    \begin{align*}
        d_{\D}^f(D,D') \coloneqq \min_{\substack{\tilde{D},\tilde{D}'\in\D \\ f(\tilde{D})=f(D)\\ f(\tilde{D}')=f(D')}} d_\D(\tilde{D},\tilde{D}') \quad \text{and} \quad
        \delta_{\D}^f(D,D') \coloneqq \min_{\substack{\tilde{D},\tilde{D}'\in\D \\ d_\D(\tilde{D},\tilde{D}')=d_\D^f(D,D')}} \delta_\D(\tilde{D},\tilde{D}').
    \end{align*}
\end{proposition}
\begin{proof}
    We fix $D,D'\in\D$ and select $\tilde{D},\tilde{D}'\in\D$, such that $f(D)=f(\tilde{D})$, $f(D')=f(\tilde{D}')$, and $d_\D(\tilde{D},\tilde{D}')$ is minimum. If there are multiple $\tilde{D},\tilde{D}'\in\D$ that satisfy the criteria, we choose the one that minimizes $\delta_{\D}(\tilde{D},\tilde{D}')$.
    In this case, $d_\D(\tilde{D},\tilde{D}')=d_\D^f(D,D')$ and $\delta_\D(\tilde{D},\tilde{D}')=\delta_\D^f(D,D')$. Then, by definition of $(d_\D,\delta_\D)$-privacy and \Cref{re:equalRandomElements},
    \begin{align*}
        \Prob\{\M(D)\in S\} &= \Prob\{\M(\tilde{D})\in S\} \\
        &\leq \e^{d_\D(\tilde{D},\tilde{D}')} \Prob\{\M(\tilde{D}')\in S\} + \delta_\D(\tilde{D},\tilde{D}') \\
        &= \e^{d_\D^f(D,D')}\Prob\{\M(D')\in S\} + \delta_\D^f(D,D').
    \end{align*}
    
    Since this holds for all $D,D'\in\D$  for all measurable $S\subseteq\S$, $\M$ is $(d_\D^f,\delta_\D^f)$-private*.\qedhere 
\end{proof}

\THapproximateACtheoremcommondomain*
\begin{proof}
    Applying \Cref{prop:MinimumPrivacyApproximate}, we obtain that $\M_i$ are $(d_i^{f_i},\delta_i^{f_i})$-private*. Then the result follows from an analogous proof of \Cref{th:ApproximateACTheoremVariableDomain}.
\end{proof}

\begin{theorem}[Approximate best bound for disjoint inputs (common domain)]\label{th:approximateAPCCommonDomain}
    Let $\D$ be a database class and $\G$ a granularity over $\D$.
    For $i\in[k]$, let $\M_i\colon\overS_{i}\times\D\to\S_i$ be a mechanism such that $\M_k(\overs_{i},\cdot)\colon\D\to\S_i$ satisfies $(\varepsilon_i d^\G_{\D},\delta_i [d^\G_{\D}]_{\varepsilon_i})$-privacy and $p_i$-dependency for any $\overs_{i}\in\overS_{i}$.
    Then mechanism $\M=(\M_1,\dots,\M_k)_{\adapt}$ is $(\varepsilon d^\G_{\D},\delta [d^\G_\D]_\varepsilon)$-private with $\varepsilon=\max_{i\in[k]} \varepsilon_i$ and $\delta=\max_{i\in[k]} \delta_i$.
\end{theorem}
\begin{proof}
    From \Cref{th:Approximated-privacytoDP}, it is equivalent to see that $\M$ is $\G$ $\varepsilon$-DP with $\varepsilon=\max_{i\in[k]}\varepsilon_i$, i.e., that for all $\G$-neighboring $D,D'\in\D$ and measurable $S\subseteq\S$,
    \[
        \Prob\{\M(D)\in S\} \leq \e^{\varepsilon}\Prob\{\M(D')\in S\}.
    \]
    
    Applying \Cref{th:ApproximateACTheoremCommonDomain}, we obtain that $\M$ is $d$-private* with
    \[
        d(D,D') = \sum^k_{i=1} (\varepsilon_i d^\G_\D)^{p_i}(D,D') = \sum^k_{i=1} \varepsilon_i d^{\G,p_i}_\D(D,D')
    \]
    and
    \[
        \delta(D,D') = \sum^k_{i=1} (\delta_i [d^\G_\D]_{\varepsilon_i})^{p_i}(D,D') = \sum^k_{i=1} \delta_i [d^{\G,p_i}_\D]_\varepsilon(D,D').
    \]
    
    Now suppose $D,D'\in\D$ are $\G$-neighboring. By definition of $d^\G_\D$-compatibility, there exist $j\in[k]$ such that $p_i(D)=p_i(D')$ for all $i\neq j$. Consequently, for all $i\neq j$, $d^{\G,p_i}_\D(D,D')\leq d^\G_\D(D,D) = 0$, since we can select $D$ as both $\tilde{D}$ and $\tilde{D}'$ in the definition (see \Cref{prop:MinimumPrivacyApproximate}). Consequently, $[d^{\G,p_i}_\D]_\varepsilon(D,D')=0$ for all $i\neq j$. Therefore, 
    \begin{gather*}
        d(D,D') = \sum^k_{i=1} \varepsilon_i d^{\G,p_i}_\D(D,D') = \varepsilon_j d^{\G,p_j}_\D(D,D') \leq \varepsilon_j d^\G_\D(D,D') \leq \varepsilon_j,
    \end{gather*}
    and
    \begin{gather*}
        \delta(D,D') = \sum^k_{i=1} \delta_i [d^{\G,p_i}_\D]_\varepsilon(D,D') = \delta_j [d^{\G,p_j}_\D]_\varepsilon(D,D') \leq \delta_j [d^{\G}_\D]_\varepsilon(D,D') \leq \delta_j,
    \end{gather*}
    where the last inequalities in both equations come from the fact that $D$ and $D'$ are $\G$-neighboring. Consequently, since $\M$ is $(d,\delta)$-private*, for all measurable $S\subseteq\S$,
    \[
        \Prob\{\M(D)\in S\} \leq \e^{\varepsilon_j}\Prob\{\M(D')\in S\}+\delta.
    \]
    
    Since $j\in[k]$ depends on the choice of the $\G$-neighboring $D,D'\in\D$, it is sufficient to select $\varepsilon=\max_{i\in[k]}\varepsilon_i$ and $\delta=\max_{i\in[k]}\delta_i$ to cover all cases. In conclusion, $\M$ is $\G$ $(\varepsilon,\delta)$-DP.
\end{proof}


\subsection*{Proofs of \Cref{sec:ExtendingTozCDP}: \nameref*{sec:ExtendingTozCDP}}

\THdzCprivacytoDP*
\begin{proof}
    First, we see that $\rho (d^\G_\D)^2$-zCprivacy implies $\G$ $\rho$-zCDP. Suppose that $\M\colon\D\to\S$ is $\rho (d^\G_\D)^2$-zCprivate. Then, for any $\G$-neighboring databases $D,D'\in\D$ and all $\alpha>1$, we have that
    \[
        D_{\alpha}(\M(D)\dline\M(D'))\leq \rho (d_{\D}^{\G})^2(D,D')\alpha.
    \]
    
    By construction of the canonical metric, $(d^\G_\D)^2(D,D')=1$ since $D$ and $D'$ are $\G$-neighboring, and therefore $\M$ is $\G$ $\rho$-zCDP.
    
    Now we prove the other implication. Suppose $\M\colon\D\to\S$ is $\G$ $\rho$-zCDP. We want to see that, for all $D,D'\in\D$ and all $\alpha>1$,
    \[
        D_{\alpha}(\M(D)\dline\M(D'))\leq \rho (d_{\D}^{\G})^2(D,D')\alpha.
    \]
    
    For this proof, we use the weak triangle inequality satisfied by the Rényi divergence:
    \[
        D_{\alpha}(P\dline Q)\leq \frac{k\alpha}{k\alpha-1}D_{\frac{k\alpha-1}{k-1}}(P\dline R)+ D_{k\alpha}(R\dline Q).
    \]
    
    We present the proof by induction on $d^{\G}_{\D}(D,D')$.
    For $d^{\G}_{\D}(D,D')=1$ we have $D\neigh_{\G}D'$ and therefore
    \[
        D_{\alpha}(\M(D)\dline\M(D'))\leq\rho\alpha= \rho\cdot 1^2\cdot\alpha.
    \]
    
     We now assume the statement holds for $d^{\G}_{\D}(D,D')=k-1$ and we prove for $d^{\G}_{\D}(D,D')=k$. Since $d^{\G}_{\D}(D,D')=k$, there exists $D_1,\dots,D_{k-1}\in\D$ such that
     \[
        D\neigh_{\G}D_1\neigh_{\G}\cdots\neigh_{\G}D_{k-1}\neigh_{\G}D'.
     \]
     
     Then, by induction hypothesis we have that $D_{\alpha}(\M(D)\dline\M(D_{k-1}))\leq\rho(k-1)^2\alpha$ and $D_{\alpha}(\M(D_{k-1})\dline\M(D'))\leq\rho\alpha$ for all $\alpha>1$. Applying now the weak triangle inequality,
     we have that for all $\alpha>1$:
     \begin{align*}
        D_{\alpha}(\M(D)\dline \M(D'))\leq{}& \frac{k\alpha}{k\alpha-1}D_{\frac{k\alpha-1}{k-1}}(\M(D)\dline \M(D_{k-1}))+D_{k\alpha}(\M(D_{k-1})\dline\M(D'))\\
        \leq{}& \frac{k\alpha}{k\alpha-1}\bigg(\rho (k-1)^2\frac{k\alpha-1}{k-1} \bigg)+\rho k\alpha\\
        ={}& \rho (k\alpha(k-1)+k\alpha) \\ 
        ={}& \rho(k^2\alpha-k\alpha+k\alpha)\\
        ={}& \rho k^2\alpha.
     \end{align*}
     
     This proves the result for all $d^\G_\D(D,D')\in\N$. Note that the case $d^\G_\D(D,D')=\infty$ holds trivially. \qedhere
\end{proof}

\THzeroconcentratedactheoremvariabledomain*
\begin{proof}
   Note that $d_\D$ is a well-defined metric since the square root of the sum of squared distances is still a distance (i.e., the $\ell_2$-norm) and
    \[
        d_\D(D,D') = \sqrt{\sum^k_{i=1} d^2_i(f_i(D),f_i(D'))}
    \]
    for all $D,D'\in\D$.
    
    Now, we need to prove, for all $\alpha\in(1,\infty)$ and for all $D,D'\in\D$, that
    \[
        D_{\alpha}(\M(D)\dline\M(D'))
        \leq {d}^2_{\D}(D,D')\alpha\coloneqq \sum^k_{i=1} d^2_i(f_i(D),f_i(D'))\alpha.
    \]
    
    We prove the result by induction over $k$. For $k=1$, the result is trivial. Therefore, fixing $k$, we suppose it is true for $k-1$ and we prove for $k$.
    
    Let $\overM=(\M^*_1\circ f^*_1, \dots, \M^*_{k-1}\circ f^*_{k-1})_\adapt$. By the induction hypothesis, for all $\alpha\in(1,\infty)$ and $D,D'\in\D$,
    \[
        D_\alpha(\overM(D)\dline\overM(D')) \leq \sum^{k-1}_{i=1} d^2_i(f_i(D),f_i(D'))\alpha.
    \]
    
    We can also rewrite $\M$ as a function of $\overM$ and $\M_k$ as
    \[
        \M(D) = (\overM(D),\M_k(\overM(D),f_k(D)))
    \]
    for all $D\in\D$.
    
    We fix $D,D'\in\D$. We denote the distributions of $\M(D)$, $\M(D')$, $\overM(D)$ and $\overM(D')$ respectively as $P$, $Q$, $\myol{P}$ and $\myol[2]{Q}$. For all $\overs\in\overS_{k}=\Range(\overM)$, we denote the distributions of $\M^*_k(\overs,f_k(D))$ and $\M^*_k(\overs,f_k(D'))$ respectively as $P|_{\overs}$ and $Q|_{\overs}$, which can be seen as conditioned distributions to $\overs$. We denote the corresponding density/probability mass functions with lowercase letters: $p$, $q$, $\myol[1]{p}$, $\myol[1]{q}$, $p|_{\overs}$ and $q|_{\overs}$. Let $\mu=(\myol{\mu},\mu_k)$ be the corresponding measure. Furthermore, we denote $d^2_k\coloneqq d^2_k(f_k(D),f_k(D'))$ and $\myol{d}^2 \coloneqq \sum^{k-1}_{i=1} d_i^2(f_i(D),f_i(D'))$ to simplify the notation.
    
    By definition of $\M$, we have for $\overs\in\overS_{k}$ and $s_k\in\S_k$, 
    \[
        p(\overs,s_k) = p|_{\overs}(s_k)\overp(\overs) \quad\text{and}\quad q(\overs,s_k) = q|_{\overs}(s_k)\overq(\overs).
    \]
    
    Then, for any measurable $S=(\myol{S},S_k)\subseteq\overS\times\S_k=\Range(\M)$, we obtain that
    \begin{align*}
        \exp((\alpha-1)D_{\alpha}(\M(D)\dline\M(D')))
        \overset{\textrm{(i)}}{=}{}& \int_{S}p(\overs,s_k)^{\alpha}q(\overs,s_k)^{1-\alpha}\,\diff{\mu(\overs,s_k)} \\
        ={}& \int_{S} \big(p|_{\overs}(s_k)\overp(\overs)\big)^{\alpha}\big(q|_{\overs}(s_k)\overq(\overs)\big)^{1-\alpha}\,\diff{\mu(\overs,s_k)} \\
        ={}& \int_{S} p|_{\overs}(s_k)^\alpha q|_{\overs}(s_k)^{1-\alpha}\overp(\overs)^\alpha\overq(\overs)^{1-\alpha}\,\diff{\mu(\overs,s_k)} \\
        ={}& \int_{\myol{S}}\overp(\overs)^\alpha\overq(\overs)^{1-\alpha}\bigg[\int_{S_k} p|_{\overs}(s_k)^\alpha q|_{\overs}(s_k)^{1-\alpha}\,\diff{\mu_k(s_k)}\bigg]\diff{\myol{\mu}(\overs)} \\
        \overset{\textrm{(i)}}{=}{}& \int_{\myol{S}}\overp(\overs)^\alpha\overq(\overs)^{1-\alpha}\exp((\alpha-1)D_\alpha(P|_{\overs}\dline Q|_{\overs}))\,\diff{\myol{\mu}(\overs)} \\
        \overset{\textrm{(ii)}}{\leq}{}& \int_{\myol{S}}\overp(\overs)^\alpha\overq(\overs)^{1-\alpha}\exp((\alpha-1)\alpha d_k^2)\,\diff{\myol{\mu}(\overs)} \\
        ={}& \exp((\alpha-1)\alpha d^2_k)\int_{\myol{S}}\overp(\overs)^\alpha\overq(\overs)^{1-\alpha}\,\diff{\myol{\mu}(\overs)} \\
        \overset{\textrm{(iii)}}{\leq}{}& \exp((\alpha-1)\alpha d_k^2)\exp((\alpha-1)\alpha\myol{d}^2)\\
        ={}& \exp((\alpha-1)\alpha(\myol{d}^2+d_k^2))\\
        ={}& \exp\!\bigg((\alpha-1)\alpha\sum^k_{i=1}d_i^2(f_i(D),f_i(D'))\bigg)
    \end{align*}
    where
    \begin{enumerate}[(i)]
        \item uses the definition of Rényi divergence,
        \item uses the fact that $\M_k(\overs,\cdot)$ is $d^2_k$-zCprivate,
        \item uses the definition of Rényi divergence and the fact that $\overM$ is $\myol{d}^2$-zCprivate.
    \end{enumerate}
    
    Therefore, by solving the equation, we obtain
    \[
        D_{\alpha}(\M(D)\dline\M(D')) \leq \sum^k_{i=1} d_i^2(f_i(D),f_i(D'))\alpha,
    \]
    completing the proof. \qedhere
\end{proof}

\begin{proposition}[Minimum privacy for zero-concentrated DP]\label{prop:MinimumPrivacyZeroConcentrated}
    Let $(\D,d_\D)$ be a privacy space, let $f$ be a deterministic map with domain $\D$ such that $f(D)\subseteq D$ for all $D\in\D$, and let $\M\colon\D\to\S$ be a $d_{\D}^2$-zCprivate mechanism. If $\M$ is $f$-dependent, then $\M$ is $(d_{\D}^{f})^2$-zCprivate* with
    \[
        (d_{\D}^f)^2(D,D') \coloneqq \min_{\substack{\tilde{D},\tilde{D}'\in\D \\ f(\tilde{D})=f(D)\\ f(\tilde{D}')=f(D')}} d^2_\D(\tilde{D},\tilde{D}').
    \]
\end{proposition}
\begin{proof}
    
    We fix $D,D'\in\D$ and select $\tilde{D},\tilde{D}'\in\D$, such that $f(D)=f(\tilde{D})$, $f(D')=f(\tilde{D}')$, and $d_\D(\tilde{D},\tilde{D}')$ is minimum.
    In this case, $d_\D(\tilde{D},\tilde{D}')=d_\D^f(D,D')$. Then, by definition of $d^2_\D$-zCprivacy and \Cref{re:equalRandomElements},
    \begin{gather*}
        D_\alpha(\M(D)\dline\M(D')) 
        = D_\alpha(\M(\tilde{D})\dline\M(\tilde{D}')) 
        \leq d^2_\D(\tilde{D},\tilde{D}') = (d^f_\D)^2(D,D'). 
    \end{gather*}
    
    Since this holds for all $D,D'\in\D$, $\M$ is $(d_\D^f)^2$-zCprivate*.\qedhere 
\end{proof}

\THzeroconcentratedACtheoremcommondomain*
\begin{proof}
    Applying \Cref{prop:MinimumPrivacyZeroConcentrated}, we obtain that $\M_i$ are $(d_i^{f_i})^2$-zCprivate*. Then the result follows from an analogous proof of \Cref{th:ZeroConcentratedACTheoremVariableDomain}.
\end{proof}

\THzeroconcentratedASC*
\begin{proof}
    Direct from \Cref{th:ZeroConcentratedACTheoremVariableDomain} by taking $\D_i=\D$ and $f_i=\id$.
\end{proof}

\THzeroconcentratedAPCvariabledomain*
\begin{proof}
    From \Cref{th:d-zCprivacyToDP}, it is equivalent to see that $\M$ is $\G$ $\rho$-zCDP with $\rho=\max_{i\in[k]}\rho_i$, i.e., that for all $\G$-neighboring $D,D'\in\D$,
    \[
        D_\alpha(\M(D)\dline\M(D')) \leq \rho\alpha.
    \]
    
    Applying \Cref{th:ZeroConcentratedACTheoremVariableDomain}, we obtain that $\M$ is $d^2$-zCprivate with
    \[
        d^2(D,D') = \sum^k_{i=1} \rho_i (d^\G_{\D_i})^2(p_i(D),p_i(D')).
    \]
    
    Now suppose $D,D'\in\D$ are $\G$-neighboring. By definition of $d^\G_\D$-compatibility, there exist $j\in[k]$ such that $p_i(D)=p_i(D')$ for all $i\neq j$. Consequently, for all $i\neq j$, $d^\G_{\D_i}(p_i(D),p_i(D')) = 0$. Moreover, by preprocessing (\Cref{prop:preprocessing}), we have that $d^\G_{\D_j}(p_j(D),p_j(D'))\leq \Delta p_j d^\G_{\D}(D,D') \leq 1$ since $D\neigh_\G D'$ and $\Delta p_j\leq 1$. Therefore, 
    \begin{gather*}
        d^2(D,D') = \sum^k_{i=1} \rho_i (d^\G_{\D_i})^2(p_i(D),p_i(D')) = \rho_j (d^\G_{\D_j})^2(p_j(D),p_j(D'))\leq \rho_j.
    \end{gather*}
    
    Consequently, since $\M$ is $d^2$-zCprivate, 
    \[
        D_\alpha(\M(D)\dline\M(D')) \leq \rho_j\alpha.
    \]
    
    Since $j\in[k]$ depends on the choice of the $\G$-neighboring $D,D'\in\D$, it is sufficient to choose $\rho=\max_{i\in[k]}\rho_i$ to cover all cases. In conclusion, $\M$ is $\G$ $\rho$-zCDP.
\end{proof}

\begin{theorem}[Zero-concentrated best bound for disjoint inputs (common domain)]\label{th:ZeroConcentratedAPCCommonDomain}
    Let $\D$ be a database class and $\G$ a granularity over $\D$.
    For $i\in[k]$, let $\M_i\colon\overS_{i}\times\D\to\S_i$ be a mechanism such that $\M_k(\overs_{i},\cdot)\colon\D\to\S_i$ satisfies $\rho_i(d^\G_{\D})^2$-zCprivacy and $p_i$-dependency for any $\overs_{i}\in\overS_{i}$.
    Then mechanism $\M=(\M_1,\dots,\M_k)_{\adapt}$ is $\rho (d^\G_{\D})^2$-zCprivate with $\rho=\max_{i\in[k]} \rho_i$.  
\end{theorem}
\begin{proof}
    From \Cref{th:d-zCprivacyToDP}, it is equivalent to see that $\M$ is $\G$ $\rho$-zCDP with $\rho=\max_{i\in[k]}\rho_i$, i.e., that for all $\G$-neighboring $D,D'\in\D$ and measurable $S\subseteq\S$,
    \[
        D_\alpha(\M(D)\dline\M(D')) \leq \rho\alpha.
    \]
    
    Applying \Cref{th:ZeroConcentratedACTheoremCommonDomain}, we obtain that $\M$ is $d$-private* with
    \begin{align*}
        d^2(D,D') = \sum^k_{i=1} (\rho_i (d^\G_\D)^2)^{p_i}(D,D')
        = \sum^k_{i=1} \rho_i (d^{\G,p_i}_\D)^2(D,D').
    \end{align*}
    
    Now suppose $D,D'\in\D$ are $\G$-neighboring. By definition of $d^\G_\D$-compatibility, there exist $j\in[k]$ such that $p_i(D)=p_i(D')$ for all $i\neq j$. Consequently, for all $i\neq j$, $d^{\G,p_i}_\D(D,D')\leq d^\G_\D(D,D) = 0$, since we can select $D$ as both $\tilde{D}$ and $\tilde{D}'$ in the definition (see \Cref{prop:MinimumPrivacyZeroConcentrated}). Therefore, 
    \begin{gather*}
        d^2(D,D') = \sum^k_{i=1} \rho_i (d^{\G,p_i}_\D)^2(D,D')) = \rho_j (d^{\G,p_j}_\D)^2(D,D') \leq \rho_j (d^\G_\D)^2(D,D') \leq \rho_j,
    \end{gather*}
    where the last inequality comes from the fact that $D$ and $D'$ are $\G$-neighboring. Consequently, since $\M$ is $d^2$-zCprivate*, for all measurable $S\subseteq\S$,
    \[
        D_\alpha(\M(D)\dline\M(D')) \leq \rho_j\alpha.
    \]
    
    Since $j\in[k]$ depends on the choice of the $\G$-neighboring $D,D'\in\D$, it is sufficient to choose $\rho=\max_{i\in[k]}\rho_i$ to cover all cases. In conclusion, $\M$ is $\G$ $\rho$-zCDP.
\end{proof}


\subsection*{Proofs of \Cref{sec:GDP}: \nameref*{sec:GDP}}
\PRdGDPtoGDP*
\begin{proof}
    First, we see that $\mu d^\G_\D$-Gprivacy implies $\G$ $\mu$-GDP. Suppose that $\M\colon\D\to\S$ is $\mu d^\G_\D$-Gprivacy. Then, for any $\G$-neighboring databases $D,D'\in\D$, we have that
    \[
        T(\M(D),\M(D')) \geq G_{\mu d_{\D}^{\G}(D,D')}.
    \]
    
    By construction of the canonical metric, $d^\G_\D(D,D')=1$ since $D$ and $D'$ are $\G$-neighboring, and therefore $\M$ is $\G$ $\mu$-GDP.
    
    Now we prove the other implication. Suppose $\M\colon\D\to\S$ is $\G$ $\mu$-GDP. We want to see that for all $D,D'\in\D$
    \[
        T(\M(D),\M(D')) \geq G_{\mu d_{\D}^{\G}(D,D')}.
    \]
    
    We now prove this by induction over $d^{\G}_{\D}(D,D')$.
    For $d^{\G}_{\D}(D,D')=1$ we have $D\neigh_{\G}D'$ and thus
    \[
        T(\M(D),\M(D')) \geq G_{\mu}.
    \]
     
     We now assume that the statement holds for $d^{\G}_{\D}(D,D')=k-1$ and we prove for $d^{\G}_{\D}(D,D')=k$. Since $d^{\G}_{\D}(D,D')=k$, there exists $D_1,\dots,D_{k-1}\in\D$ such that
     \[
        D\neigh_{\G}D_1\neigh_{\G}\cdots\neigh_{\G}D_{k-1}\neigh_{\G}D'.
     \]
     
     By the induction hypothesis, we have that
     \[
        T(\M(D_{k-1}),\M(D')) \geq G_{\mu}
    \]
    and
    \[
        T(\M(D),\M(D_{k-1}))\geq G_{\mu(k-1)}.
    \]
    
    Then, by Lemma A.5 in~\cite{dong2019Gaussian}, we have that
    \[
     T(\M(D),\M(D')) \geq G_{\mu}(1-G_{\mu(k-1)}(\alpha)).
    \]
    
    Therefore, in conclusion, 
    \begin{align*}
        G_{\mu}(1-G_{\mu(k-1)}(\alpha)) &= \Phi(\Phi^{-1}(G_{\mu(k-1)}(\alpha)) - \mu) \\ &=\Phi(\Phi^{-1}(1-\alpha) -\mu-(1-k)\mu) \\
        &= G_{\mu+\mu(k-1)}(\alpha) \\
        &=G_{\mu k}(\alpha).
    \end{align*}
    
    This proves the result for all $d^\G_\D(D,D')\in\N$. Note that the case $d^\G_\D(D,D')=\infty$ holds trivially since $G_\infty \equiv 0$. \qedhere
\end{proof}

\begin{remark}[Tensor product~\cite{dong2022Gaussian}]\label{re:Tensor}
    The \textit{tensor product} of two trade-off functions, $f=T(P,Q)$ and $g=T(P',Q')$, is defined as
    \[
        f \otimes g \coloneqq T(P\times P',Q\times Q').
    \]
    
    The well-definition and the properties of the tensor product are proven in~\cite{dong2019Gaussian}. In our proofs, we will use that $\otimes$ is associative and commutative, and verifies $g\otimes f\geq g'\otimes f$ for all trade-off functions $f$ and $g\geq g'$.
\end{remark}

\THgaussianactheorem*
\begin{proof}
    Note that $d_\D$ is a well-defined metric since the square root of the sum of squared distances is still a distance (i.e., the $\ell_2$-norm).
    
    Now we need to prove for all $D,D'\in\D$ that
    \[
        T(\M(D),\M(D')) \geq G_{d_\D(D,D')}.
    \]
    
    We prove the result by induction over $k$. For $k=1$, the result is trivial. Therefore, fixing $k$, we suppose it is true for $k-1$ and we prove for $k$.
    
    Let $\overM=(\M^*_1\circ f^*_1, \dots, \M^*_{k-1}\circ f^*_{k-1})_\adapt$. By the induction hypothesis, for all $D,D'\in\D$,
    \[
        T(\overM(D),\overM(D')) \geq G_{\myol{d}(D,D')}
    \]
    with $\myol{d}=\sqrt{{d_1^2+\dots+d^2_{k-1}}}$. We can also rewrite $\M$ as a function of $\overM$ and $\M_k$ as
    \[
        \M(D) = (\overM(D),\M_k(\overM(D),f_k(D)))
    \]
    for all $D\in\D$.
    
    We fix $D,D'\in\D$. Since $\M_k(\overs_{k},\cdot)$ is $d_k$-Gprivate for all $\overs_{k}\in\overS_{k}$, we have that 
    \[
        T(\M(D),\M(D')) = T(\overM(D),\overM(D'))\otimes G_{d_k(D,D')}.
    \]
    
    This fact follows from Lemma C.1 in~\cite{dong2019Gaussian} as explained in their proof of Lemma C.3. Since $\overM$ is $\myol{d}$-private, we obtain 
    \[
        T(\M(D),\M(D')) \geq G_{\myol{d}(D,D')}\otimes G_{d_k(D,D')}.  
    \]
    by the properties of $\otimes$ (see \Cref{re:Tensor}). Finally, by Proposition D.1 in~\cite{dong2019Gaussian}, we obtain 
    \[
        G_{\myol{d}(D,D')}\otimes G_{d_k(D,D')} = G_{\sqrt{\myol{d}(D,D')^2+d_k(D,D')^2}}     
    \]
    where 
    \[
        \sqrt{\myol{d}(D,D')^2+d_k(D,D')^2}=\sqrt{\sum^k_{i=1}d_i(D,D')^2}. \qedhere
    \]
\end{proof}

\THgaussianACtheoremcommondomain*
\begin{proof}
    We just need to prove that $\M_i$ are $d_i^{f_i}$-private* and it follows from an analogous proof of \Cref{th:GaussianACTheorem}.

    By hypothesis, $\M_i$ are $d_i$-Gprivate, which means that for all $D,D'\in\D$, $T(\M_i(D),\M_i(D'))\geq G_{d_i(D,D')}$. Using \Cref{re:equalRandomElements}, we obtain that 
    \[
        T(\M_i(D),\M_i(D'))=T(\M_i(\Tilde{D}),\M_i(\tilde{D}'))
    \]
    for all $\tilde{D},\tilde{D}'\in\D$ such that $f(D)=f(\Tilde{D})$ and $f(D')=f(\tilde{D}')$. In conclusion, $T(\M_i(D),\M_i(D'))\geq G_{d_i^f(D,D')}$ with
    \[
        d^f_i(D,D') = \min_{\substack{\tilde{D},\tilde{D}'\in\D \\ f(\tilde{D})=f(D)\\ f(\tilde{D}')=f(D')}} d_i(\tilde{D},\tilde{D}'). \qedhere
    \]
\end{proof}

\THgaussianASC*
\begin{proof}
    Direct from \Cref{th:GaussianACTheorem} by taking $\D_i=\D$ and $f_i=\id$.
\end{proof}

\THgaussianparallel*
\begin{proof}
    By \Cref{th:GaussianACTheorem}, we have that $\M$ is $d_{\D}$-Gprivate for 
    \begin{align*}
        d_{\D}(D,D') =\sqrt{\sum^k_{i=1}\mu_i^2 d^*_{\D_i}(p_i(D),p_i(D'))^2} 
        \leq \max_{i\in[k]}\mu_i\sqrt{\sum^k_{i=1} d^*_{\D_i}(p_i(D),p_i(D'))^2}.
    \end{align*}
    
    Then, we have that
    \begin{align*}
     \sum^k_{i=1}d^*_{\D_i}(p_i(D),p_i(D'))^2\overset{\textrm{(i)}}{\leq}{}& \bigg(\sum^k_{i=1}d^*_{\D_i}(p_i(D),p_i(D'))\bigg)^2\\
       \overset{\textrm{(ii)}}{=}{}& d^*_{\D}\bigg(\bigcup^k_{i=1} p_i(D),\bigcup^k_{i=1} p_i(D')\bigg)^2 \\
       \overset{\textrm{(ii)}}{\leq}{}& d^*_{\D}(D,D')^2,
    \end{align*}
    where 
    \begin{enumerate}[(i)]
        \item comes from the fact that $\sum^k_{i=1}a_{i}^{2}\leq(\sum^k_{i=1}a_{i})^{2}$ for all $a_i\geq0$,
        \item is due to the commutativity of $d^*$ with respect to $p$.
    \end{enumerate}
    
    Since the square root is a monotonically increasing function, we have the result. \qedhere
\end{proof}

\THgaussianAPCvariabledomain*
\begin{proof}
   From \Cref{th:d-gdpToDP}, it is equivalent to see that $\M$ is $\G$ $\mu$-GDP with $\mu=\max_{i\in[k]}\mu_i$, i.e., that for all $\G$-neighboring $D,D'\in\D$,
    \[
        T(\M(D),\M(D')) \geq G_{\mu}
    \]
    
    Applying \Cref{th:GaussianACTheorem}, we obtain that $\M$ is $d$-Gprivate with
    \[
       d_{\D}(D,D')=\sqrt{\sum^k_{i=1}\mu_i^2 d_{\D_i}^{\G}(p_i(D),p_i(D'))^2}.
    \]
    
    Now suppose $D,D'\in\D$ are $\G$-neighboring. By definition of $d^\G_\D$-compatibility, there exist $j\in[k]$ such that $p_i(D)=p_i(D')$ for all $i\neq j$. Consequently, for all $i\neq j$, $d^\G_{\D_i}(p_i(D),p_i(D')) = 0$. Moreover, by preprocessing (\Cref{prop:preprocessing}), we have that $d^\G_{\D_j}(p_j(D),p_j(D'))\leq \Delta p_j d^\G_{\D}(D,D') \leq 1$ since $D\neigh_\G D'$ and $\Delta p_j\leq 1$. Therefore,
    \begin{gather*}
        d_{\D}(D,D') = \sqrt{\sum^k_{i=1} \mu_j^2 d_{\D_j}^{\G}(p_j(D),p_j(D'))^2} = \mu_j d_{\D_j}^{\G}(p_j(D),p_j(D'))\leq \mu_j.
    \end{gather*}

    Consequently, since $\M$ is $d$-Gprivate,
    \[
         T(\M(D),\M(D')) \geq G_{\mu_j}.
    \]
    
    Since $j\in[k]$ depends on the choice of the $\G$-neighboring $D,D'\in\D$, it is sufficient to choose $\mu=\max_{i\in[k]}\mu_i$ to cover all cases. In conclusion, $\M$ is $\G$ $\mu$-GDP.
\end{proof}

\begin{theorem}[Gaussian best bound for disjoint inputs (common domain)]\label{th:GaussianAPCCommonDomain}
    Let $\D$ be a database class and $\G$ a granularity over $\D$. Let $p$ be a $d^\G_\D$-compatible $k$-partitioning function.
    For $i\in[k]$, let $\M_i\colon\overS_{i}\times\D\to\S_i$ be a mechanism such that $\M_k(\overs_{i},\cdot)\colon\D\to\S_i$ satisfies $\mu_i d^\G_{\D}$-Gprivacy and $p_i$-dependency for any $\overs_{i}\in\overS_{i}$.
    Then mechanism $\M=(\M_1,\dots,\M_k)_{\adapt}$ is $\mu d^\G_{\D}$-Gprivate with $\mu=\max_{i\in[k]} \mu_i$.
\end{theorem}
\begin{proof}
    From \Cref{th:d-gdpToDP}, it is equivalent to see that $\M$ is $\G$ $\mu$-GDP with $\mu=\max_{i\in[k]}\mu_i$, i.e., that for all $\G$-neighboring $D,D'\in\D$,
    \[
        T(\M(D),\M(D')) \geq G_{\mu}.
    \]
    
    From \Cref{th:GaussianACTheoremCommonDomain}, we have that $\M$ is $d_{\D}$-Gprivate* with
    \[
        d_{\D}(D,D') = \sqrt{\sum^k_{i=1} \mu_i^2 d_{\D}^{\G,p_i}(D,D')^2}.
    \]
    
    Since $d^\G_\D$-compatible, there exist only one $j\in[k]$ such that $p_j(D)\neq p_j(D')$. Consequently, for all $i\neq j$, $d^{\G,p_i}_\D(D,D')\leq d^\G_\D(D,D) = 0$, since we can select $D$ as both $\tilde{D}$ and $\tilde{D}'$ in the definition (see \Cref{prop:MinimumPrivacy}). Therefore, 
    \begin{gather*}
        d_{\D}(D,D') = \sqrt{\sum^k_{i=1} \mu_i^2 d_{\D}^{\G,p_i}(D,D')^2}
        =\sqrt{\mu_j^2 d_{\D}^{\G,p_j}(D,D')^2+0}=\mu_j d_{\D}^{\G,p_j}(D,D')
        \leq \mu_j d_{\D}^{\G}(D,D')\leq \mu_j,
    \end{gather*}
    where the last inequality comes from the fact that $D\neigh_{\G} D'$. Since $j$ depends on the choice of $D$ and $D'$, it is sufficient to take $\mu=\max_{i\in[k]}\mu_i$ to cover all possible cases.
\end{proof}

\begin{corollary}\label{th:GaussianBoundedParallel}
    Let $p$ be a $k$-partitioning function of \Cref{ex:partition}. For all $i\in[k]$, let $\M_i\colon\D\to\S_i$ be mutually independent bounded $\mu_i$-Gprivacy mechanisms that are $p_i$-dependent.
    Then mechanism $\M=(\M_1,\dots,\M_k)_{\ind}$ with domain $\D$ is bounded $\mu$-GDP with $\mu=\max_{i,j\in[k];\,i\neq j} \sqrt{\mu_i^2+\mu_j^2}$.
\end{corollary}
\begin{proof}
    From \Cref{th:d-gdpToDP}, it is equivalent to see that $\M$ is bounded $\mu$-GDP ($\mu=\max_{i,j\in[k];\,i\neq j} \sqrt{\mu_i^2+\mu_j^2}$), i.e., that for all bounded-neighboring $D,D'\in\D$,
    \[
        T(\M(D),\M(D')) \geq G_{\mu}.
    \]
    
    Applying \Cref{th:GaussianACTheoremCommonDomain}, we obtain that $\M$ is $d_\D$-Gprivate* with
    \[
        d_{\D}(D,D') = \sqrt{\sum^k_{i=1} \mu_i^2 d_{\D}^{\B,p_i}(D,D')^2}.
    \]
    
    Now suppose $D,D'\in\D$ are bounded-neighboring. We know there exists $x\in D$ and $x'\in D'$ such that $D\triangle D'=\{x,x'\}$. Then, we have the following possibilities:
    \begin{enumerate}[(a)]
        \item $x,x'\in\X_j$ for a $j\in[k]$. This implies that $p_i(D)=p_i(D')$ for all $i\neq j$.
        \item $x\in\X_j$ and $x'\in\X_l$ for different $j,l\in[k]$. This implies that $p_i(D)=p_i(D')$ for all $i\neq j,l$.
        \item $x\in\X_j$ for $j\in[k]$ and $x'\not\in\X_l$ for any $l\in[k]$ (or vice-versa). This implies that $p_i(D)=p_i(D')$ for all $i\neq j$.
        \item $x,x'\not\in\X_l$ for any $l\in[k]$. Then $p_i(D)=p_i(D')$ for all $i\in[k]$.
    \end{enumerate} 
    
    In the worst case scenario, there are at most two subindices $j,l\in[k]$ such that $p_i(D)=p_i(D')$ for all $i\neq j,l$. For these subindices, $d^{\B,p_j}_\D(D,D'),d^{\B,p_l}_\D(D,D')\leq d^{\B}_\D(D,D')\leq1$, since $D$ and $D'$ are bounded-neighboring. Therefore, 
    \begin{gather*}
        d_\D(D,D') = \sqrt{\sum^k_{i=1} \mu_i^2 d^{\B,p_i}_\D(D,D')^2} \leq \max_{j,l\in[k];\,j\neq l} \sqrt{\mu^2_j+\mu^2_l} = \mu
    \end{gather*}
    for all bounded-neighboring $D,D'\in\D$. In conclusion, $\M$ is bounded $\mu$-GDP since it is $d_\D$-Gprivate*.
\end{proof}

\subsection*{Proofs of \Cref{sec:reciprocal}: \nameref*{sec:reciprocal}}

\THpostprocessing*
\begin{proof}
    We need to prove that if $\M\colon\D\to\S$ is $d_\D$-private, then $g\circ\M\colon\D\to g(\S)$ is $d_\D$-private for all deterministic functions $g\colon\S\to g(\S)$; and analogously for $(d_\D,\delta_\D)$-privacy, $d^2_\D$-zCprivacy and $d_\D$-Gprivacy. By construction do note that $\Range(g\circ\M)=g(\Range(\M))\eqqcolon g(\S)$. 
    
    For $d_\D$-privacy and $(d_\D,\delta_\D)$-privacy, the proof follows directly from the fact that $\Prob\{\M(D)\in S\}=\Prob\{g(\M(D))\in g(S)\}$ for all measurable $S\subseteq\S$ and $D\in\D$. 
    
    From the properties of the Rényi divergence~\cite{bun2016Concentrated}, we have that 
    \[
        D_\alpha(g(\M(D))\dline g(\M(D')))\leq D_\alpha(\M(D)\dline\M(D'))
    \]
    for all $\alpha\in(1,\infty)$ and $D,D'\in\D$, which proves the result for $d^2_\D$-zCprivacy. 
    
    Finally, from Lemma 2.9 in~\cite{dong2019Gaussian}, we obtain the following inequality: 
    \[
        T(g(\M(D)),g(\M(D')))\geq T(\M(D),\M(D'))\geq d(D,D').
    \] 
    This proves the result for $d_\D$-Gprivacy. 
\end{proof}

\THReciprocalIndependent*
\begin{proof}
    Fix $i\in[k]$. Consider the deterministic projection to the $i$th coordinate $\pi_i$. In this case, $\M_i=\pi_i\circ\M$. Since $\M$ satisfies $\PN$ and $\PN$ is robust to post-processing, $\M_i$ satisfies $\PN$ too. 
\end{proof}

\THReciprocalAdaptive*
\begin{proof}
    Fix $i\in[k]$. Consider the deterministic projection to the $i$th coordinate $\pi_i$. In this case, $\MM_i=\pi_i\circ\M$. Since $\M$ satisfies $\PN$ and $\PN$ is robust to post-processing, $\MM_i$ satisfies $\PN$ too. 
\end{proof}

\end{appendices}

\end{document}